\documentclass[prodmode,acmtoms]{acmsmall} 

\usepackage{amsmath, amssymb}

\usepackage[ruled]{algorithm2e}
\renewcommand{\algorithmcfname}{ALGORITHM}
\SetAlFnt{\small}
\SetAlCapFnt{\small}
\SetAlCapNameFnt{\small}
\SetAlCapHSkip{0pt}
\IncMargin{-\parindent}

\doi{0000001.0000001}

\issn{1234-56789}

\usepackage{algorithmicx,algpseudocode}

\begin{document}

\markboth{B. Reps and T. Weinzierl}
{Complex additive geometric multilevel solvers for Helmholtz equations on spacetrees}

\title{Complex additive geometric multilevel solvers for Helmholtz equations on spacetrees}

\providecommand{\i}{\imath}
\providecommand{\todo}[2]{\marginpar{\footnotesize #1: #2 }}
\newcommand{\degrees}{\ensuremath{^{\circ}}} 
\newcommand{\br}[1]{\begin{color}{red}[#1]\end{color}}


%
%
\newcommand{\EventBeginIteration}{{\tt \footnotesize beginIteration}}
\newcommand{\EventEndIteration}{{\tt \footnotesize  endIteration}}

\newcommand{\EventEnterCell}{{\tt \footnotesize  enterCell}}
\newcommand{\EventLeaveCell}{{\tt \footnotesize leaveCell}}

\newcommand{\EventTouchVertexFirstTime}{{\tt \footnotesize touchVertexFirstTime}}
\newcommand{\EventTouchVertexLastTime}{{\tt \footnotesize touchVertexLastTime}}

\newcommand{\EventMergeWithMaster}{{\tt mergeWithMaster}}
\newcommand{\EventMergeWithNeighbour}{{\tt mergeWithNeighbour}}
\newcommand{\EventMergeWithWorker}{{\tt mergeWithWorker}}
\newcommand{\EventPrepareSendToMaster}{{\tt prepareSendToMaster}}
\newcommand{\EventPrepareSendToWorker}{{\tt prepareSendToWorker}}
\newcommand{\EventCreateCell}{{\tt createCell}}
\newcommand{\EventDestroyCell}{{\tt destroyCell}}
\newcommand{\EventLoadVertex}{{\tt loadVertex}}
\newcommand{\EventStoreVertex}{{\tt storeVertex}}
\newcommand{\EventCreateHangingVertex}{{\tt createHangingVertex}}
\newcommand{\EventDestroyHangingVertex}{{\tt destroyHangingVertex}}

\renewcommand{\vec}{\mathbf}

 \author{
  BRAM REPS
  \affil{University of Antwerp}
  TOBIAS WEINZIERL
  \affil{Durham University}
 }

 \begin{abstract}
We introduce a family of implementations of low order, additive, geometric
multilevel solvers for systems of Helmholtz equations arising from
Schr\"odinger equations.
Both grid spacing and arithmetics may comprise complex numbers and we thus can
apply complex scaling to the indefinite Helmholtz operator. 
Our implementations are based upon the notion of a spacetree and work
exclusively with a finite number of precomputed local element matrices.
They are globally matrix-free.

Combining various relaxation factors with two grid transfer
operators allows us to switch from additive multigrid over a hierarchical
basis method into a Bramble-Pasciak-Xu (BPX)-type solver, with several multiscale smoothing variants within one code
base.
Pipelining allows us to realise full approximation storage (FAS) within
the additive environment where, amortised, each grid vertex carrying degrees of
freedom is read/written only once per iteration.
The codes realise a single-touch policy.
Among the features facilitated by matrix-free FAS is arbitrary dynamic mesh
refinement (AMR) for all solver variants.
AMR as enabler for full multigrid (FMG) cycling---the grid unfolds throughout
the computation---allows us to reduce the cost per unknown per order of accuracy.

The present paper primary contributes towards software realisation and design questions.
Our experiments show that the consolidation of single-touch FAS, dynamic AMR and vectorisation-friendly, complex scaled, matrix-free FMG cycles delivers
a mature implementation blueprint for solvers of Helmholtz equations in general.
For this blueprint, we put particular emphasis on
a strict implementation formalism as well as some implementation correctness
proofs. 
\end{abstract}

 %
%



%
%


 \acmformat{
  Bram Reps, Tobias Weinzierl, 2015. Complex additive geometric
  multilevel solvers for Helmholtz equations on spacetrees.
 }

 \begin{bottomstuff}
%
  Author's addresses: 
  B. Reps, Department of Mathematics and Computer Science, University of
  Antwerp, 2020 Antwerp, Belgium;
  T. Weinzierl, School of Engineering and Computing Sciences, Durham University,
  DH1 3LE Durham, United Kingdom
 \end{bottomstuff}

 \maketitle

  \section{Introduction}
\label{section:introduction}

%
%
The present paper's efforts are driven by an interest in 
the dynamics of $p$ quantum particles, originally described by the time-dependent Schr\"odinger wave equation
\cite{Schroedinger:26:QMTheory}.
The wave equation's operator is the sum of the kinetic energies of the
individual particles and their potential energy
determined by their nature and interaction.
The \emph{Hamiltonian} operator drives the complex-valued time derivative.
Numerical methods solving Schr\"odinger equations are of great
interest for the simulation and prediction of the reaction rates of fundamental
processes in few-body physics and chemistry
\cite{Faddeev:93:QMScattering,NRC:07:PlasmaScience,Allen:91:GasDischarge,Smirnov:15:GasDischarge,Vanroose:05:Science}.
The breakup of the $H_2$-molecule with $p=2$ particles for example
is solved in \cite{Vanroose:05:Science}.
A key ingredient there is the expansion of the time-dependent Schr\"odinger
equation into a time-independent Schr\"odinger equation being a $d=3p$-dimensional Helmholtz
problem.
This ansatz allows a post-processing step to determine the \emph{far
field map} (FFM) which yields the probability distribution of particles
escaping.
It is written down as function of the angle from the point of interaction and is
calculated by integrating over stationary Helmholtz solutions (Figure \ref{fig::twofixedparticles}).
Since the inversion of the arising matrices becomes unfeasible for growing $p$,
we require robust and fast iterative solvers.

%
%
To achieve this, the present paper follows the aforementioned work
\cite{Vanroose:05:Science}.
It uses a partial wave expansion to 
decompose the time-independent $d=3p$-dimensional system further into a cascade
of $p$-dimensional Helmholtz problems. In essence, this tackles the problem with a
transformed basis, measuring distances between free and stationary particles.
They are referred to as \emph{channels}. 
The base expansion is truncated, i.e.~we focus only on the dominant channels, and we end up with an iterative scheme where a set of uncoupled Helmholtz problems has to be solved
within the iterative loop as preconditioner.
This yields perfectly parallel (embarrassingly parallel \cite{Foster:95:ParallelPrograms}) channel solves. 
Further, we rely on the fact that FFM integrals can be calculated on a complex contour
instead of on the original real domain \cite{Cools:14:FFM}. 
This implies that we may solve Helmholtz equations which are rotated in the
complex domain. 
These equations are better posed while still yielding high quality FFM
results.
A robust, fast solver thus has to perform particularly good on the lowest level of parallelisation, i.e.~exploit
vector units, where the perfect parallelism does not pay off
directly.

%
%
We propose to realise a complex-valued, matrix-free, additive multilevel
solver for the present challenge where combinations of well-suited relaxation
parameters and grid transfer operators allow us to apply an additive multigrid
solver, a hierarchical basis approach or a Bramble-Pasciak-Xu (BPX)-type
solver \cite{Bramble:90:BPX}. 
Within one code base, we may choose a stable solver depending on the equations' characteristics.
It combines the idea of additive multigrid
\cite{Bastian:98:AdditiveVsMultiplicativeMG} with full
approximation storage (FAS) based upon a hierarchical generating system \cite{Griebel:90:HTMM} which
resembles the MLAT idea \cite{Brandt:73:MLAT,Brandt:77:MLAT}, and it realises
all data structures within a $p$-dimensional spacetree traversed by a
depth-first search automaton \cite{Weinzierl:2009:Diss,Weinzierl:11:Peano}. 
Our Helmholtz solver supports dynamically adaptive meshes resolving localised
wave characteristics.
Furthermore, its in-situ meshing allows us to unfold 
the grid on-the-fly similar to FMG cycles in the multiplicative context.
In our preconditioning environment, multiple problem parameter 
choices finally are fused into one grid sweep as long as it increases the
arithmetic intensity.

%
%
%
%
%
%
%

%
%
The novel algorithmic contributions of the present paper span algorithms,
application-specific experience and a challenging application:
First, we integrate various sophisticated multigrid techniques
concisely into one code base such that we can offer 
the additive multilevel solvers with a one-touch policy.
Each vertex is, on average, loaded into the caches only once per iteration and
resides inside the caches only briefly.
All ingredients are vertically, i.e.~between the grid levels, and
horizontally, i.e.~in the grid, integrated.
Similar techniques have been proposed for multilevel solvers 
\cite{Adams:15:Segmental,Mehl:06:MG,Ghysels:12:MGArithmIntensity,Ghysels:15:GeometricMultigridPerformance}
or Krylov solvers \cite{Chronopoulos:89:sStep,Hoemmen:10:CAKrylov,Ghysels:13:HideLatency,Ghysels:14:HideLatency},
but, to the best of our knowledge, no other approach offers a solution representation on all levels plus single touch.
Multilevel solution representations simplify the handling of hanging nodes,
non-linear problems and scale-dependent discretisations \cite{Cools:14:LevelDependent}.
Second, we show that this realisation embeds into a
depth-first traversal of a tree spanning the cascade of grids embedded into
each other.
Such trees---octrees, quadtrees and variations of those---are popular in many
application areas.
They allow us to realise arbitrary dynamic refinement and coarsening in
combination with in-situ mesh generation that seamlessly integrate into the
stream-like data processing.
There is no significant setup cost for the meshing, while the memory footprint
is minimalist as we work matrix-freely
and encode the underlying grid structure with only few bits per cell
\cite{Weinzierl:11:Peano}.
Simple data flow analysis reveals that such a merger of FAS with the ideas
of additive multigrid and spacetrees is not straightforward if data is not processed multiple times.
Third, we show that the cheap dynamic adaptivity allows us to tailor the grid to the
solution characteristics. 
Furthermore, if the solver runs a cascade of additive
cycles over finer and finer grids, the resulting scheme mirrors FMG where coarse
solves act as initial guess for finer grids.
It is able to reduce, for benchmark problems, the residual by one order of magnitude every 1.5 traversals.
This is remarkable given that we use merely a Jacobi smoother and geometric
inter-grid transfer operators.
Fourth, we report on reasonable MFlop rates and vectorisation efficiency as we
fuse the solution of multiple preconditioning problems into one adaptive grid;
despite the fact that we employ a low order scheme which is notoriously
bandwidth bound and do not rely on sophisticated smoother optimisation techniques
\cite{Kowarschik:00:CacheAwareMG}.
This renders the algorithmic mindset well-suited for upcoming machine
generations that are expected to obtain a significant part of their capability
from vectorisation's extreme concurrency in combination with constrained memory
\cite{Dongarra:14:ApplMathExascaleComputing}.
Fifth, we demonstrate how complex scaling and various choices of
relaxation and very few operators allow a user to obtain a set of solvers that can be tailored
to many problems.
Notably, we can robustly solve four-dimensional Helmholtz problems which is
a significant improvement over previous work.
Finally, all algorithmic steps are presented in a compact form and all
ingredients come along with correctness proofs.
While linear algebra packages supporting complex arithmetics per se are rare,
a rigorous formal description enables reprogramming and reuse for different
applications.

%
%
We identify four major limitations of the present work.
First, we do not offer a strategy to tackle the curse of dimensionality
\cite{Bellman:61:CurseOfDimensionality} spelled through $p$ rising.
Though we show that the FMG-type cycles reduce the cost per unknown per accuracy
by magnitudes, such a reduction of cost, even in combination with the
vectorisation and memory access results, is far below what is required to
tackle large $p$.
In practice, we are still bound to $p \leq 4$ though the implementation would
support arbitrary big $p$.
Second, we do not pick up the discussion on well-suited smoothers for the
present problems.
We show that our solvers achieve robustness due to the complex rotations.
However, their efficiency deteriorates.
We emphasise that the efficient nature of the present implementation
patterns makes us hope that they
can be used as starting point to realise more competitive smoothers as proposed
in 
\cite{Chen:12:DispersionMinimizing,Ernst:11:HelmholtzIterativeMethods,Ghysels:12:MGArithmIntensity,Ghysels:15:GeometricMultigridPerformance,Stolk:15:DispersionMinimizing},
e.g.
Yet, this is future work.
Third, we do not realise problem-dependent grid transfer or coarse grid
operators.
Such operators are mandatory to tackle problems with spatially
varying PDE properties as they occur for matching boundary conditions, e.g.,
within the multigrid setting.
For the present case studies, our solvers' efficiency here suffers.
We refer to promising tests with BoxMG within the spacetree paradigm
\cite{Weinzierl:13:Hybrid,Yavneh:12:Nonsymmetric} for pathways towards future
work.
Finally, any multiplicative considerations are out of scope here.
All shortcomings highlight that the present paper primary contributes  
towards software realisation and design questions.
We also focus on single node performance as our algorithmic framework is
perfectly parallel.
This neither implies that the presented approach can not be tuned further with
respect to parallelism nor do we address
related challenges such as proper load balancing.
Notably, techniques such as segmental refinement \cite{Adams:15:Segmental}
that spatially decompose fine grid solves could help to increase the concurrency
while preserving the present work's vertical and horizontal integration as well
as the fusion of parameter spaces.

%
%

%
%
The remainder is organised as follows: 
We detail the physical context and solver framework in
Section \ref{section:model-problem} before we introduce our notion of spacetrees
yielding the multiscale grid
(Section \ref{section:spacetrees}).
Starting from a recapitulation of standard additive multigrid, we introduce our
particular single-touch multilevel ingredients in Section
\ref{section:multigrid}.
They are fused into one additive scheme with different flavours in Section
\ref{section:faddmg}.
We next give some correctness proofs for the realisation (Section
\ref{section:properties}).
Some numerical results in Section \ref{section:results} precede a brief
conclusion and an outlook.

  \section{Application context}
\label{section:model-problem}

\begin{figure}
  \begin{center} 
    \includegraphics[width=0.48\textwidth]{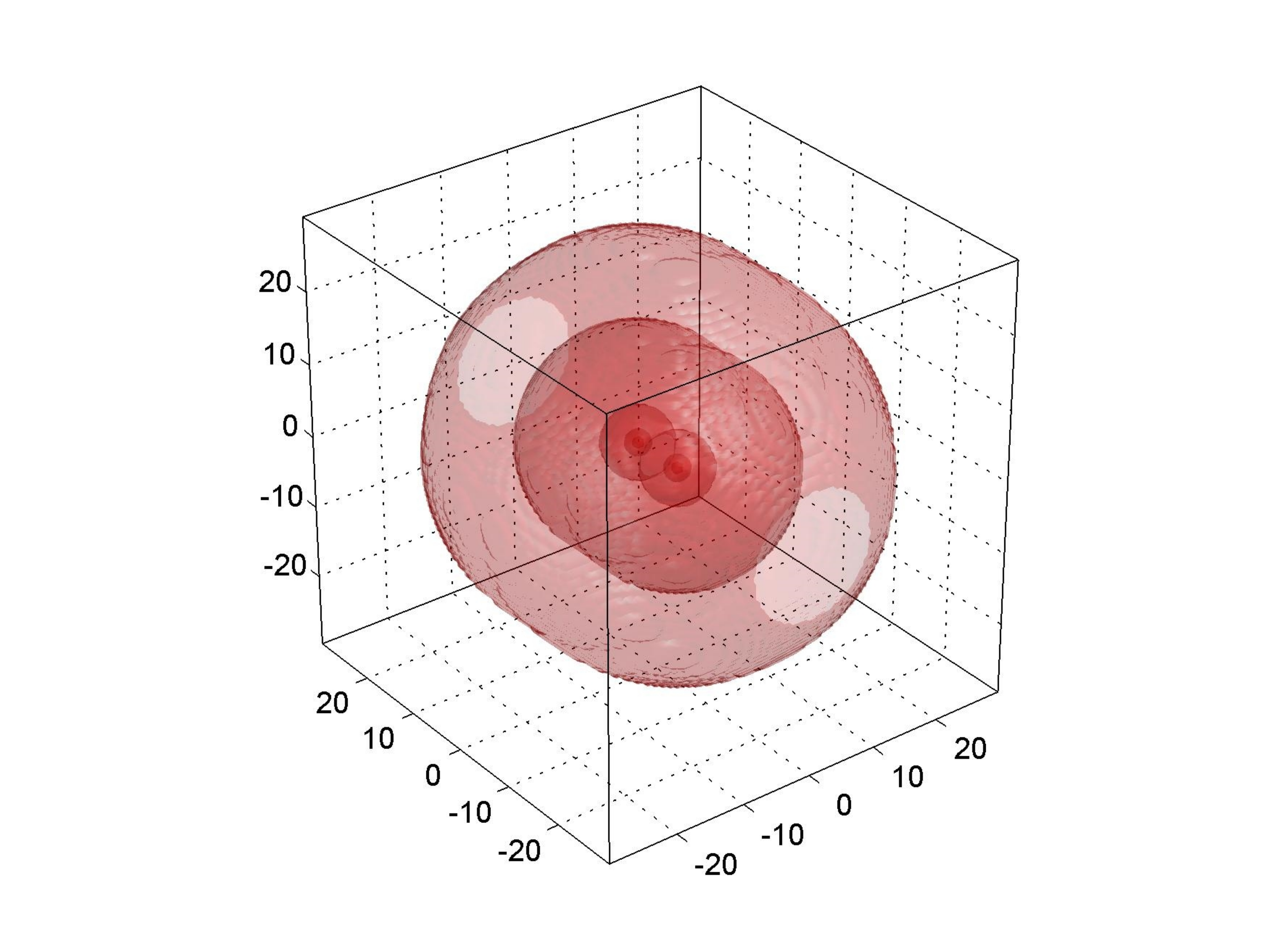}
    \includegraphics[width=0.45\textwidth]{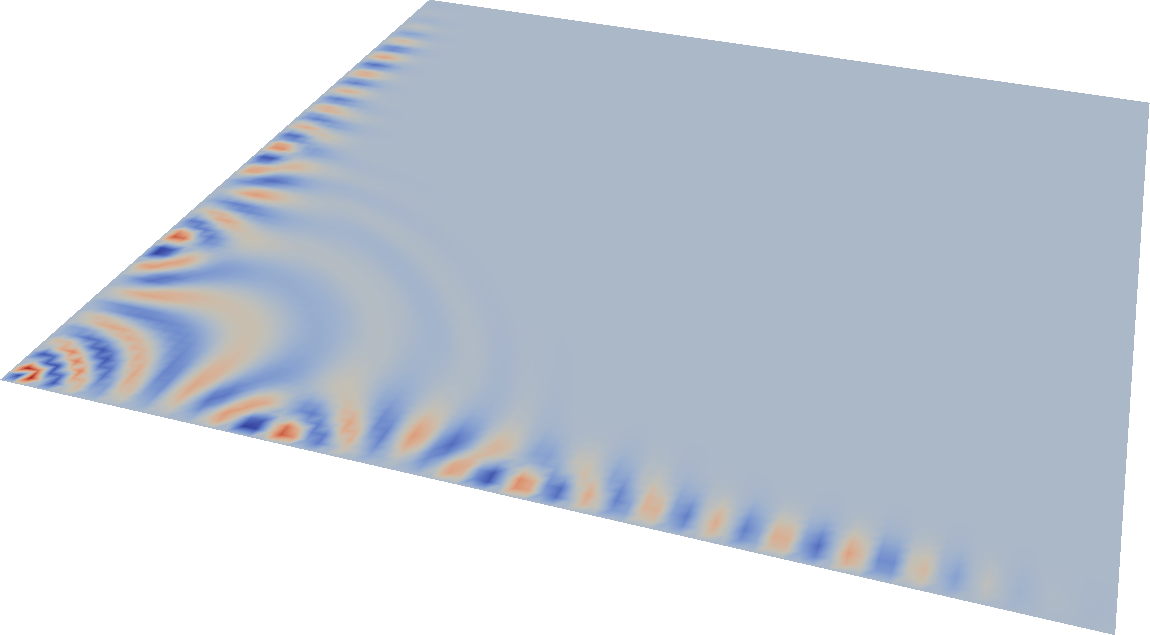}
  \end{center}
  \vspace{-0.4cm}
  \caption{
    Left: Potential field defined by two fixed (heavy)
    particles. A three-dimensional Helmholtz equation \eqref{eq::3pHelmholtz}
    describes the behaviour, i.e.~probability distribution, of a single free particle
    within this field.
    Applications are interested in its far field map (FFM) which is a volume
    integral over Helmholtz solutions within a sphere enclosing the
    setup---typically for $p>1$.
    Right: Solution of one channel's Helmholtz problem of type \eqref{eq::pHelmholtz} for a $p=2$ setup.
    The bottom and left axis encode the distance from the $p=2$ centers of mass.
    The domain expands infinitely into the right and top
    direction.
  }
  \label{fig::twofixedparticles}
\end{figure}

A time-dependent Schr\"odinger equation for $p$ free particles can be solved
by projecting the initial state ($t=0$) onto the Hamiltonian's eigenstates. 
Each quantum eigenstate $\Psi$ parameterised in spherical coordinates $r_j$
around the particles is factorisable as
\begin{equation*}
\Psi(\vec{r}_1,\ldots,\vec{r}_p,t) \equiv e^{-\i Et}\psi^{3p}(\vec{r}_1,\ldots,\vec{r}_p),
\end{equation*}
where $E$ is the eigenvalue, that is the eigenstate's total energy.
As the probability distribution
$|\Psi(\vec{r}_1,\ldots,\vec{r}_p,t)|^2=|\psi^{3p}(\vec{r}_1,\ldots,\vec{r}_p)|^2$
is constant in time \cite{Cohen:77:QM}, these modes are stationary states.
Substituting a stationary state into the time-dependent Schr\"odinger
equations yields a $3p$-dimensional time-independent Schr\"odinger equation,
that is a Helmholtz equation, of the form
\begin{eqnarray}
H_{3p} \psi^{3p}(\vec{r}_1,\ldots,\vec{r}_p) &\equiv& \left[ -\Delta_{3p} - \phi^{3p}(\vec{r}_1,\ldots,\vec{r}_p) \right]  \psi^{3p}(\vec{r}_1,\ldots,\vec{r}_p) \nonumber \\
 &=& \chi^{3p}(\vec{r}_1,\ldots,\vec{r}_p), 
 \label{eq::3pHelmholtz}
\end{eqnarray}

\noindent
where $\phi^{3p}$ and $\chi^{3p}$ depend on the setup's
configuration comprising also the impact of additional fixed (heavy) particles
as well as the free particles' properties.
Usually solely long-term steady state solutions of the particle quantum
system can be measured, so \eqref{eq::3pHelmholtz} needs only be solved once
for the known total energy $E$ of the system.

\begin{figure}\label{fig::context_workflow}
  \begin{center} 
  \includegraphics[width=0.8\textwidth]{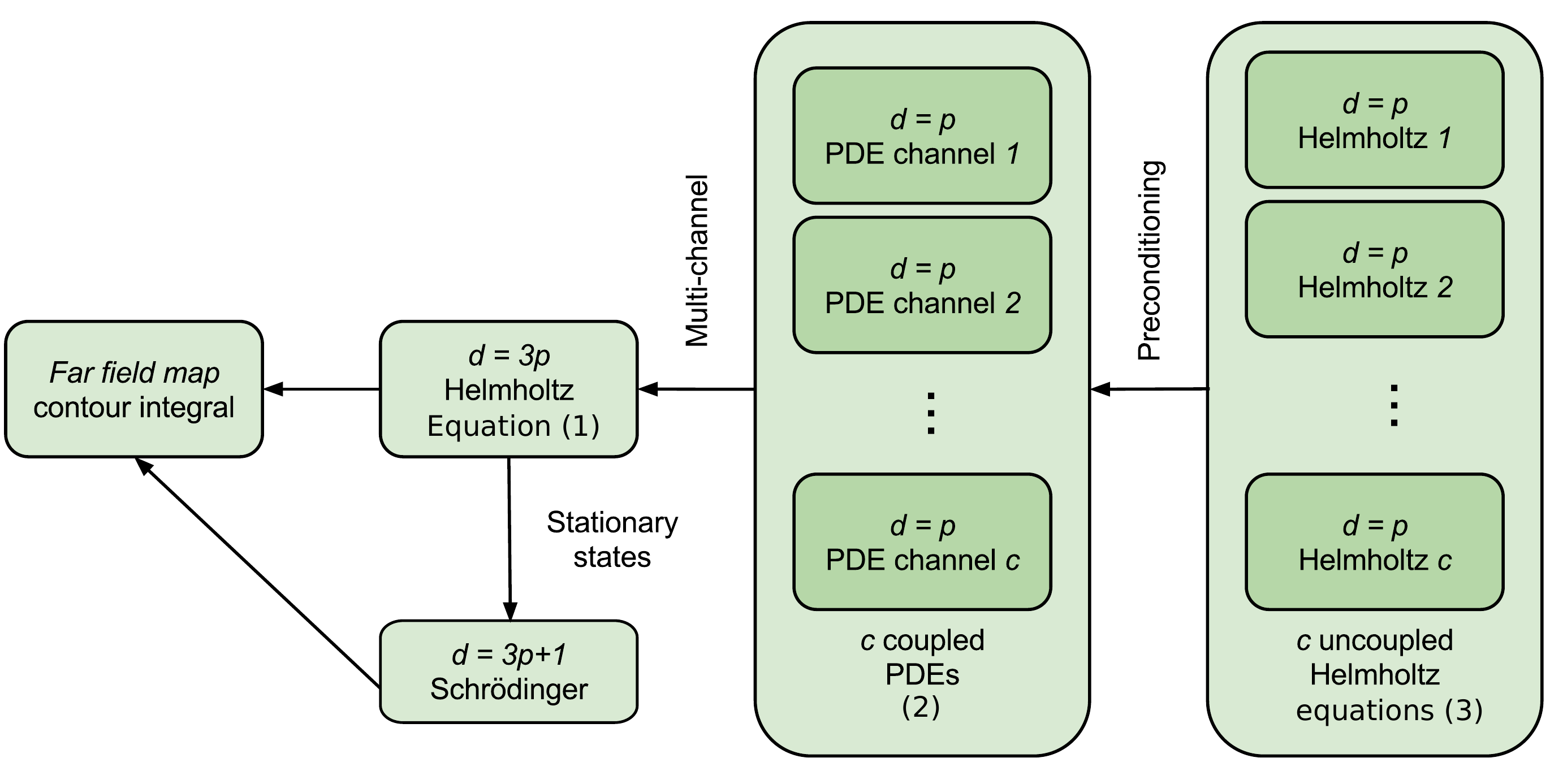}
  \end{center}
  \vspace{-0.4cm}
  \caption{
    Visualisation of the overall workflow from \cite{Vanroose:05:Science} to
    determine the far field map of a quantum mechanical scattering problem.
    The workflow circumnavigates the Schr\"odinger equation where the $+1$
    dimension represents the time. We aim at tackling the set of $c$ Helmholtz problems of dimension $d=p$ for $p>2$.
  }
  \label{fig::workflow}
\end{figure}

%
Although the time-dependence is removed from the governing equation, the dimensionality still grows large with the number of particles. Following \cite{Baertschy:01:ThreeBodySparseLinAlgebra},
the multi-channel approach \cite{Vanroose:05:Science} expresses 
the $3p$-dimensional solution of \eqref{eq::3pHelmholtz} in terms of 
partial waves.
It rewrites $\psi^{3p}$ as a sum of projections onto partial waves ({\em
channels}) in decreasing order of magnitude and truncates this sum after
$c\in\mathbb{N}_0$ terms. 
This involves a transformation of $\vec{r}_j$ into a radial distance and solid
angle.
It yields $c$ coupled PDEs
\begin{equation}\label{eq::coupledHelmholtz}
\begin{pmatrix}
  H_{11}   	& A_{12}	& \hdots	& A_{1c} \\
  A_{21}   	& H_{22}	&    	& A_{2c} \\
  \vdots   	&      	& \ddots	& \vdots \\
  A_{c1}		& A_{c2}	& \hdots	& H_{cc}
 \end{pmatrix}
 \begin{pmatrix}
 \psi^{mc}_1 \\
 \psi^{mc}_2 \\
 \vdots \\
 \psi^{mc}_c
 \end{pmatrix}
 =
 \begin{pmatrix}
 \chi^{mc}_1 \\
 \chi^{mc}_2 \\
 \vdots \\
 \chi^{mc}_c
 \end{pmatrix}
\end{equation}

\noindent
with $p$-dimensional Helmholtz operators $H_{ii}:\mathcal{C}^2(\mathbb{R})\rightarrow \mathcal{C}^2(\mathbb{R})$ on the diagonal. 
Off-diagonal operators $A_{ij}:\mathcal{C}^2(\mathbb{R})\rightarrow
\mathcal{C}^2(\mathbb{R})$ contain potential terms and couple the channels.
The closer to the diagonal, the stronger the coupling.
We split 
\begin{equation*}
A\equiv
 \begin{pmatrix}
  H_{11}   	& 		&  		&  \\
     		& H_{22}	&    	&  \\
   		   	&  		& \ddots	&  \\
  			& 		& 		& H_{cc}
 \end{pmatrix}
 +
 \begin{pmatrix}
  	\mathbf{0}	   	& A_{12}	& \hdots	& A_{1c} \\
  A_{21}   	& \mathbf{0}		&    	& A_{2c} \\
  \vdots   	&      	& \ddots	& \vdots \\
  A_{c1}		& A_{c2}	& \hdots	& 	\mathbf{0}
 \end{pmatrix},
\end{equation*}

\noindent
bring the non-diagonal blocks to the right-hand side and end up with an
iterative scheme on the block level.
Each block row yields a problem of the form
\begin{equation}\label{eq::pHelmholtz}
\left[ -\Delta_{p} - \phi^{p}(\rho_1,\ldots,\rho_p) \right] 
\psi^{p}(\rho_1,\ldots,\rho_p) = \chi^{p}(\rho_1,\ldots,\rho_p).
\end{equation}

\noindent
It is solved on the unit hypercube $\rho \in (0,1)^p$ as finite subregion of
$(0,\infty)^p$, while the modified right-hand side $\chi^{p}$ anticipates the
coupling operators. 
The technical details of this transformation are given in
\cite{Zubair:12:channels}.
Whenever we drop the $p$ superscripts from here on, the symbols are 
generic for any of the channel PDEs.

Each of these equations has to be solved efficiently.
Each solve acts as a preconditioning step within the overall algorithm.
While the spectrum of $\Delta_{p}$ per Helmholtz operator on the diagonal
retains large condition numbers, all system matrices are sparse.
A concurrent solve of such channels in a Jacobi-type fashion is perfectly
parallel and thus not studied further here.
We note that a $p=2$-dimensional Helmholtz problem is solved successfully with direct methods
for the blocks in \cite{Vanroose:05:Science} using a
parallel computer.
For $p>2$, direct solves however are not feasible anymore.

For $\phi^{p}(\rho_1,\ldots,\rho_p)>0$, the Helmholtz operator can be indefinite, which disturbs the convergence of standard iterative methods \cite{Ernst:11:HelmholtzIterativeMethods}. 
Among many other publications on the subject, the field of \emph{shifted Laplacian preconditioning} has greatly inspired the solvers in the current paper. 
The first preconditioners of this kind were the Laplacian and the positively
shifted Laplacian introduced in \cite{Bayliss:83:LaplacePrecon}, later
generalised to complex-valued shifts \cite{Erlangga:04:CSL,Erlangga:06:CSL}. 
Alternative preconditioners and solution methods are derived from 
frequency shift time integration \cite{Meerbergen:09:FreqShift}, 
moving perfectly matched layers \cite{Engquist:11:SweepingPrecon}, 
a transformation of the Helmholtz equation to a reaction-advection-diffusion problem \cite{Haber:11:RADPrecon}, 
separation of variables \cite{Plessix:03:SOVPrecon}, 
the wave-ray approach \cite{Brandt:97:WaveRay}, 
Krylov subspace methods as smoother substitute \cite{Elman:01:KrylovSmoothing}, 
or algebraic multilevel methods
\cite{Bollhoefer:09:AMGPrecon,Tsuji:15:AMGPrecon}.
This list is not comprehensive.

We use the \emph{Complex Scaled Grid (CSG)}
operator \cite{Reps:10:CSG}.
It maps \eqref{eq::pHelmholtz} onto a complex scaled or rotated domain,
i.e.~makes $\rho_j \in (0,e^{\i \theta})$.
$\theta$ denotes this rotation in the complex plane. 
We thus solve closely related Helmholtz equations that are better conditioned and still merge into a good block diagonal preconditioner for \eqref{eq::coupledHelmholtz}.
Near to open domain boundaries, complex-valued scaling introduces absorbing
boundary layers that we use in combination with zero-Dirichlet conditions.
This complex scaling is chosen independent of parameter $\omega$ in the smoother.
In fact, technically, CSG broken down to a grid is equivalent to treating the
entire domain as an absorbing boundary layer.

Because of the complex domain rotation, the eigenvalues are rotated in the
complex plane; away from the origin.
Standard multigrid methods thus can be applied, whereas multigrid on an
unmodified equation ($\theta \mapsto 0$) fails. 
The approach is inspired by the \emph{complex shifted Laplacian (CSL)}
where complex damping is introduced in the Helmholtz shift, i.e.~$\phi
\to (1+\alpha\i)\phi$ with $0<\alpha<1$ \cite{Erlangga:04:CSL,Erlangga:06:CSL}. 
We refer to \cite{Magolu:00:CSL,vanGijzen:07:CSL,Osei:10:CSL,Cools:13:LFAofCSL} and the extensive literature on the CSL operators for a study of the appropriate choice for parameter $\theta$, a discussion that is beyond the scope of the present work.
Both CSG acting as preconditioner and the channel decomposition
preserve the solution characteristics of the FFM (Figure \ref{fig::workflow}) which is the final
quantity of interest, and we conclude that the channel decomposition's reduction of dimensionality of
the subproblems starts to pay off for $p\geq2$.

  \section{Spacetrees}
\label{section:spacetrees}

%
%

\begin{algorithm}[htb]
  \caption{Textbook additive multigrid. 
    An iteration is triggered by 
    \textsc{add}($\ell _{max}$), i.e.~runs from finest to coarsest grid. 
    See Remark \ref{remark:exact-coarse-grid-solve} on exact $\ell _{min}$
    solves.
   \label{algo::add}
  }
 \SetAlgoNoLine
    \begin{algorithmic}[1]
      \Function{add}{$\ell $} 
         \State $r_\ell \leftarrow \chi_\ell - H_\ell u_\ell$ 
           \Comment $p$-linear shape functions determine $H$.
         \State $u_\ell \leftarrow u_\ell + \omega_{S}
           S(u_\ell,\chi_\ell)$
           \Comment Smoother $S$ with damping $\omega_{S} \in (0,1)$.
         \State
         \If{$\ell > \ell _{min}$}{
           \State \phantom{xx} $u_{\ell -1} \leftarrow 0$
           \State \phantom{xx} $\chi_{\ell -1} \leftarrow Rr_\ell$
             \Comment $p$-linear geometric restriction.
           \State \phantom{xx} \Call{add}{$\ell -1$}
           \State \phantom{xx} $u_\ell \leftarrow u_\ell + \omega_{cg}
           Pu_{\ell-1}$ \Comment Coarse grid damping $\omega_{cg} \in (0,1)$.
           \State
             \Comment $p$-linear prolongation $P$.
         }
      \EndFunction
  \end{algorithmic}
\end{algorithm} 


Our implementation of multilevel solvers such as additive
multigrid (Algorithm \ref{algo::add}) relies on a finite element formulation where the geometric elements are hypercubes over the complex domain. 
Their dimension is $3p$ for Helmholtz problems of type \eqref{eq::3pHelmholtz} and $p$ for Helmholtz problems of type \eqref{eq::pHelmholtz} in the channel approach. 
Except for some illustrations from the application domain, we focus on numerical results for the latter and therefore use $p$ as dimension from hereon.
For $p=2$, we start with the unit square suitably scaled by $e^{\i\theta} \in
\mathbb{C}$.
This square is our coarsest grid $\Omega _{h,0}$ holding one cell and
$2^p$ vertices.
It coincides with the computational domain.
Let its grid entities have level $\ell = 0$.
We next split the cube equidistantly into three parts along each coordinate axis and
end up with $3^p$ new squares. 
They describe a grid $\Omega _{h,1}$ and belong to level $\ell = 1$.
Our choice of three-partitioning stems from the fact that we use the software
Peano \cite{Software:Peano} for our realisation.
All algorithmic ideas work for bi-partitioning as well.

\begin{figure}
  \includegraphics[width=0.45\textwidth]{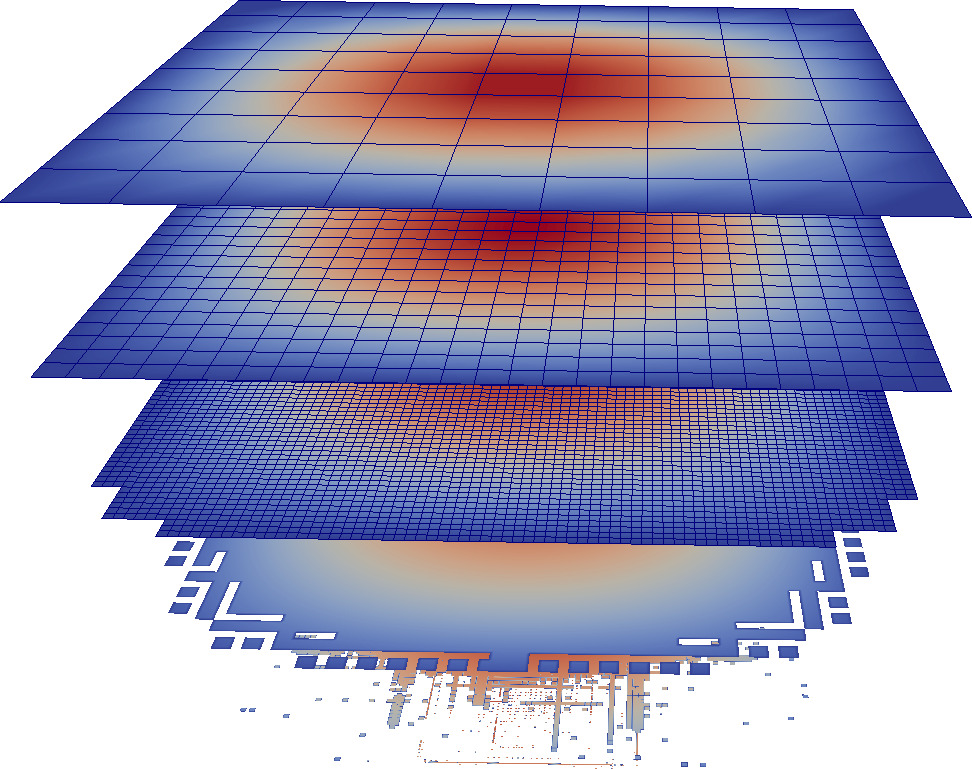}
  \includegraphics[width=0.45\textwidth]{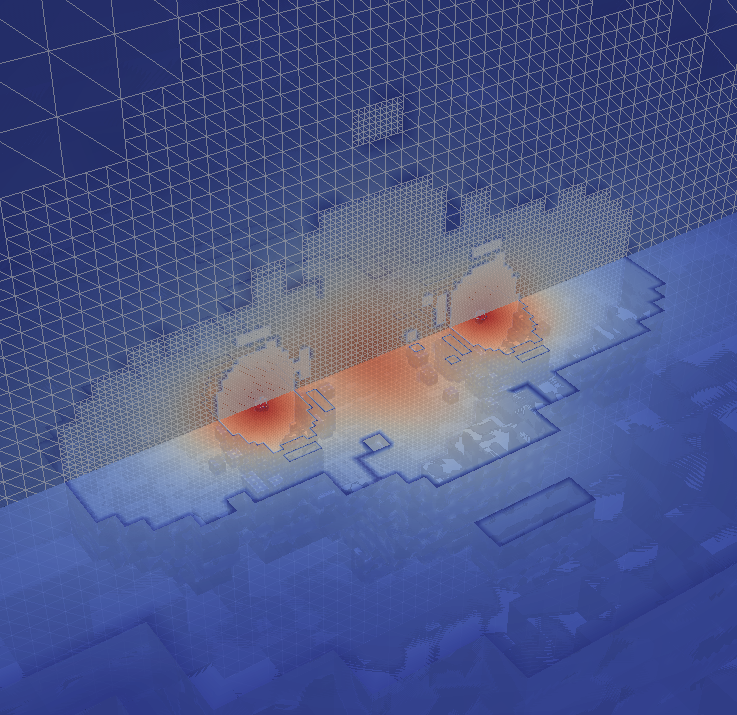}
  \caption{
    Left: The spacetree yields a cascade of regular grids. 
    Two dimensional setting with grid levels from top to down. While the
    union forms an adaptive Cartesian grid and while the grids are
    geometrically embedded into each other, a grid on a single level might
    be ragged (finest level).
    Multiple vertices belonging to different levels coincide spatially, all the
    vertices on the two coarsest levels left are {\em refined}.
    Right: Illustration of heavy hydrogen $^2$\hspace{-2pt}$H$ with one free electron and a fixed proton and neutron.
    Direct solution of \eqref{eq::3pHelmholtz}.
    The adaptivity yields grids that range over ten levels in the simulation
    domain. 
  }
  \label{spacetrees:cascade-of-grids}
\end{figure}

For each of the $3^p$ squares of level $\ell=1$ we decide independently whether
we refine them further.
As we continue recursively, we end up with a cascade of regular grids that
might be disconnected (Figure \ref{spacetrees:cascade-of-grids} left).
The extension of this construction to any $p>2$ is
straightforward (e.g.\ Figure \ref{spacetrees:cascade-of-grids} right).
Our overall scheme describes a spacetree
\cite{Weinzierl:2009:Diss,Weinzierl:11:Peano} and is an extension of the
classic octree and quadtree idea to three-partitioning in combination with
arbitrary spatial dimension $d$.
The spacetree exhibits the following properties that are important to our
solvers:

\begin{enumerate}
  \item It yields a set of grids where each grid of level
  $\ell$ is a refinement of the grid of level $\ell-1$, i.e.~the grids are
  strictly embedded into each other.
  \begin{equation}
    \Omega _{h,0} \subset \Omega _{h,1} \subset \ldots \subset \Omega _{h,\ell
    -1} \subset \Omega _{h,\ell }. 
    \label{eq:composition-of-grids}
  \end{equation}
  \item The union of all grids yields an adaptive Cartesian grid
  \begin{equation}
    \Omega _h = \bigcup _\ell \Omega _{h,\ell } .
    \label{eq:union-of-grids}
  \end{equation}
  \item A vertex is unique due to its spatial
  position plus its level, i.e.~multiple vertices may coincide on the same position
  $x\in\Omega$.
  \item The individual grids are regular Cartesian grids aligned to each other.
  But they can be ragged as not all cells of a level $\ell$ have to be refined
  to obtain the grid at level $\ell+1$. $\Omega _h$ comprises hanging nodes.
\end{enumerate}

\noindent
The equivalence of the cascade of adaptive Cartesian grids with a
spacetree is well-known (\cite{Weinzierl:2009:Diss}
or \cite{Bader:13:SFCs} and references therein).
It motivates our inverse use of the term level with respect to standard 
multigrid literature.
In this context, we also emphasise that we prefer the term erase as counterpart
to refinement, as coarsening has already a semantics in the multigrid context.

Let a cell $a \in \Omega _{h,\ell}$ be a parent of $b \in \Omega _{h,\ell+1}$,
if $b$ is constructed from $a$ due to one refinement step.
This parent-child relation introduces a partial order on the set of all grid
cells of all grid levels. 
It defines the spacetree.
Given a spacetree, any tree traversal is equivalent to a multiscale element-wise
grid traversal.
Particular advantageous is the combination of space-filling curves (we use the Peano
curve here) \cite{Bader:13:SFCs}, adaptive Cartesian multiscale grids and a
depth-first traversal, as it yields memory-efficient codes: Basically, the tree traversal can be mapped onto a depth-first pushback
automaton where a child is never visited prior to its parent.
As soon as this automaton encounters an unrefined spacetree cell, a leaf, it
backtracks the tree and continues to descend within another subtree.
Two bits per vertex then are sufficient to encode both adjacency and dynamic
adaptivity information \cite{Weinzierl:11:Peano}---though we typically use a
whole byte to make programming easier---while the whole data structure 
is linearised on few stacks or streams.
We work with a linearised octree \cite{Sundar:08:BalancedOctrees}.
All presented solver ingredients fit also to other tree traversals.
Notably, breadth-first of parallel traversals \cite{Weinzierl:15:Peano} do work.
We solely require two properties to hold: Parents have to be traversed prior to
their children, and any solver has to have the opportunity to plug into both
the steps from level $\ell $ to level $\ell +1$ and the other way round.
The former requirement allows us to realise arbitrary dynamic adaptivity; the
tree may unfold throughout the steps down.
The latter requirement allows us to distribute algorithmic ingredients both
among the unfolding of the traversal and its backtracking.
The depth of the backtracking is limited by the maximum tree depth.
It is small.


We close our spacetree discussion with a few technical terms. 
A {\em hanging vertex} is a vertex with less than $2^p$ adjacent spacetree cells
on the same level.
A {\em refined vertex} is a vertex where all $2^p$ adjacent spacetree cells
on the same level have children (Figure \ref{spacetrees:cascade-of-grids}).
A  non-hanging vertex is a {\em fine grid vertex} if no other non-hanging vertex
at the same spatial position with higher level does exist.
Let $\mathbb{V}$ be the set of non-hanging vertices in the grid. 
A fine grid cell is an unrefined spacetree cell.
The finest level of the spacetree from here on shall be $\ell _{max}$.
As adaptive multigrid
solves often do not coarsen the problem completely, we rely formally on a
coarsest compute level $\ell _{min}\geq 1$, as all vertices on level 0 coincide
with the domain boundary and, for Dirichlet boundary conditions, do not carry
unknowns.

  \section{Multigrid realisation}
\label{section:multigrid}

%
%
Our work
is based on a Ritz-Galerkin finite element formulation of \eqref{eq::pHelmholtz}
with $p$-linear shape functions and 
a nodal unknown association.
Better-suited, problem-tailored discretisations such as dispersion minimising
schemes \cite{Chen:12:DispersionMinimizing,Stolk:15:DispersionMinimizing} are
beyond scope here but can be realised within our computational framework.
All shape functions are centred around vertices, and we make each shape
functions cover exactly the $2^p$ cells of the vertex's level.
As the spacetree yields multiple vertices at one space coordinate and as each
level $\ell$ of the grid spans one function space $U_\ell$, the whole
spacetree induces a hierarchical generating system \cite{Griebel:94:Multilevel}.
Let each $v \in \mathbb{V}$ hold a three-tuple $(u,\phi,\chi)
(v) \in \mathbb{C}^3$.
$u$ is the weight of the shape function associated to the vertex,
i.e.~for shape functions $\psi (v)$ the solution of the discretised problem
is given by $\psi = \sum _{v \in \mathbb{V}} u(v) \psi (v)$.
According to \eqref{eq::pHelmholtz}, $\phi$ holds a weight of the identity
discretised by shape and test function space.
For a fine grid vertex, $\chi$ accordingly holds the weight of the discretised
right-hand side.
For a refined vertex, $\chi $ holds the right-hand side of the multigrid scheme. 
To reduce notation, we reuse the function symbols $u$, $\phi$ and $\chi$ from
the continuous formulation for the vectors of nodal weights, as the semantics of the
symbols is disambiguous, and we omit the $v$-parameterisation.
For all quantities, let the 
subscript $_\ell$ select levels.

$H_\ell$ in Algorithm \ref{algo::add} has a two-fold meaning. 
In unrefined vertices, it is a generic symbol for one Helmholtz operator from \eqref{eq::coupledHelmholtz}. 
It comprises a complex rotation factor and is subject to a smoother $S$ that, in our case, denotes one Jacobi smoothing step 
\[
  S_\ell(u_\ell, \chi _\ell) = diag(H_\ell )^{-1} (\chi _\ell -
  H_\ell u_\ell).
\]
In refined vertices, it encodes a correction equation. 
On one level $\ell$, regions may exist where $H_\ell $ has either of
these two functions, i.e.\ where fine cells align the edge of a coarser cell.

A generating system in combination with the fact that each
multigrid sweep starts from a coarse grid correction being zero renders the
realisation of an additive multigrid for a regular grid straightforward:
We set $u(v)=0$ for all refined vertices, and initialise two temporary variables
$r(v)=0$ and $diag(H_\ell)(v)=0$ everywhere.
Once we run into a cell, we evaluate the local system matrix in that cell and
accumulate the result within the residual $r$ of the $2^p$ adjacent vertices.
Further, it is convenient to assemble the diagonal element.
For plain Dirichlet boundary conditions and rediscretisation, an 
accumulation of $diag(H_\ell)(v)=0$ is unnecessary.
The value is known explicitly from $\phi$.
For complex scaling of cells near to the boundary as required for absorbing
boundary layers or spatially varying $\phi $, such an explicit accumulation of
$diag$ is mandatory.
Once all adjacent cells of a vertex on one level have been visited, we may
add the right-hand side to the residual and update the nodal value accordingly.
The right-hand side is not required earlier throughout the solve process.
The overall process is detailed in \cite{Mehl:06:MG} and reiterated in Appendix
\ref{section:appendix:element-wise}.
It is completely matrix-free.
It never sets up any global matrix.
This property also holds for restriction and prolongation which are done
throughout the grid/tree traversal process as well.

We use geometric grid transfer operators.
$P$ is a $p$-linear interpolation and $R=P^T$.
Since we rely on finite element rediscretisation and a uniform complex rotation
on all grid levels, $H_\ell$ on each level is a linear combination of a Laplace
operator and a mass operator.
The weighting of the two operators is, besides suitable $h$-scaling, invariant.
It is straightforward to validate that both operator component
rediscretisations equal a Galerkin coarse grid operator \cite{Trottenberg:01:Multigrid}
if the complex scaling and $\phi$ are invariant. 
We obtain $H_\ell = R H_{\ell+1} P $.
As we rely on rediscretisation, $diag(H_\ell )^{-1}$ further can be determined
on-the-fly on all levels as we accumulate the residual.
For varying $\phi$ or along boundaries with complex scaling, our rediscretised
operator slightly deviates from the Galerkin coarse grid operator.
However, it
still yields reasonable coarse grid corrections. 
Hanging vertices are interpolated from the next coarser grid through $p$-linear
interpolation, i.e.~their value results from $P$.
We note that we do not impose any balancing condition \cite{Sundar:08:BalancedOctrees} on the
spacetree.
With these ingredients, we can
interpret the fine grid solve as a domain decomposition approach: each fine grid
region of a given resolution is assigned Dirichlet values through its hanging
nodes  as well as real boundary vertices and computes a local solve. 
The next multigrid iteration, these Dirichlet values are changed to the most
recent, i.e.~updated solution on the coarser grids.

Our algorithmic sketch so far tackles the treatment of hanging nodes but omits
two major challenges: How do fine grid solutions from a level $\ell$ interact
with fine grid solutions of coarser resolutions, and how do we handle, on any
fixed level, the transition of the operator semantics from PDE discretisation
into multigrid operator along resolution boundaries?
We propose to rely on the concept of MLAT \cite{Brandt:73:MLAT,Brandt:77:MLAT}
for hierarchical generating systems which leads into the idea of Griebel's HTMG
\cite{Griebel:90:HTMM}.
To make a function representation unique within the generating system, we
enforce
\begin{equation}
   u_{\ell -1}  = Iu_{\ell}
   \label{equation:multigrid:injection}
\end{equation}
where $I$ is the injection, i.e.~plain copying of $u$ weights within our
spacetree.
This way, vertices that coincide with non-hanging vertices on finer levels act
as Dirichlet points for coarser fine grid problems.
Constraint \eqref{equation:multigrid:injection} formalises a
full approximation storage scheme (FAS) on our hierarchical generating system.
Homogeneous Dirichlet boundary conditions imply $u_0\equiv 0$.
Though no equation systems are to be solved on $\ell < \ell_{min}$, 
we nevertheless define $u_\ell$ on the coarser levels. 
 This does simplify our formulae and arithmetics, as we rely on the notion of
hierarchical surpluses \cite{Griebel:94:Multilevel} defined by the image of
the operator $id-PI$ with $id$ being the identity.
We stick to $ H u = \chi $
on the fine grid, but formally split the unknown scaling on refined vertices
into current solution plus correction. Following \cite{Griebel:90:HTMM}, this yields
\begin{eqnarray}
  H_\ell (u_\ell + e_\ell) & = & \chi _\ell + H_\ell u_\ell 
         \nonumber \\
         & = & R r_{\ell +1} + R H_{\ell +1}P Iu_{\ell+1} = R (\chi _{\ell +1} -
         H_{\ell +1} u_{\ell +1} + H_{\ell +1}P Iu_{\ell+1} ) 
         \nonumber \\
         & = & R (\chi _{\ell +1} - H_{\ell +1} \hat u_{\ell +1}) 
         =: R \hat r _{\ell +1}
         \qquad
         \mbox{with } \hat u_\ell := (id-PI) u_\ell.
         \label{equation:multigrid:HTMG}
\end{eqnarray}
The hierarchical surplus $\hat u$ is easy to compute if the coarse
grid already holds the injected solution.
Computing $\hat r$ parallel to the nodal residual for the smoother is
easy as well as it requires the same operations as the original fine grid
equation system.
The prolongation of the coarse grid correction finally simplifies if we add
an updated value on level $\ell $ to the hierarchical surplus $\hat u_{\ell +1}$ to obtain a new
nodal solution $u_{\ell +1}$ on level $\ell+1$:
\[
  u_{\ell} \gets u_{\ell} + P e_{\ell-1} = \hat u_{\ell} + P u_{\ell-1}.
\]
Such a scheme is agnostic whether original equation (in unrefined vertices) or multigrid correction (in refined vertices) is
solved on a level $\ell$.
All unknowns are of the same scaling and the $u$ weights have the same semantics
everywhere.
However, it misfits top-down tree traversals.

\begin{remark}
  \label{remark:exact-coarse-grid-solve}
  Textbook multigrid classically relies on an exact solve on the coarsest grid.
  We replace this solve by a sole smoothing step. 
  On the one hand, our experiments demonstrate empirically that this yields
  sufficiently accurate coarse grid solutions here.
  On the other hand, our solver acts as preconditioner. 
  Convergence to machine precision is not required.
\end{remark}

\noindent
This remark requires three addenda:
First, our additive multigrid solvers damp coarse grid contributions for
stability reasons.
The impact of the coarsest correction then is small anyway and an exact solve
would not pay off.
Second, in Poisson-like experiments we coarsen into very small grids (with
only $2^p$ unknowns, e.g.).
The small grid problems then are that small that one smoothing step already
reduces the residual significantly.
Finally, our complex rotations rotate the solution into a Poisson-like regime.
Otherwise, the problem could not be represented reasonably on very coarse
grids and would require an exact solve.

  \section{Single-sweep additive FAS within tree traversals}
\label{section:faddmg}

\begin{figure}
  \begin{center}
    \includegraphics[width=0.7\textwidth]{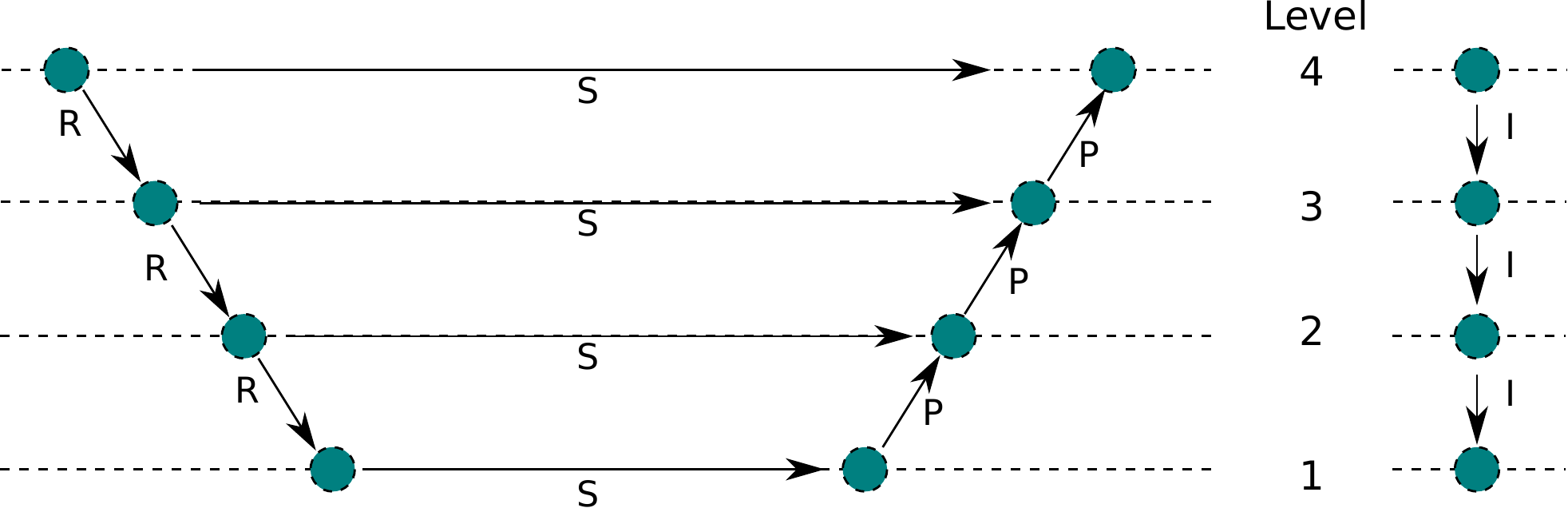}
  \end{center}
  \caption{
    Data flow diagram for the additive multigrid from Algorithm \ref{algo::add}
    (left) and data flow diagram for the injection of FAS (right).
    \label{fig:multigrid:data-flow}
  }
\end{figure}



\begin{theorem}
   Classic additive multigrid in combination with full approximation storage
   (FAS) can not be implemented with one tree traversal per cycle if we
   do not use additional variables as well as algorithmic reformulations.
\end{theorem}

\begin{proof}
  A proof relies on a data flow analysis (Figure
  \ref{fig:multigrid:data-flow}). The additive multigrid's prolongation prolongs
  corrections from coarser to finer grids where these are merged with the
  smoother's impact into an unknown update. Each update 
  requires a data flow to coarser grids due to the injection in
  \eqref{equation:multigrid:injection}. The resulting graph has cycles.
\end{proof}

\noindent
One might argue that a temporary violation of
\eqref{equation:multigrid:injection} on coarse grids is acceptable.
Yet, 
\begin{enumerate}
  \item our notions of hierarchical surplus and hierarchical residual
  require \eqref{equation:multigrid:injection} to hold. 
  \item our element-wise operator evaluation relies on the fact that all 
  vertices adjacent to this cell hold valid values. If some of them hold values
  that are not injected yet from a finer grid, the residual evaluation yields wrong results.
  \item our algorithm shall be allowed to refine and erase without restrictions.
  If \eqref{equation:multigrid:injection} holds, the deletion of whole subtrees
  is allowed. Otherwise, valid values first have to be injected prior to a grid
  update.
\end{enumerate}

\noindent
Additional tree traversals reconstructing \eqref{equation:multigrid:injection}
on-demand require multiple unknown reads and writes.
Notably, reconstruction schemes might run into a rippling effect
\cite{Sundar:08:BalancedOctrees} where an update implies follow-up
updates.
All techniques introduced from hereon avoid additional data access and 
advocate for the minimisation of multiscale grid sweeps.
This makes them future-proof regarding a widening memory gap, i.e.~growing
latency and shrinking bandwidth per core \cite{Kowarschik:00:CacheAwareMG}.

They rely on two ingredients. 
On the one hand, we shift the additive algorithm's unknown updates by
half an iteration, i.e.~grid sweep.
We run through the spacetree and determine unknown corrections to the current approximate solution. 
However, we do not feed them back into the solution immediately. 
Instead, we postpone the update and apply them in the beginning of the next
solver cycle.
On the other hand, we introduce two helper variables keeping track of updates
w.r.t.~the solution and, hence, also the injection. 
We preserve \eqref{equation:multigrid:injection} due to an exchange of updates.
We never compute the injection directly.

\begin{algorithm}[htb]
 \SetAlgoNoLine
    \label{algo::tdadd}
    \begin{algorithmic}[1]
       \Function{tdAdd}{$\ell $}
         \State $sc_\ell \leftarrow sc_\ell + P sc_{\ell -1}$ 
           \Comment{Add coarse grid correction to $sc_\ell$ which}
         \State
           \Comment{so far,
           holds update resulting from a Jacobi smoothing step.} 
         \State $u_\ell \leftarrow u_\ell + sc_{\ell} + sf_{\ell}$
           \Comment{Update $u$ with update from}
         \State \Comment{previous line plus all updates
           done on finer grids.} 
         \State $\hat u \leftarrow u_\ell - Pu_{\ell -1}$
           \Comment{Determine new hierarchical surplus.}
         \State
         \If{$\ell < \ell _{max}$}{
           \State \phantom{xx} \Call{tdAdd}{$\ell +1$}
         }
           \Comment{Go to next finer level.}
         \State $r_\ell \leftarrow b_\ell - H_\ell u_\ell$
           \Comment{Determine residual and}
         \State $\hat r_\ell \leftarrow b_\ell - H_\ell \hat u_\ell$
           \Comment{hierarchical residual.}
         \State $sc_\ell \leftarrow \omega _{\ell} S(u_\ell,b_\ell)$
           \Comment{Bookmark update due to a Jacobi}
         \State \Comment{smoothing step for next traversal. Usually}
         \State \Comment{uses $r$, otherwise there is no need to compute $r$.}
         \State
         \If{$\ell > \ell _{min}$}{
           \State \phantom{xx} $b_{\ell -1} \leftarrow R \hat r_\ell$
             \Comment{Determine right-hand side}
           \State \phantom{xx} 
             \Comment{for multigrid correction.}
           \State \phantom{xx} $sf_{\ell -1} \leftarrow I\left( sf_\ell +
           sc_\ell \right)$
             \Comment{Inform coarser grid about update}
           \State
             \Comment{due to smoothing that will happen in next traversal.}
         }
      \EndFunction
  \end{algorithmic}
  \caption{Additive multigrid with FAS that integrates into a tree traversal,
  i.e.~a coarse-to-fine sweep through the grid hierarchies.
  A call to \textsc{tdAdd}($\ell _{min}$) starts one multigrid cycle.
  $sc$ and $sf$ are helper variables introduced by pipelining that facilitate a
  single-touched policy.
  The transition from a correction scheme into FAS through the
  temporary helper variables $\hat r$ (hierarchical residual) and $\hat u$
  (hierarchical surplus) is illustrated in Algorithm
  \ref{algo::buadd} in the appendix.
  }
\end{algorithm}

All ideas materialise in Algorithm \ref{algo::tdadd} and realise the following
invariant:
\begin{eqnarray*}
  sc_\ell & \equiv & 0 \qquad \forall \ell <\ell_{min}, \qquad \mbox{and}
   \\
  sf_\ell & = & sc_\ell = 0 \qquad \mbox{at startup.}
\end{eqnarray*}

\noindent
Different to Algorithm \ref{algo::add} our realisation relies on a top-down
tree traversal: We start at the coarsest level rather than on the bottom of the tree.
Dynamic refinement thus remains simple.
Whenever a spacetree cell is traversed that has to be refined, one adds
an arbitrary number of levels, initialises all hierarchical surpluses with zero
and immediately descends into the new grid entities.
Higher order prolongation can be constructed starting from this $p$-linear
scheme.

\begin{remark}
  \label{remark:fmg:cycle}
  The single sweep FAS facilitates arbitrary on-the-fly refinement. Such a
  dynamic refinement plays two roles. On the one hand, it allows the algorithm
  to resolve local solution characteristics.
  On the other hand, it yields a full multigrid (FMG)-like algorithm. 
  We start from a coarse grid tackled by a series of additive V-cycles.
  Throughout these solves, we dynamically and locally add additional grid levels
  and thus unfold a coarse start grid into an adaptively refined accurate grid.
  Coarse solution approximations act as initial guesses for refined grids.
\end{remark}

We note that our pipelined implementation requires us to store an additional
two values $sc$ and $sf$ per vertex. They hold smoothing contributions
(therefore the symbol $s$) from the coarser grids (symbol $c$) plus 
the current grid or smoothing contributions from the finer grids (symbol $f$).
Besides these two values, also $\hat r$ and $\hat u$ have to be held.
However, we may discard those in-between two iterations, and so they do not permanently increase the memory footprint.

While the present approach picks up ideas of pipelining---rather than multiple
data consistency sweeps, we need only one amortised traversal per update; an
achievement made possible due to additional helper variables
(cmp.~\cite{Ghysels:13:HideLatency,Ghysels:14:HideLatency,Ghysels:15:GeometricMultigridPerformance})---the
exchange of updates might induce stability problems, i.e.~values on different levels that should hold the same values might
diverge.
We analysed this effect and studied the impact of a resync after a few 
iterations, but for none of the present experiments this resync has proven to be necessary.

  \subsection{$\omega$ choices}

Algorithms \ref{algo::add} and \ref{algo::tdadd} rely on relaxation
parameters which have severe impact on the efficiency and stability of the
resulting multigrid solver.
Proper choices deliver several well-established multigrid flavours. 
While the semantics of the Jacobi relaxation $\omega_S$ in Algorithm 
\ref{algo::add} is well-understood and can be studied in terms of local Fourier
analysis, Algorithm \ref{algo::add} also introduces $\omega_{cg}$-scaling of the
coarse grid contribution.
Multiple valid choices for this parameter do exist with
different properties \cite{Bastian:98:AdditiveVsMultiplicativeMG}.
In Algorithm \ref{algo::tdadd}, we refrain from distinguishing $\omega_{S}$ and
$\omega_{cg}$ in the formula but instead introduce a vertex-dependent
relaxation $\omega_\ell$, i.e.~each vertex may have its individual relaxation
factor.
As a vertex is unique due to its spatial position plus its level, this 
facilitates level-dependent $\omega$ choices.

Let $succ(v) \in \{0,\ldots,\ell _{max}-\ell _{min}\}$ be an integer
variable per vertex $v$ with
\[
  succ(v) = \left\{ 
    \begin{array}{ll}
      0 & \mbox{if $v$ is an unrefined vertex, or} \\
      \min _{i} (succ(v_i))+1 & \mbox{for all children $v_i$ of $v$
      otherwise.}
    \end{array}
  \right.
\]
A child vertex of a parent vertex is a vertex with at least one adjacent cell
whose parent in turn is adjacent to the parent vertex.
This property deduces from the parent-child relation on the tree.
Furthermore, we define the predicate $cPoint$ that holds for any
vertex whose spatial position coincides with a vertex position on the next coarser levels.
The predicate distinguishes c-points from f-points in the multigrid terminology.
We obtain various smoother variants:

\begin{center}
 \begin{tabular}{l|p{7cm}}
   Relaxation parameter & Description \\
   \hline
   \\[-0.2cm]
   $\omega_\ell(v)=\left\{  
   \begin{array}{lcl}
     \omega_{S}<1 & \mbox{if} & succ(v)=0 \ \mbox{and} \\
     0           && \mbox{otherwise.}
   \end{array}
   \right.$
   & Relaxed Jacobi on the dynamically adaptive
   grid as the {\it coarse relaxation parameter equals zero}. 
   \\ 
   \hline
   \\[-0.2cm]
   $\omega_\ell(v)= \omega_{S} < 1$
   & 
   {\it Undamped coarse grid correction}.
   \\
   \hline
   \\[-0.2cm]
  $\omega_\ell (v)=\left\{  
    \begin{array}{lcl}
      \omega_{S}<1 & \mbox{if} & succ(v)\leq L  \  \mbox{and} \\
      0           && \mbox{otherwise}
    \end{array}
  \right.$ 
  &
  {\it Undamped $L$-grid} scheme on adaptive grids
  \\
   \hline
   \\[-0.2cm]
  $\omega_\ell (v)= \left( \omega_{S} \right)^{succ(v)+1}\ \ \omega_S<1$
  &
  Classic additive multigrid from \cite{Bastian:98:AdditiveVsMultiplicativeMG} where
  coarse grid updates contribute to the fine grid solution with an {\em
  exponential damping}.
  \\
   \hline
   \\[-0.2cm]
  $\omega_\ell (v)= \left( \omega_{S} \right) ^{ (1-1/n) \cdot
  (succ(v)+1)}$
  &
  {\em Transition} relaxation. $n$ is the iteration counter.
  \\
  \hline
 \end{tabular}
\end{center}

\noindent
Schemes that use the same $\omega$ on each and every level (undamped coarse grid
correction) become unstable \cite{Bastian:98:AdditiveVsMultiplicativeMG} for
setups with many levels.
They tend to overshoot.
If only one or two grids are used for the multigrid scheme ($L\in\{2,3\}$), this
overshooting is not that significant and solvers are more robust.
However, we loose multigrid efficiency.
Exponential damping is thus used by most codes.
The coarser the grid the smaller its influence on the actual solution.
This renders an exact coarse grid solve unnecessary.
With the relaxation factors from above all the schemes apply straightforwardly
to dynamically adaptive grids.
Empirically, we observe that undamped schemes outperform their stable
counterparts in the first few iterations. 
The overshooting is not dominant yet. 
We therefore propose a hybrid smoother choice that transitions from an undamped
coarse grid correction into exponential damping.

\subsection{Hierarchical basis and BPX-type solvers}

\begin{figure}[htb]
  \begin{center} 
    \includegraphics[width=0.35\textwidth]{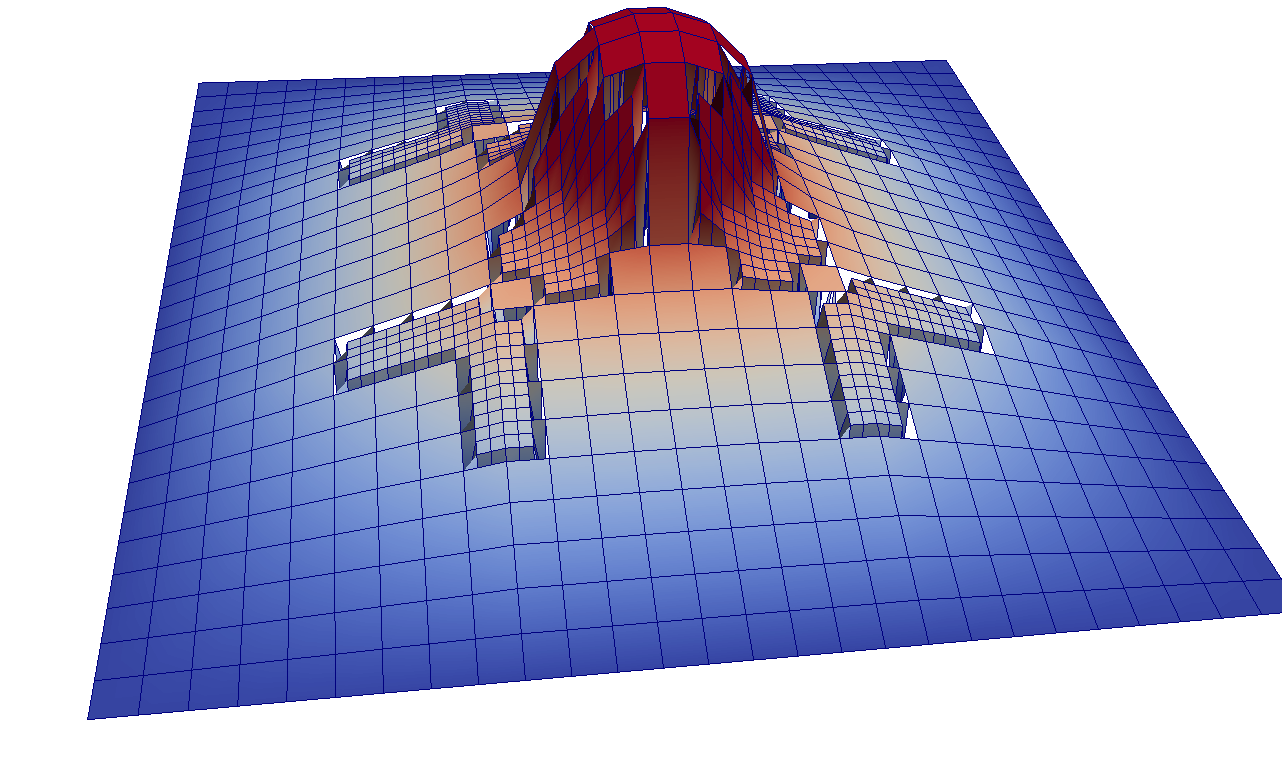}
    \includegraphics[width=0.35\textwidth]{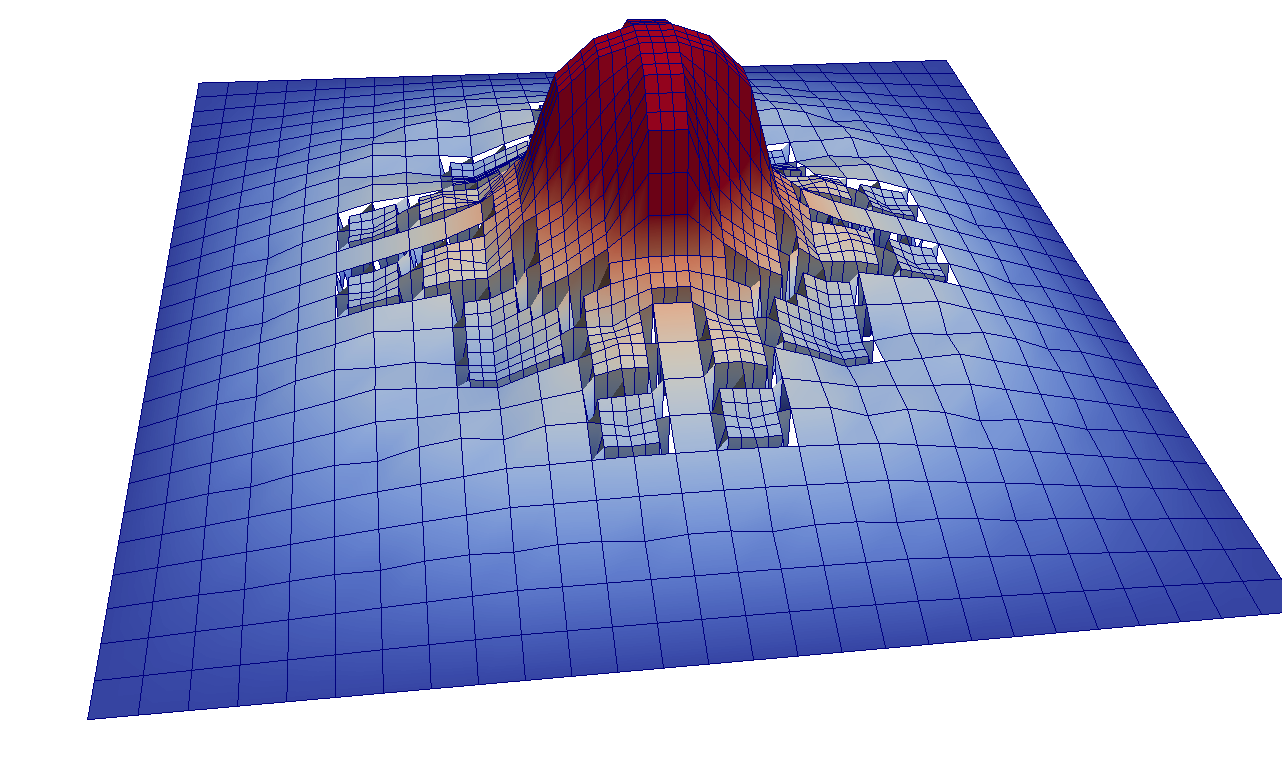}
    \\
    \includegraphics[width=0.35\textwidth]{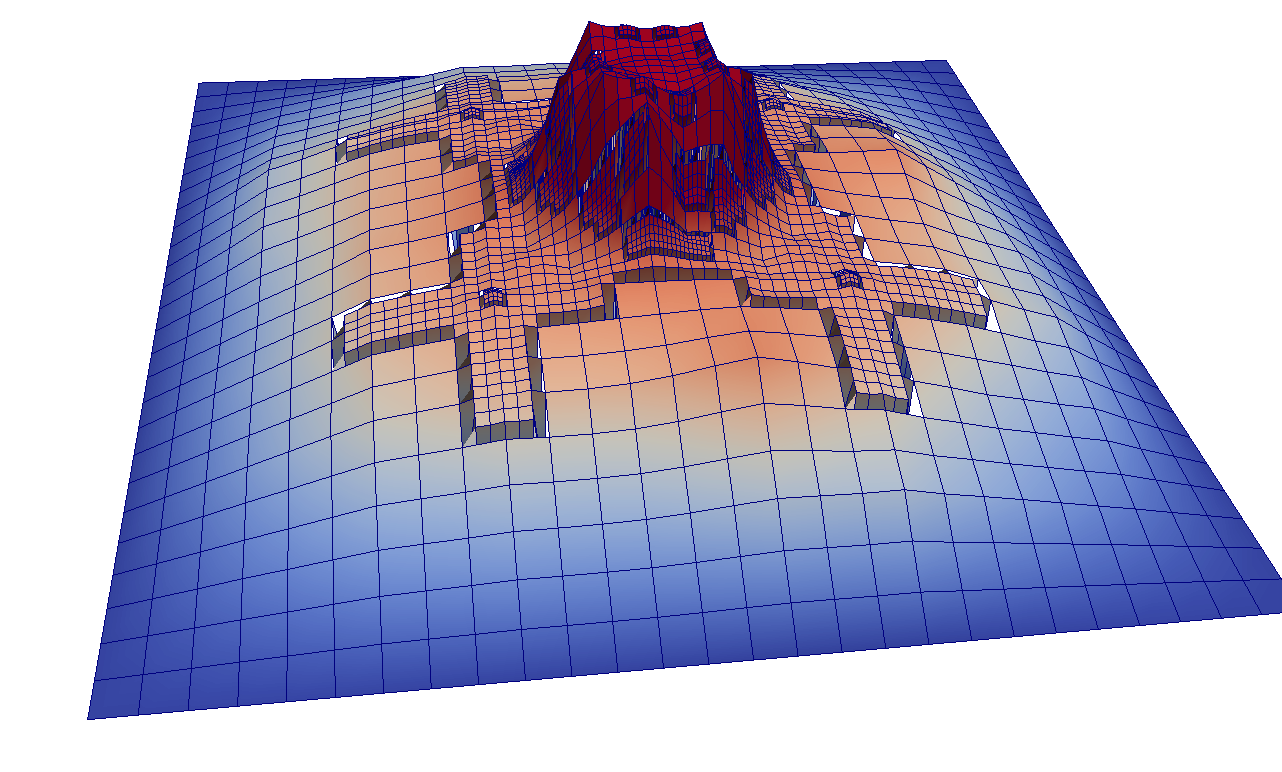}
    \includegraphics[width=0.35\textwidth]{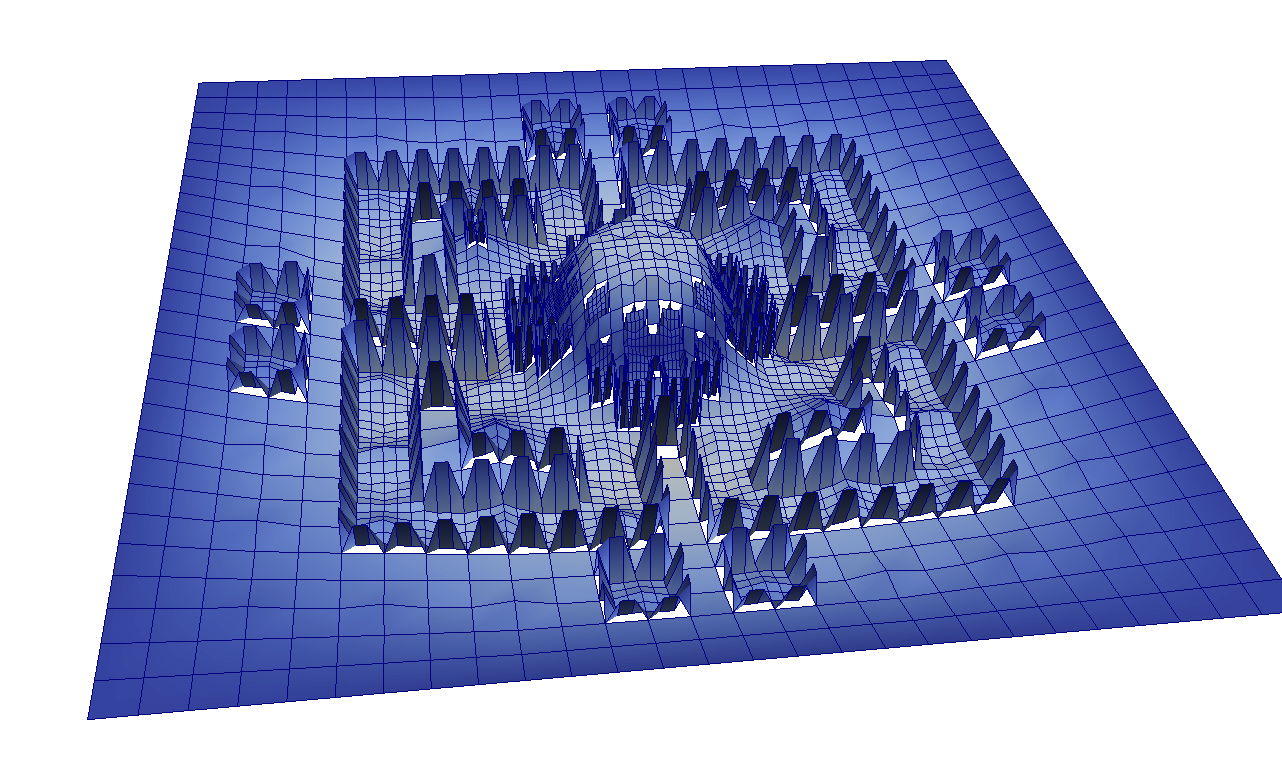}
    \\
    \includegraphics[width=0.35\textwidth]{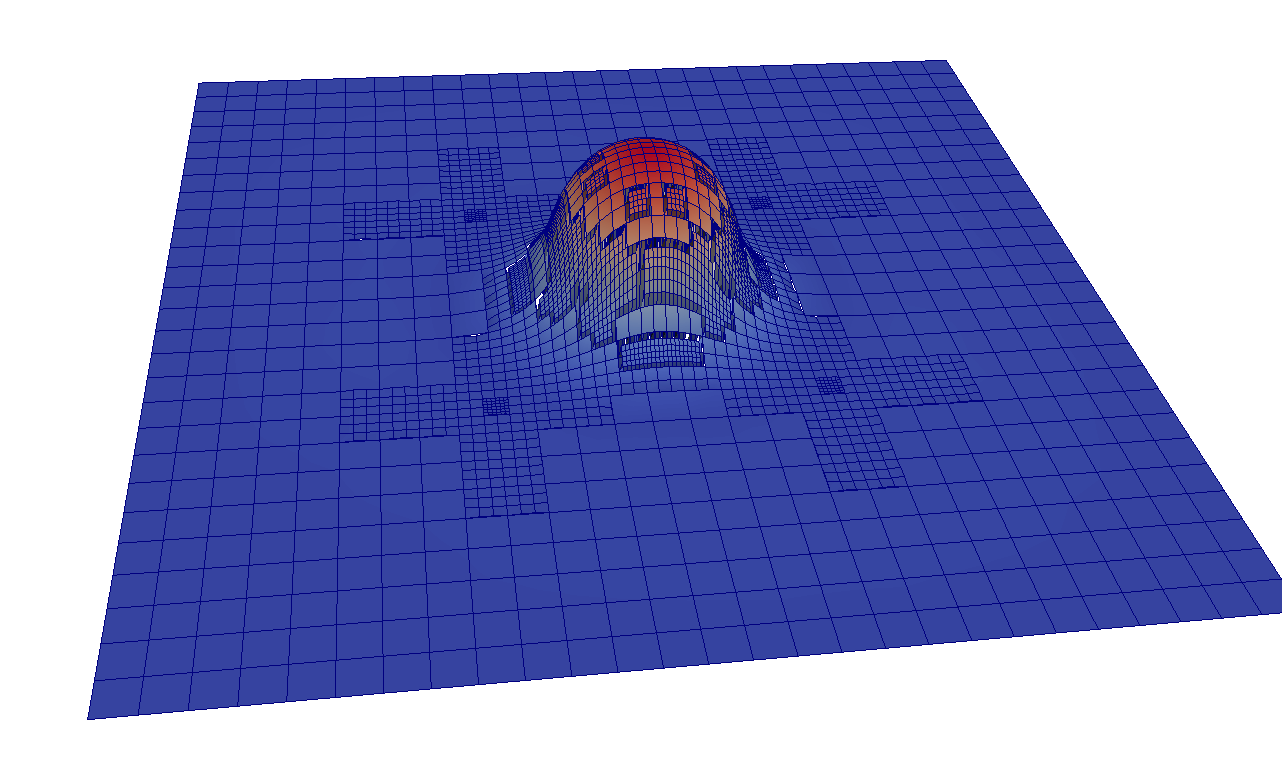}
    \includegraphics[width=0.35\textwidth]{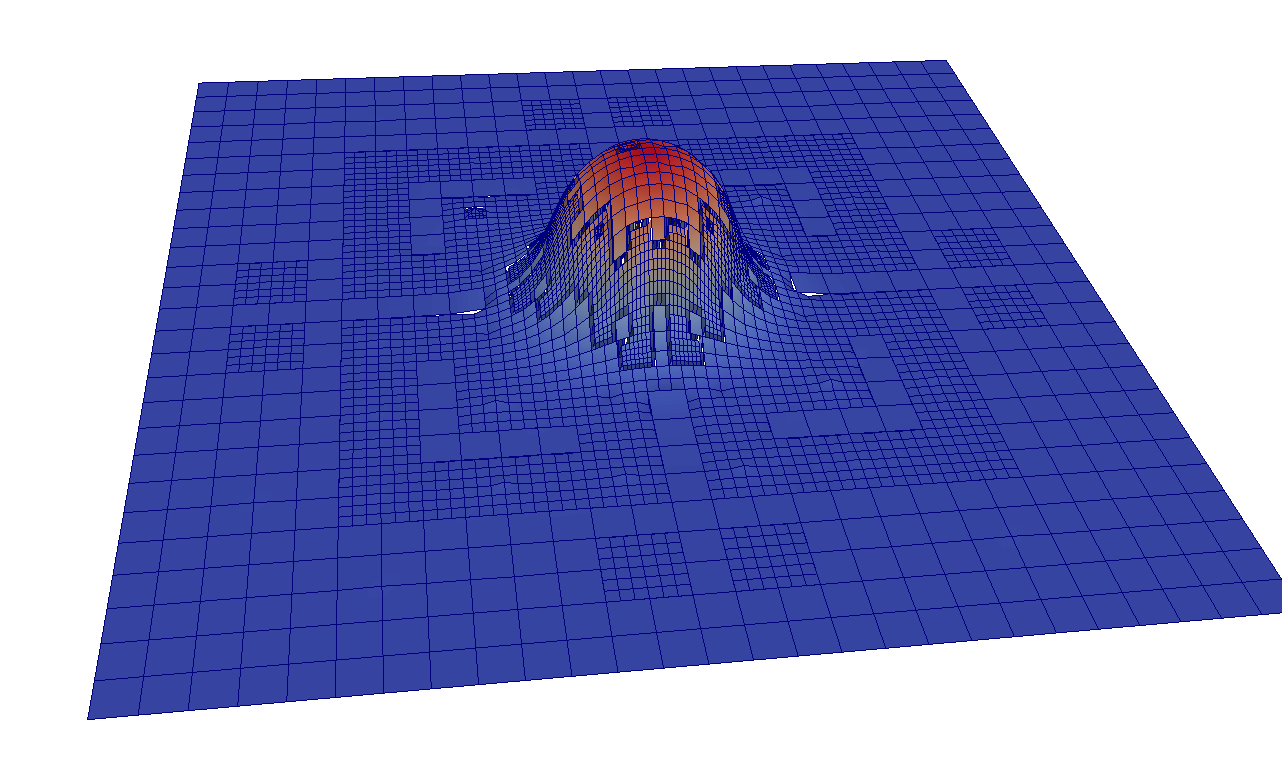}
  \end{center}
  \vspace{-0.4cm}
  \caption{
    Snapshot of the solution of $-\Delta u + 1000\ u = \chi$ with the transition
    scheme of the plain additive multigrid (left) and the hierarchical basis
    approach (right) after 5, 8 and 18 iterations. $\chi$ is a load of one within a circle around the
    domain's centre, i.e.~a characteristic function which can be modelled via a
    Heavyside operator.
    For $-\phi \approx 4000$, the plain additive scheme
    left becomes unstable due to overshooting already visible here.
    The \texttt{hb}-scheme updates per level only additional fine grid points
    compared to the coarser grids and thus is less sensitive to overshooting. 
    Both visualisations set hanging nodes to value zero. 
    The rough components in the pictures thus are visualisation artefacts; no
    real high frequency contributions.
   \label{fig::add-vs-hb}
  }
\end{figure}

While additive multigrid with exponential damping or transition is robust for
the Poisson equation, it runs into instabilities if we encounter a non-zero
$\phi<0$ term in \eqref{eq::pHelmholtz}; despite the fact that the problem
remains well-posed positive-definite on all levels due to the additional minus sign in front of $\phi$.
The solver is sensitive to the reaction term.
Robustness with respect to a reaction term however is mandatory prior to
tackling any ill-definiteness. 

We find the shift $\phi<0$ make the additive multigrid
overshoot on coarser levels, pollute the approximation and introduce a non-local
oscillation in the follow-up iteration.
The overshooting/oscillation typically grows per iteration if the diffusive
operator is not dominating (Figure \ref{fig::add-vs-hb}).
A straightforward fix to this instability is the switch from a
hierarchical generating system into a hierarchical basis
\cite{Griebel:90:HTMM,Griebel:94:Multilevel}.
It is identified by an \texttt{hb-}
prefix from here on.
Following \cite{Bastian:98:AdditiveVsMultiplicativeMG}, such a switch results
from a modification of the generic relaxation parameter into
\begin{equation}
  \omega_\ell (v) \gets \left\{ 
    \begin{array}{rcl}
      0 & \mbox{if} & cPoint(v) \qquad \mbox{and} \\
      \omega_\ell (v) & & \mbox{otherwise.} 
    \end{array}
  \right..
  \label{equation:hb-omega}
\end{equation}

\noindent
We mask out c-points.
Such a modification unfolds a variety of reasonable
and unreasonable smoothing schemes due to the various choices of
$\omega_\ell(v)$ on the right-hand side.
Our numerical results detail this.

While \eqref{equation:hb-omega} with its localised updates---vertices coinciding spatially are updated solely on the
coarsest level---prevents the additive scheme from overshooting too
significantly for $\phi<0$, it comes at the price of a deteriorating convergence
speed. 
It continues to assume a uniform smooth geometric multiscale behaviour of the solution, as 
any unknown update is determined by the update in the point plus the c-point
updates of surrounding vertices.
Increasing absolute values of $\phi$ in combination with non-smooth right-hand
sides however decreases the smoothness of the solution along jumps of
the latter.
This becomes apparent immediately at hands of a gedankenexperiment with a Heavyside $\chi$.
A \texttt{hb-}solver locally overshoots where $\chi$ changes and the overall
approximation starts to creep towards the correct solution due
to local oscillations while non-local oscillations are eliminated.

One fix to this challenge adds an additional $-PI\omega _\ell
S(u_\ell, b_\ell)$ term to all smoothing updates.
This is known as BPX
\cite{Bastian:98:AdditiveVsMultiplicativeMG,Bramble:90:BPX}.
Reiterating through our data dependency analysis, we find that the BPX operator
can not be implemented straightforwardly within Algorithm
\ref{algo::tdadd}---even with the pipelining variables in place---as 
c-point impacts spread to their surrounding through the coarser grids while they
are not altered themselves.

These considerations lead to a BPX FAS in Algorithm
\ref{algo::tdbpx}.
The key idea is to keep $sf$ and $sc$ and to introduce another helper variable
$si$ holding the injected value of a smoother update without any c-point
distinction. 
It is set as soon as we determine the smoother impact.
This impact is discarded for c-points due to
\eqref{equation:hb-omega}.
Finally, the one-sweep realisation modifies the prolongation by adding an
additional $-Psi$ term. 
\eqref{equation:hb-omega} in combination with this term ensures the BPX
inter-level correlation as we have $(id - PI)=0$ on vertices for
which $cPoint$ holds.

\begin{algorithm}[htb]
 \SetAlgoNoLine
    \begin{algorithmic}[1]
       \Function{tdBPX}{$\ell $}
         \State $sc_\ell \leftarrow sc_\ell + P sc_{\ell -1}$ 
         \State \If{not $cPoint(v)$}
         {
           \State \phantom{xx} $sc_{\ell} \leftarrow sc_{\ell} -
           Psi_{\ell-1}$
           \Comment BPX-type modification of fine grid correction
           \label{line:tdbpx:pmodification}
         }         
         \State $u_\ell \leftarrow u_\ell + sc_{\ell} + sf_{\ell}$
         \State $\hat u \leftarrow u_\ell - Pu_{\ell -1}$
         \State
         \If{$\ell < \ell _{max}$} {
           \State \phantom{xx} \Call{tdBPX}{$\ell +1$}
         }
         \State $r_\ell \leftarrow b_\ell - H_\ell u_\ell$
         \State $\hat r_\ell \leftarrow b_\ell - H_\ell \hat u_\ell$
        \State
		\eIf{$cPoint(v)$} {
			\State \phantom{xx}  $sc_\ell (v) \leftarrow 0$
			\Comment Realise \eqref{equation:hb-omega}, i.e.~cancel out
			update}{
			\State \phantom{xx}  $sc_\ell \leftarrow \omega _{\ell} S(u_\ell,b_\ell)$
			\Comment{Anticipate coarse correction}
		}
		\State
         \If{$\ell > \ell _{min}$}{
           \State \phantom{xx} $si_{\ell-1} \leftarrow I\omega _{\ell}
           S(u_\ell,b_\ell)$
             \Comment{Memorise dropped fine grid update}
           \State \phantom{xx} $b_{\ell -1} \leftarrow R \hat r_\ell$
           \State \phantom{xx} $sf_{\ell -1} \leftarrow I\left( sf_\ell +
           sc_\ell \right)$
         }
      \EndFunction
  \end{algorithmic}
  \caption{
    Single-sweep BPX variant realisation incorporating FAS. Invoked by
     \textsc{tdBPX}($\ell _{min}$). We do not rely on
     \eqref{equation:hb-omega} here, i.e.~$\omega $ is $cPoint$-agnostic, as we
     realise the case distinction within the multilevel code.
    \label{algo::tdbpx}
  }
\end{algorithm}

We emphasise that \eqref{equation:hb-omega} follows the same pipelining idea we
introduced for the additive scheme and at the same time renders the storage of a
fine grid correction $sf$ unnecessary.
We could add any fine smoothing impact directly onto $sc$ and at the same time
skip the injection of $sf_\ell + sc_\ell$.
Such a BPX realisation uses the same data layout as the additive multigrid.
No $sf$ is to be held, but we need an additional $si$.
This preserves the number of variables.
The reason for this possibility results from the fact that the coarsest
vertex in $\Omega_h$ holds the valid nodal representation of the
solution in Algorithm
\ref{algo::tdbpx}.
All finer vertices at the same location are copies.
We preserve $sf$ in the presented code to emphasise the closeness to the
additive scheme.
While this wastes one entry per vertex, it might
make sense to preserve the fine grid injection and thus to allow
BPX's fine grid update to change $cPoints$ as well:
in applications with non-trivial boundary conditions, those sometimes are
simpler to evaluate in a nodal setting rather than a hierarchical basis.
The injection then automatically reconstructs the data consistency on all
levels.

\subsection{Feature-based dynamic adaptivity}

All algorithmic ingredients introduced are well-suited for any arbitrary 
adaptivity.
Throughout the top down steps, we may add any number of vertices as
long as we initialise their hierarchical surplus with 0 and prolong the solution
$p$-linearly.
They then seamlessly integrate into the solver's workflow. 
For faster convergence, higher order interpolation might be advantageous.
Discarding vertices is permitted throughout the backtracking, i.e.~the steps up
in the grid hierarchy.
The FAS ensures that all solution information is already available on the
coarsened mesh.
Multilevel meta information such as $cPoint$ or $succ(v)$ can be computed
on-the-fly throughout the tree traversal's backtracking.
It then automatically adopts to updated refinement patterns.

In the present paper, we stick to simple feature-based refinement and specify
both regular grids and adaptive grids through a maximal and
minimal mesh size $h_{max}$ and $h_{min}$.
We start from a grid satisfying $h_{max}$ and, in parallel to
the smoothing steps, measure the value $s = \max _{d \in \{1,\ldots,p\}} 
|\Delta _d u|$ per vertex on each grid level.
Per step, we refine the 10\% of the vertices with the highest $s$ value, while
we erase the 2\% vertices with the smallest $s$ value.
These values are shots from the hip but empirically show reasonable
grid refinement structures.
They yield a grid that adopts itself to solution characteristics.
We realise feature-based adaptivity.
More sophisticated schemes with proper error estimators are out of scope.

To avoid global sorting, we split up the whole span of $s$
values into 20 subranges and bin vertices into these ranges.
All vertices fitting into a fixed number of bins holding the largest $s$
values are refined.
This fixed number is selected such that the 10\% goal is met as close as
possible.
Erasing works analogously with the bins with the smallest $s$ values.
Refining and erasing are vetoed in two cases:
if maximal or minimal mesh constraints would be violated;
or if residual divided by diagonal element exceeds $10^{-2}$.
In the latter case, the vertex is still subject to major updates, i.e.~has not
`converged', and we postpone a refinement or coarsening.

The interplay of the feature-based refinement with the creation of a FMG cycle
is detailed in Remark \ref{remark:fmg:cycle}.
We note that our criterion yields different grid refinement patterns
for different solvers as we integrate refinement into the solve (Figure
\ref{fig::add-vs-hb}).
This advocates for better criteria and renders the present experiments 
feasibility studies.

  \section{Solver properties}
\label{section:properties}

In the following, we validate fundamental properties of the algorithm.
We prove its correctness. 
Let the iterates of an unknown $x$ be $x^{(n)},x^{(n+1)},\ldots$
As $sc$ is updated twice per iteration we distinguish $sc^{(n)},
sc^{(n+0.5)}$ and $sc^{(n+1)}$.

\begin{lemma}
  Whenever we evaluate $r_\ell^{(n+1)} = b_\ell - H_\ell u_\ell^{(n)}$, we
  have
  \[
    \forall n: \qquad  
    u_{\ell-1}^{(n)} = I u_{\ell}^{(n)} \qquad \forall \ell _{min} < \ell
    \leq \ell _{max}.
  \]
\end{lemma}

\begin{proof}

At construction, $u_\ell^{(0)} = I u_{\ell+1}^{(0)}$. 
The proof then relies on a simple induction over $n$:
\begin{eqnarray*}
  u_{\ell-1}^{(n+1)} 
    & = & u_{\ell-1}^{(n)} + sf_{\ell-1}^{(n+1)} + sc_{\ell-1}^{(n+1)} \\
    & = & Iu_\ell^{(n)} + I \left( sf_{\ell}^{(n+1)} + sc_{\ell}^{(n+1/2)}
    \right) + sc_{\ell-1}^{(n+1)} 
    \\
    & = & Iu_\ell^{(n)} + I sf_{\ell}^{(n+1)} + I \left( sc_{\ell}^{(n+1)}
    - P sc_{\ell -1}^{(n+1)} \right) + sc_{\ell-1}^{(n+1)} 
    \\
    & = & Iu_\ell^{(n)} + I sf_{\ell}^{(n+1)} + I sc_{\ell}^{(n+1)}
    + (id-IP) sc_{\ell-1}^{(n+1)} 
    \\
    & = & I u_{\ell}^{(n)}
\end{eqnarray*}

\noindent
where we apply the induction hypothesis on $u_\ell^{(n-1)}$, the algorithm's
update operations on the remaining operators, and exploit $IP=id$ for c-points.
Dynamic adaptivity has to preserve the construction constraint.
\end{proof}

\noindent 
The smoother here is a black-box and we do not make
any assumption about the correctness of the solved equation systems. Our
notation of $I, P$ and $R$ further is generic, i.e.~$I^k$ indicates that $I$ is
applied multiplicatively $k$ times in a row with an $I$ fitting to the preimage.

The lemma implies that transitions between grid resolutions require no special treatment: 
As each vertex on each grid level holds a valid
representation of the solution, we can apply the same stencils irrespective
whether they overlap a refined region or are fine grid stencils.
Only hanging nodes have to be interpolated $p$-linearly from the coarser grid.

\begin{lemma}
  The operators from Algorithm \ref{algo::tdadd} realise a FAS.
\end{lemma}

\begin{proof}
Each level's smoother tackles a correction equation
of the form
\[
  H_\ell \left( u_\ell + Iu_{\ell+1} \right) = R r_{\ell+1} +
  H_\ell Iu_{\ell+1}.
\]
The left-hand side has been studied before.
So we follow \cite{Griebel:94:Multilevel} and write 
\begin{eqnarray*}
  R \hat r_{\ell+1} 
    & = & R \left( b_{\ell+1} - H_{\ell +1} \hat u_{\ell+1} \right)
    \\
    & = & R \left( b_{\ell+1} - H_{\ell +1} (u_{\ell +1}-Pu_{\ell}) \right)
    \\
    & = & R r_{\ell+1} + RH_{\ell+1}Pu_{\ell} = R r_{\ell+1} +
    RH_{\ell+1}PIu_{\ell +1} \quad \mbox{(cmp.~Lemma 1)} \\
    & = &     
    R r_{\ell+1} + H_\ell Iu_{\ell+1}.
\end{eqnarray*} 
\end{proof}

\noindent
The proof holds for space-independent Helmholtz shifts $\phi$ and Dirichlet and Neumann boundary
conditions.
As mentioned before, absorbing boundary layers with varying complex scaling or
jumping material coefficients violate the equivalence of a Galerkin multigrid
operator and plain rediscretisation.
We have to assume that these are the regions many adaptivity criteria refine
towards to.
For these, the equations above comprise an additional error term if we do not
apply operator-dependent grid transfer operators.
The multigrid equations are perturbed.
Based on empirical evidence, we assume this perturbation to be small.
However, we have to expect a small deterioration of the multigrid
performance.
Without pollution, the prolongation of 
\[
  u_\ell^{(n+1)} - u_\ell^{(n)} = \omega _{\ell -1} S(u_\ell
  ^{(n)},b_\ell) = cs _\ell^{(n+1)}
\]
to level $\ell+1$ equals the correction term in
multigrid.

\begin{remark}
The lemma extends naturally to Algorithm \ref{algo::tdbpx}, i.e.~the BPX
variant.
\end{remark}

\begin{theorem}
  The additive top-down FAS from Algorithm \ref{algo::tdadd} realises an
  additive multigrid algorithm.
\end{theorem}

\begin{proof}
We compare the algorithm to additive blueprint in Algorithm \ref{algo::add} and
study one iteration of the scheme.
\begin{itemize}
  \item We first study the one-grid problem with $\ell _{min}=\ell _{max}$. We
  further focus on $u_\ell ^{(n+1)}$, i.e.~work backwards from the update of
  this unknown.
  \begin{eqnarray*}
    u_\ell ^{(n+1)} & = & u_\ell ^{(n)} +  Psc_{\ell-1} ^{(n+1)} + sc_\ell
    ^{(n+1)} +  sf_\ell ^{(n+1)}
    \\
      & = & u_\ell ^{(n)} +  sc_\ell ^{(n+1)} +  sf_\ell ^{(n+1)} \qquad
      \mbox{as } sc_{\ell -1} \equiv 0
    \\
      & = & u_\ell ^{(n)} +  \omega _\ell S(u_\ell ^{(n)},b_\ell) +  sf_\ell
      ^{(n+1)}
      \qquad
  \end{eqnarray*}
  $sf_\ell$ is never modified which closes this step.
  \item We next switch to a two-grid problem with $\ell \equiv \ell_{max}$ and
  $\ell-1 \equiv \ell_{min}$.
  \begin{eqnarray*}
    u_\ell ^{(n+1)} & = & u_\ell ^{(n)} +  Psc_{\ell-1} ^{(n+1)} + sc_\ell
    ^{(n+1)} +  sf_\ell ^{(n+1)}
    \\
      & = & u_\ell ^{(n)} +  Psc_{\ell-1} ^{(n+1)} + sc_\ell ^{(n+1)}
      \qquad \mbox{induction} 
    \\
      & = & u_\ell ^{(n)} +  Psc_{\ell-1} ^{(n+1)} + \omega _\ell S(u_\ell
      ^{(n)},b_\ell)
    \\
      & = & u_\ell ^{(n)} +  P\left( \omega _{\ell -1}
      S(u_\ell ^{(n)},b_\ell) \right) + \omega _\ell
      S(u_\ell ^{(n)},b_\ell)
  \end{eqnarray*}
\end{itemize}
The proof follows from induction on the grid levels if we choose 
$\omega _{\ell} = \omega _S$ on the finest level and $\omega _{\ell} = \omega _S
\cdot \omega _{cg}^l$ with $l$ being the difference of the current level to the
finest level. The latter is an attribute that can be computed on-the-fly
throughout the bottom-up steps of the algorithm within the spacetree.
\end{proof}

\begin{remark}
We assume that a theorem for the BPX solver from Algorithm \ref{algo::tdbpx} is
proven analogously.
\end{remark}

  \section{Results}
\label{section:results}

Our results split into five parts. 
First, we study the convergence behaviour of the additive multigrid
variants for a simple Poisson equation.
This validates the algorithmic building blocks at hands of a well-posed setup
and yields insight into the convergence efficiency.
Second, we switch to Helmholtz problems with $\phi<0$ in
\eqref{eq::pHelmholtz}.
This reveals the shortcomings of the additive scheme compared to a 
hierarchical basis approach that arises naturally from our chosen data
structures.
Third, we study Helmholtz problems with $\phi>0$, i.e.~the difficult case of 
ill-conditioned problems.
This quantifies the efficiency and robustness of the proposed solution with
complex grid rotation.
Fourth, we apply our toolset on the motivating scattering example.
This validates that the approach is well-suited to tackle varying material
parameters $\phi$.
Finally, we study the efficiency of the proposed solver regarding hardware
characteristics.
For the present perfectly parallel setup, this notably has to
focus on vectorisation and memory access efficiency.
The latter case studies are conducted either on a local workstation with
Sandy Bridge-EP Xeon E5-2650 processors running at 2.0 GHz or on a Xeon Phi
5110P accelerator running at 1,053 GHz.
The latter is a pioneer of future manycore architectures delivering performance
to a significant extent through vectorisation.
At the same time, it is well-known to be sensitive to proper memory layout and
access \cite{Reinders:15:ProgrammingPearls}. 
Statements on this machine thus facilitate an extrapolation to upcoming hardware
generations.
Our codes were translated with the Intel compiler 15.0.1, while we used the
Likwid tools \cite{Treibig:10:Likwid} to obtain statements on hardware counters.

Global quantities are specified over all fine grid vertices, i.e.~a subset of
$\mathbb{V}$; either in the maximum norm, the standard Euclidean norm or
\[
 | x |^2_h = \sum _i |h_i|^d |x_i|^2,
\]
where each entry of the input vector is scaled by the volume of one spacetree
cell of the corresponding mesh level.
While the $_h$-norm
is an Euclidean inner product norm (cmp.~(1.3.6)
from \cite{Trottenberg:01:Multigrid}) anticipating nonuniform mesh elements and
thus allows us to compare adaptive grids with each other, we emphasise that it is not an exact equivalent of the continuous L2 norm as small slices of the domain along adaptivity
boundaries are integrated several times due to our hierarchical ansatz space.
For regular grids, it is exact.
The norm allows us to compare
grids of different resolution \cite{Trottenberg:01:Multigrid} or adaptivity
patterns with each other.
We thus have to expect minor peaks in the residual whenever
the grid is dynamically refined.

\begin{eqnarray}
  \chi (x)& = & d \pi ^2 \cdot \prod _{i=0,\ldots, d} \sin (\pi x_i)
    \qquad \mbox{and}  
    \label{eq:results:sin} \\
  \chi(x) & = & 
    \begin{cases}
      1, \quad \text{if }\sum_{j=1}^{p}(x_j-\frac{1}{2})^2 < 0.1^2, \\
      0, \quad \text{elsewhere.}
    \end{cases} 
    \label{eq:results:helmhomo} 
\end{eqnarray}

\noindent
act as artificial benchmark problems before we switch to realistic setups in
Section \ref{section:results:gaussian}.
Homogeneous Dirichlet boundary conditions are supplemented for the benchmarks.

\subsection{Poisson problems}

\begin{figure}
  \begin{center} 
  \includegraphics[width=0.45\textwidth]{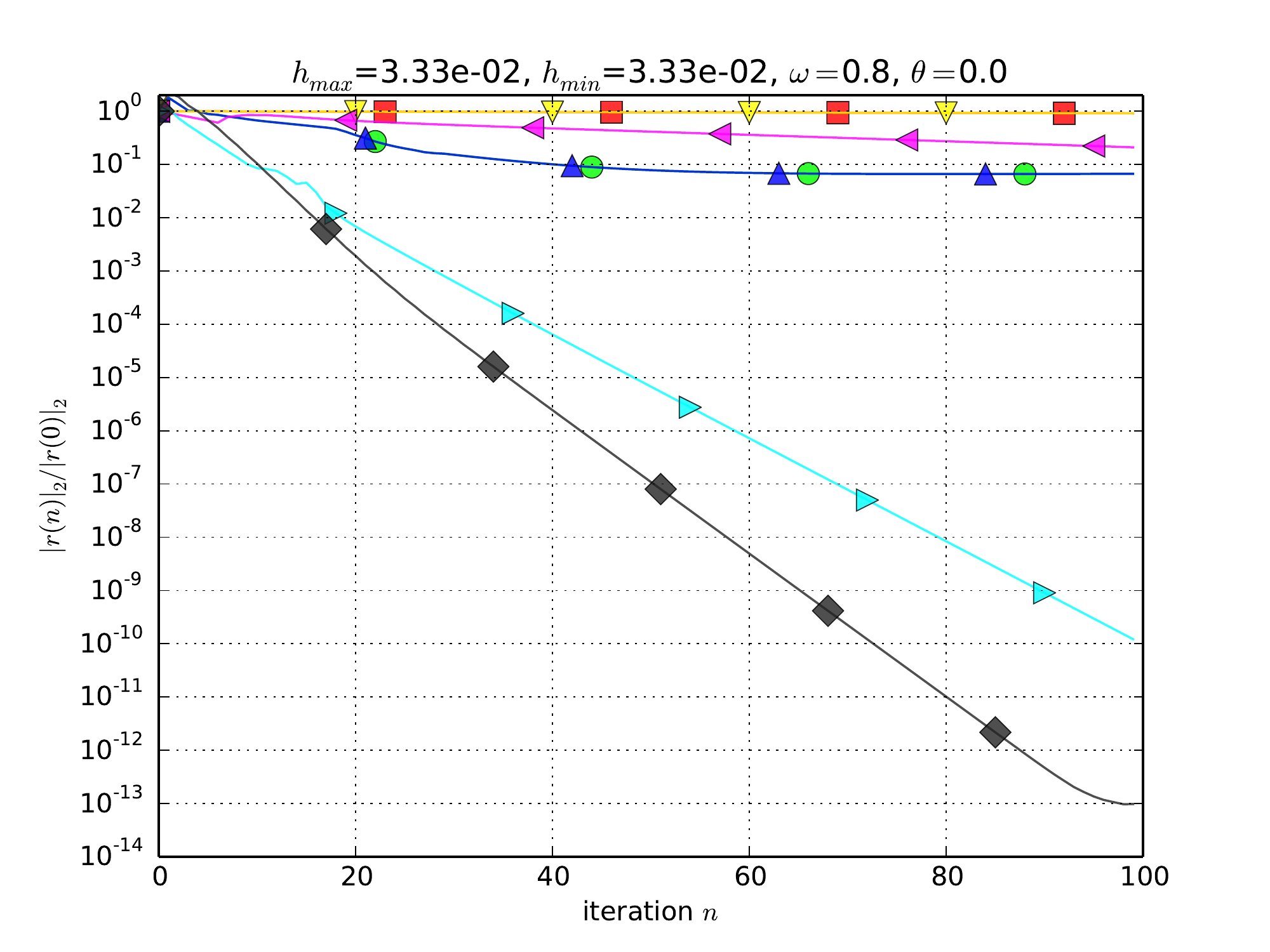}
  \includegraphics[width=0.45\textwidth]{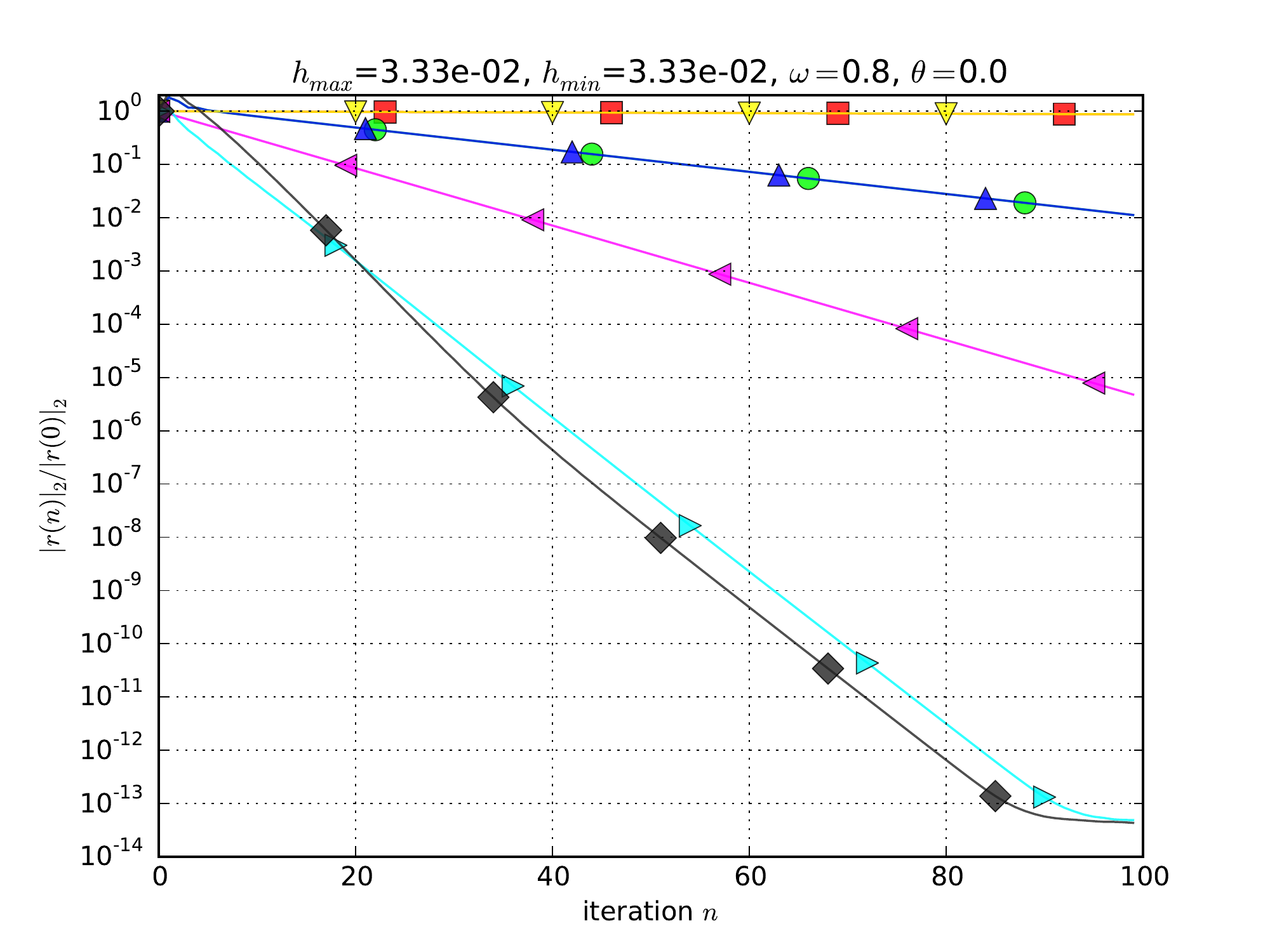}
  \\
  \includegraphics[width=0.45\textwidth]{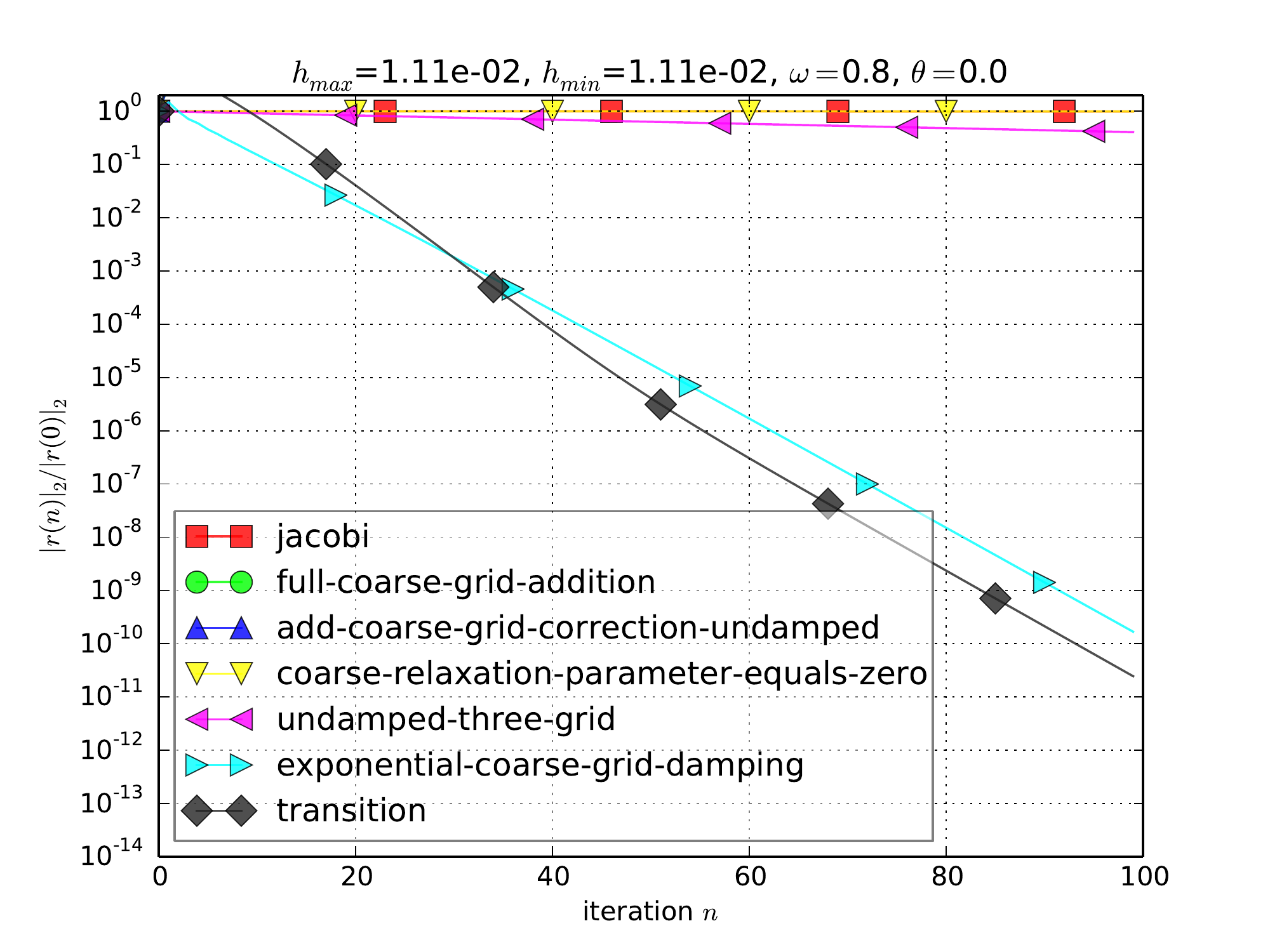}
  \includegraphics[width=0.45\textwidth]{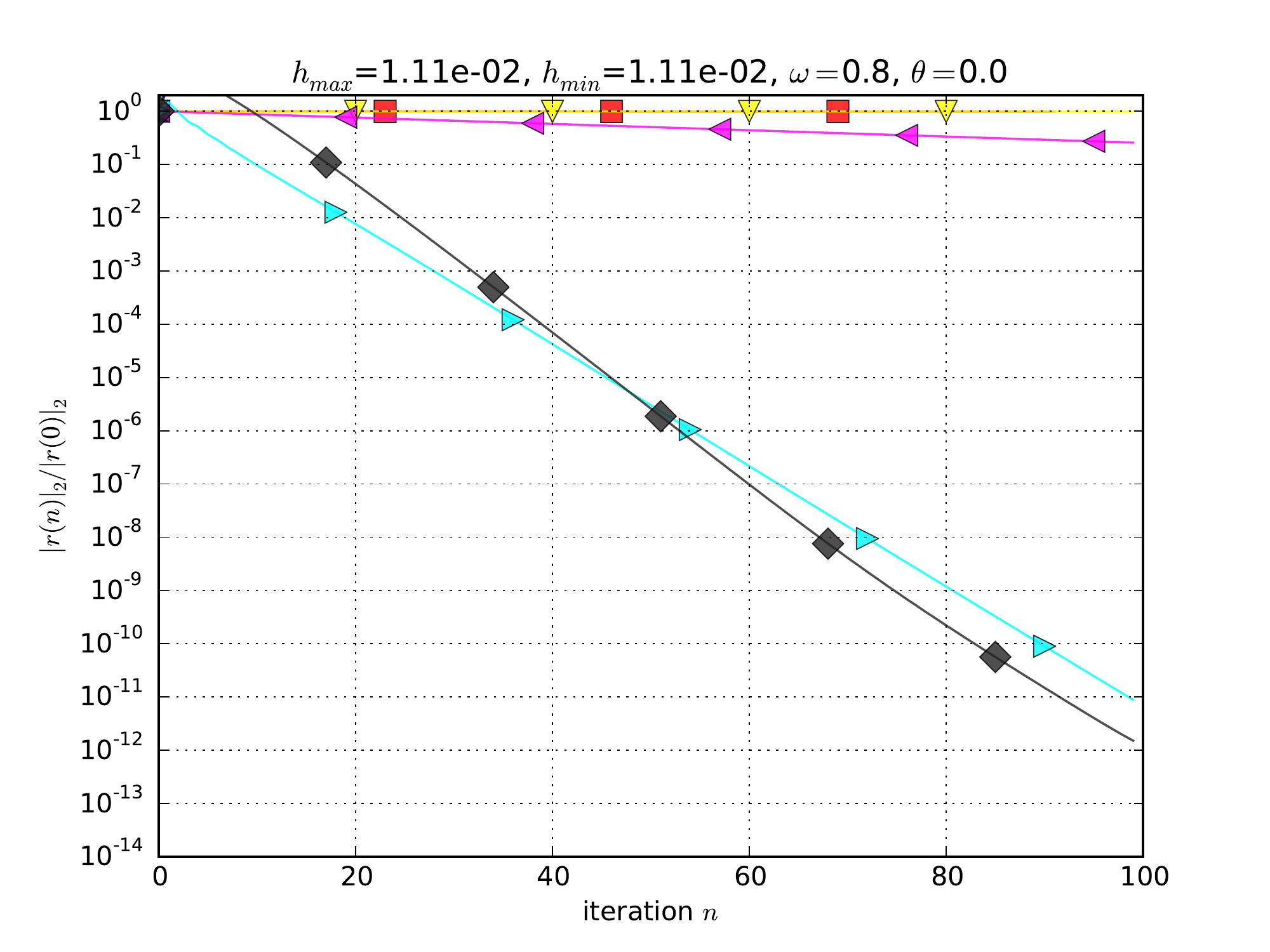}
  \end{center}
  \vspace{-0.2cm}
  \caption{
    Residual development on regular grids (left: $p=2$; right: $p=3$) for the
    additive multigrid and the right-hand side \eqref{eq:results:sin}.
  }
  \label{figure:08a:convergence-regular}
\end{figure}

We start from \eqref{eq:results:sin} and $\phi \equiv 0$.
The complex rotation is $\theta=0$. 
An analytical solution to this problem has no imaginary parts, and we have 
validated for all choices of $\omega_\ell(v)$ the convergence
towards the analytic solution.
A convergence study on uniform grids exhibits a convergence speed
that is almost mesh-independent for the different multigrid variants (Figure
\ref{figure:08a:convergence-regular}).
We observe that the speed slightly decreases for decreasing mesh width if we use
exponential or transition coarse grid damping.
While these two schemes are the robust multiscale variants, they downscale the
impact of coarse grid corrections per additional coarse grid level.
We can not expect perfect mesh-independent convergence.
In general, the transition scheme slightly outperforms the exponential coarse
grid damping after a few additive cycles. 
Undamped additive multigrid converges only for the coarser mesh
widths---it would still work if we reduced $\omega $ with each additional grid
level; which decreases the convergence speed---while the three grid and Jacobi
solver exhibit the well-known mesh-dependent convergence behaviour.

\begin{figure}
  \begin{center} 
  \includegraphics[width=0.45\textwidth]{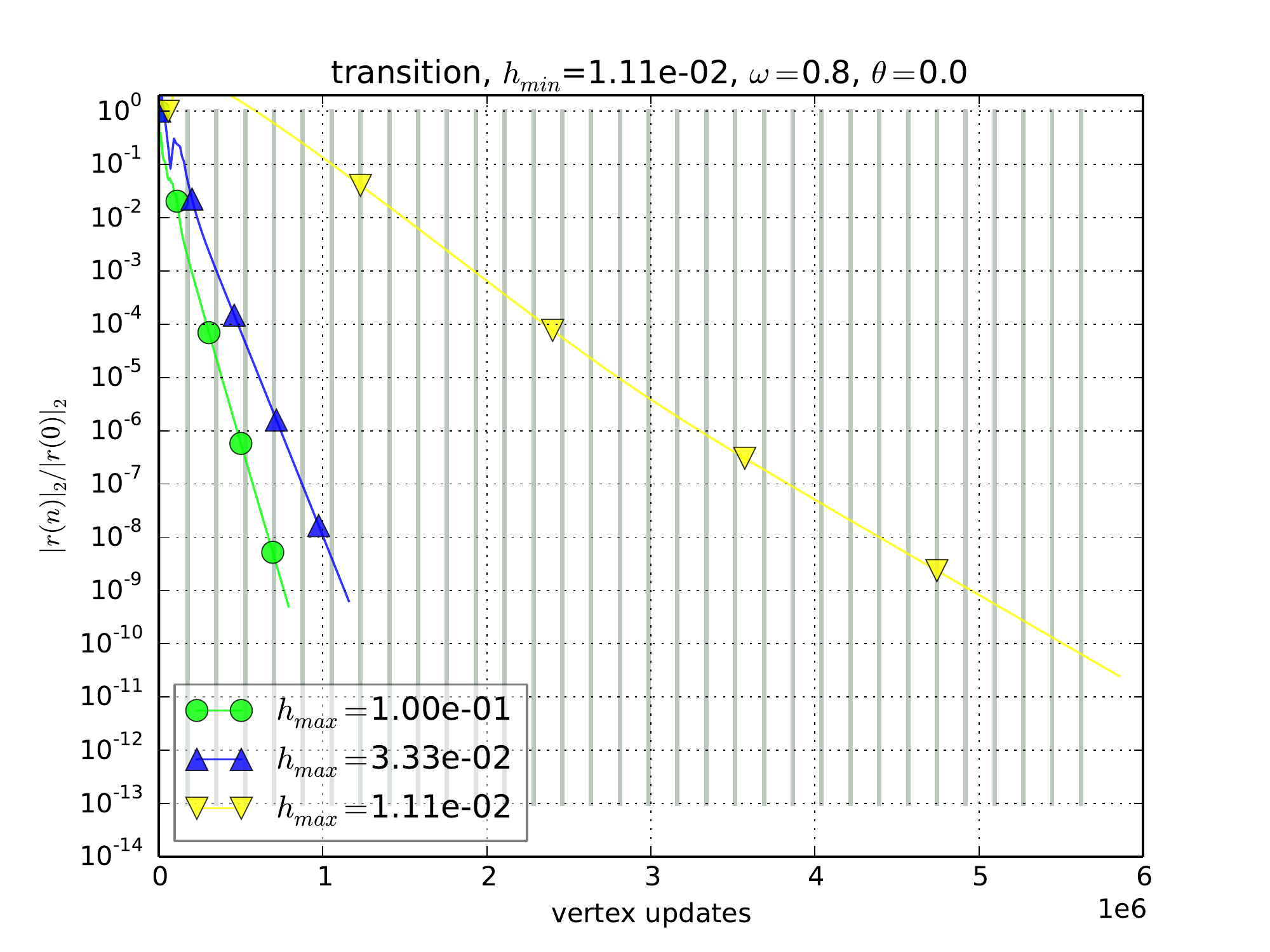}
  \includegraphics[width=0.45\textwidth]{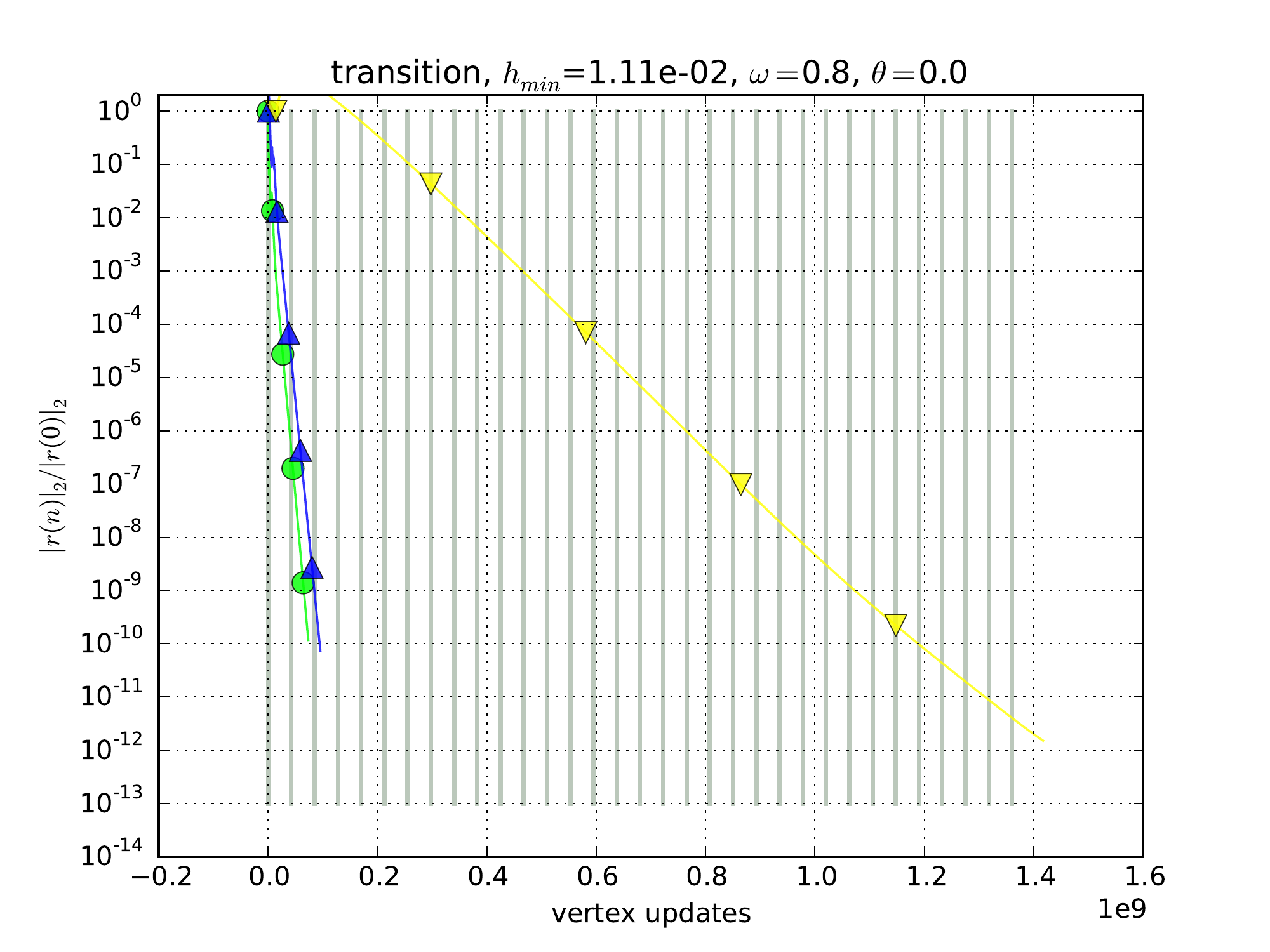}
  \end{center}
  \vspace{-0.2cm}
  \caption{
    Adaptive grid with additive multigrid that is successively refined from a
    prescribed maximum mesh size $h_{max}$ to the minimum mesh size $h_{min}$
    for $p=2$ (left) and $p=3$ (right) where it pays off. Each thick vertical
    line denotes the cost of three additive cycles on a regular grid with
    $h_{min}$.
    Transition is used for $\omega_\ell(v)$. The right-hand side is \eqref{eq:results:sin}. 
  }    
  \label{figure:08a:convergence-adaptive}
\end{figure}

The aforementioned convergence statements reveal to be too pessimistic if we
switch to an adaptive, FMG-type setting.
We then observe that we obtain two orders of
accuracy at the cost of around three fine grid cycles (Figure
\ref{figure:08a:convergence-adaptive}).
This holds for both the exponential coarse grid damping (not shown) and the transition
scheme while the latter performs slightly better again.
We are close to multiplicative multigrid efficiency.

\begin{figure}
  \begin{center} 
  \includegraphics[width=0.45\textwidth]{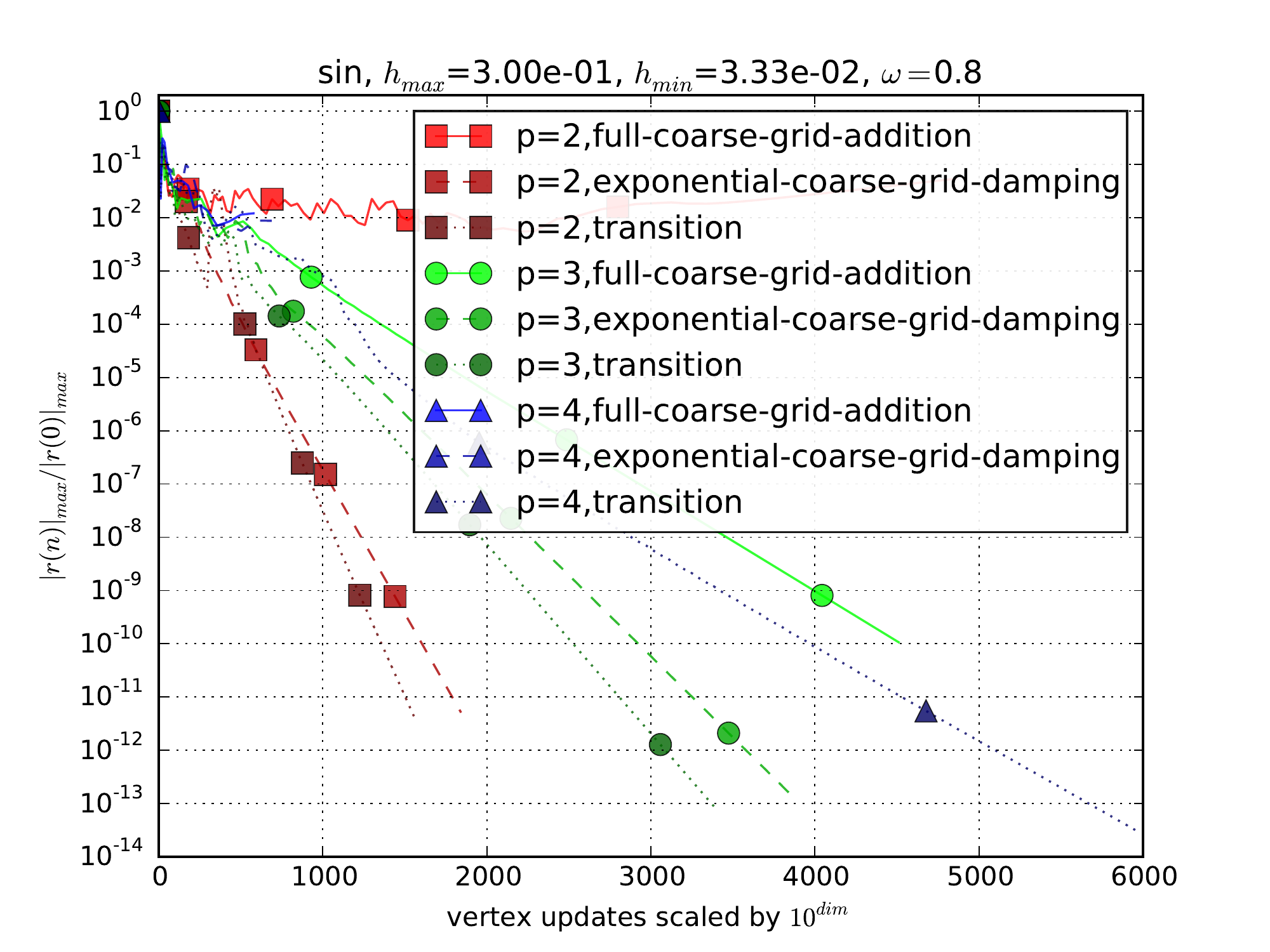}
  \includegraphics[width=0.45\textwidth]{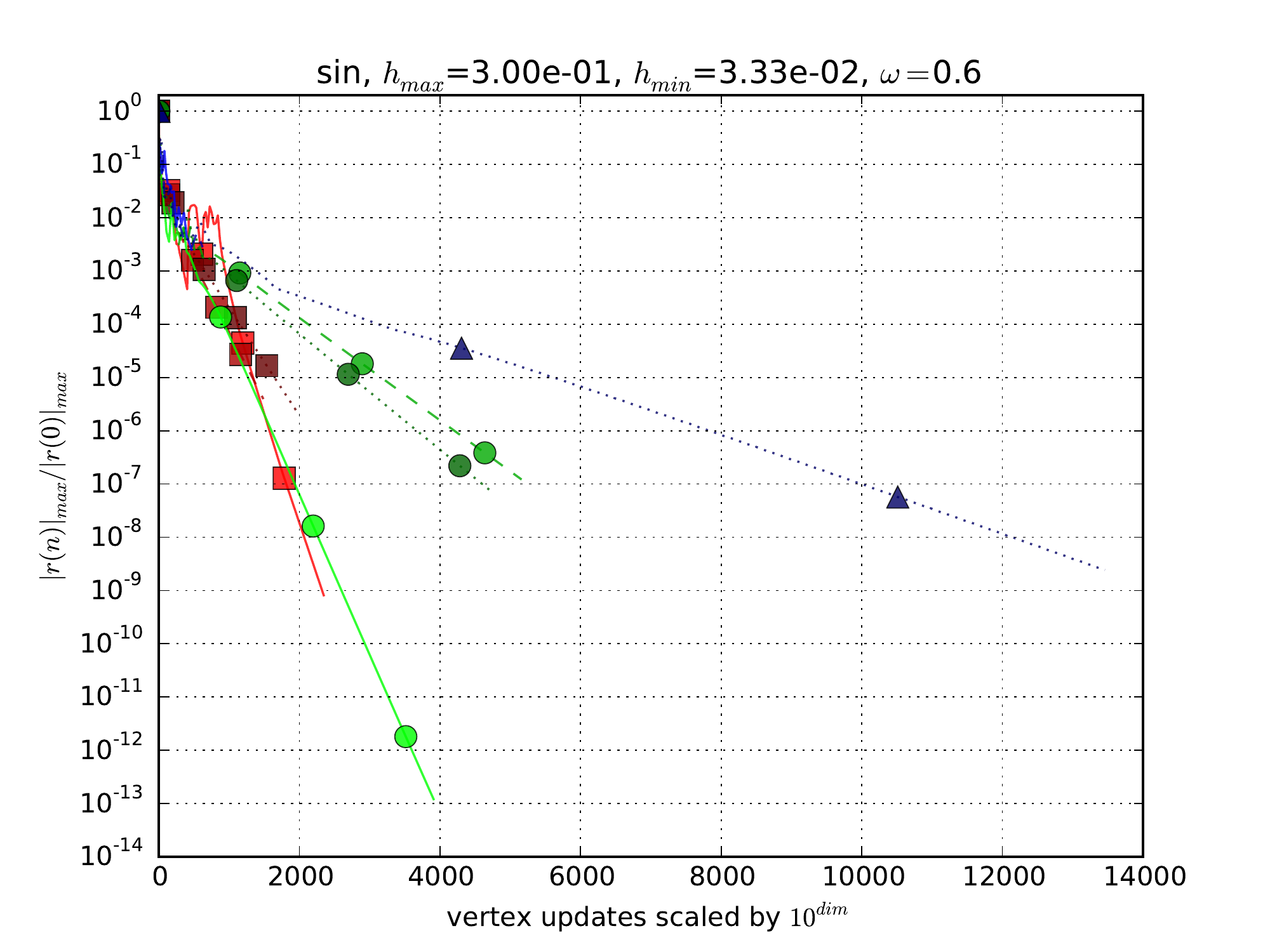}
  \end{center}
  \vspace{-0.2cm}
  \caption{
  Comparisons of the solver efficiency for different dimensions for one
  adaptive mesh configuration choice.
  Note the dimension-dependent scaling of the $x$-axis.}
  \label{figure:08a:high-dimensional}
\end{figure}

We finally observe that optimal, i.e.~mesh-independent, convergence is obtained
for different dimensions if $\omega = 0.8$.
While larger values make the solver diverge, stronger damping such as
$\omega = 0.6$, e.g., does not deliver comparable performance robustly
though some mesh choices benefit from reduced factors (Figure
\ref{figure:08a:high-dimensional}).
Again, exponential damping and transition are the only schemes that are stable
for all $p\in \{2,3,4\}$ and all mesh size configurations.
Transition usually is faster than exponential damping.
However, the convergence speed of both approaches continues to deteriorate with
increasing $p$.
Reasons for this might be found in the poor smoother and the rather aggressive
coarsening by a factor of $3^p$.

We summarise that our approach works for $p\geq 2$, but remains not
practical for $p\geq 5$ due to the curse of dimensions.
To the best of our knowledge, even results for a numerical solution with $p\geq 3$ in our
application area are rare. 
The aforementioned 
 \cite{Baertschy:01:ThreeBodySparseLinAlgebra,Vanroose:05:Science,HMRMM07,Zubair:12:channels,Cools:14:FFM}
 for example all restrict to $p \in \{2,3\}$, i.e.\ two-dimensional or
 three-dimensional grids.
 To the best of our knowledge, very few spacetree codes offer
 four-dimensional or even higher-dimensional dynamically adaptive grids.

\subsection{Case $\phi<0$: positive definite Helmholtz problems}

%
%
\begin{table}
  \begin{center}
    \tbl{
      Cost to reduce the initial residual by $10^{-6}$. The first entry is the 
      number of grid sweeps required by the adaptive solvers, the second
      entry normalises the required unknown updates: It shows how many
      regular grid sweeps yield the same cost.
      \texttt{ucg} denotes undamped coarse grid correction, \texttt{exp}
      exponential coarse grid damping and \texttt{t} transition. 
      $p=2$. The upper part shows results for $\eqref{eq:results:sin}$, the
      lower for \eqref{eq:results:helmhomo}.
      $\bot$ denotes divergence, $^*$ setups where the adaptivity criterion had
      not created stationary grid setups yet.
      \label{table:results:iterations:negative-phi}
    } 
    {    
    \setlength{\tabcolsep}{2pt}
    
    \footnotesize
    \begin{tabular}{l@{}r|ccc|cc|ccc}
      \hline
      &           & \multicolumn{3}{c}{additive multigrid}
      & \multicolumn{2}{|c}{\texttt{hb}} & \multicolumn{3}{|c}{\texttt{bpx}}
 \\
      $\phi$ & $\ h_{min}$ & \texttt{ucg} & \texttt{exp} & \texttt{t} & \texttt{ucg} & \texttt{exp} & \texttt{ucg} & \texttt{exp}
      & \texttt{t}
 \\
 \hline  
       -400 & 0.05   & 32/27.6 & 38/32.8 & 35/30.2 & 183/157.8 & 197/169.9 & 31/26.7    & 35/30.2 & 32/27.6 \\ 
     & 0.005  & $\bot$  & 67/7.7  & 70/10.6 & 370/44.5  & 517/55.9  & 49/8.9     & 74/11.3 & 69/10.5 \\ 
     & 0.0005 & $\bot$  & 95/8.0  & 97/7.2  & $\bot$    & $\bot$    & 65/$^*$     & 93/6.3  & 95/6.7 \\ 
-4000 & 0.05  & 32/27.6 & 38/32.8 & 35/30.2 & 183/157.8 & 197/169.9 & 31/26.7    & 35/30.2 & 32/27.6 \\ 
     & 0.005  & $\bot$  & 67/7.7  & 70/10.6 & 370/44.5  & 517/55.9  & 49/8.9     & 74/11.3 & 69/10.5 \\ 
     & 0.0005 & $\bot$  & 95/8.0  & 97/7.2  & $\bot$    & $\bot$    & 65/$^*$     & 93/6.3  & 95/6.7 \\ 
 \hline 
-200 & 0.05   & 46/39.7 & 29/25.0 & 28/24.1 & 254/219.0 & 332/286.3   & 29/25.0 & 29/25.0 & 28/24.1 \\ 
     & 0.005  & $\bot$  & 64/3.7  & 60/3.4  & 544/4.3   & $\geq 1000$ & 61/3.5  & 63/3.4  & 62/0.4 \\ 
     & 0.0005 & $\bot$  & 196/$^*$ & $\bot$  & $\bot$    & $\bot$      & 107/$^*$ & 71/0.9  & $\bot$ \\ 
-1000 & 0.05  & $\bot$  & $\bot$  & $\bot$  & $\bot$    & 651/561.3   & 18/15.5 & 18/15.5 & 17/14.7 \\ 
     & 0.005  & $\bot$  & 46/2.6  & 61/3.5  & $\bot$    & $\geq 1000$ & 50/2.3  & 73/0.4  & 61/2.9 \\ 
     & 0.0005 & $\bot$  & $\bot$  & 115/$^*$ & $\bot$    & $\geq 1000$ & 58/0.1  & $\bot$  & 101/0.2 \\ 
-4000 & 0.05  & $\bot$  & $\bot$  & $\bot$  & $\bot$    & $\bot$      & 26/22.4 & 33/28.5 & 31/26.7 \\ 
     & 0.005  & $\bot$  & $\bot$  & $\bot$  & $\bot$    & $\geq 1000$ & 36/1.3  & 46/1.7  & 76/2.7 \\ 
     & 0.0005 & $\bot$  & 148/$^*$ & 181/$^*$ & $\bot$    & $\geq 1000$ & 97/$^*$  & 188/$^*$ & 243/$^*$ \\

     %
    \end{tabular}
    }
  \end{center}
\end{table}

Negative $\phi$ in \eqref{eq::pHelmholtz} with the right-hand side from
\eqref{eq:results:sin} have a negligible impact on the convergence behaviour of
the solvers (Table \ref{table:results:iterations:negative-phi}; upper part).
Here, additive multigrid and BPX yield comparable results.
BPX seems to become superior for sufficiently fine grids.
The hierarchical basis approach is slower.
For BPX, undamped coarse grid damping is the method of choice.
As we stop the study for rather big residuals being in the order of
$10^{-6}$ relative to the start residual, the multigrid's transition scheme has
not yet overtaken the exponential coarse grid damping.
We observe for all setups that finer mesh resolutions require more sweeps.
The grid has to unfold completely due to these sweeps.
The cost (in terms of unknown updates) normalised by the cost per sweep on a
regular grid of the finest mesh size however decreases with additional levels.
Inaccuracies occur for BPX that stopped right after the refinement
criterion had inserted an additional grid level. 
It thus does not make sense to compare the number of updates to a regular
grid---the new level that just had been inserted makes the adaptive scheme seem
to be too good.
Longer simulation runs/a lower termination threshold would put these values into
perspective.

The convergence characteristics change for
\eqref{eq:results:helmhomo} acting as right-hand side (Table
\ref{table:results:iterations:negative-phi}; lower part). 
Additive multigrid starts to diverge for coarser mesh sizes already as $\phi<0$
becomes smaller.
For fine meshes, it always diverges.
The ill-behaviour stems from the fact that our coarse grid updates mimic a long-range diffusive behaviour of the solution.
The smaller $\phi<0$ the less significant this diffusion component in
\eqref{eq::pHelmholtz}.
Instead, we face steep gradients at the transition of the right-hand.
They can not be resolved on coarse grids.
Even worse, any coarse grid change pollutes a fine grid approximation due to
unnatural diffusion introduced by our $p$-linear $P$.
It thus excites oscillations around the $\chi$ transition.

The hierarchical basis is more robust w.r.t.~these non-diffusive oscillations
if we use exponential damping.
However, its convergence speed deteriorates.
Exponential damping in combination with the \texttt{hb}-filtering of c-points on
the fine grids yields a scheme where unknowns within the computational
domain that are induced by coarse levels are not updated significantly anymore
once finer grids are introduced; or once restricted residuals average out on coarse levels.
Any combination of \texttt{hb} without full coarse grid addition of the
correction thus makes only limited sense and is not followed-up further.
Results for \texttt{hb} with transition are not even shown.
\texttt{hb} seems to be a problematic solver variant here.
It is due to the additive framework and might be completely different for
multiplicative settings.

Our BPX-type variant finally yields the best results.
BPX starts to reduce the cost per accuracy for shrinking $\phi$s unless its hits
convergence (at a cost of around 0.1 regular fine grid sweeps). 
Hereby, an undamped coarse grid addition is superior to the other variants;
as it materialises in the number of sweeps.
As the PDE deteriorates to an explicit equation $-\phi \psi = \chi$ with a relatively small
diffusive addendum on the left-hand side, 
fine grid unknowns introduced by coarse grid levels are updated almost to the
right solution immediately due to the dominance of the diagonal in the system
matrix.
The residual for other unknowns is (almost) correct on the finest grid
resolution, too.
Where the multigrid and hierarchical basis update the latter points plus add a
prolonged correction from the c-points---the latter update component over-relaxes
the unknowns---BPX explicitly removes the coarse grid contribution, as the
coarse grid contribution equals exactly the fine grid update.

\subsection{Case $\phi>0$: indefinite Helmholtz problems on complex rotated grids}

Depending on the grid resolution and topology,
\eqref{eq::pHelmholtz} can become indefinite
or yield a non-trivial null space for $\phi>0$.
Outgoing waves mimicked by absorbing boundary layers 
add natural damping to most of the eigenmodes and therefore make the
discretised operator better conditioned.
In the present section, we stick to homogeneous Dirichlet boundary conditions
only, i.e.~study a worst-case scenario regarding the numerical stability.
In turn, we isolate the impact of complex grid rotation from any other damping
induced by outgoing wave boundary conditions in typical applications.


In the following experiments, the mesh width is scaled according to $h\to
he^{i\theta}$ in each dimension ($0\le\theta\le\frac{\pi}{4}$).
Complex rotation has a positive effect on the convergence behaviour of the
multigrid solver (cmp.~Figure~\ref{figure:08b:convergence-thetas:phi2025}
with $\phi=45^2$), while $\theta=0\degrees$ makes the solver diverge.
$\omega =0.8$ is used throughout all experiments. 
If we rotate the constant mesh width over $\theta > 30\degrees$, then the solver gets into the regime of convergence. 
We observe the same behaviour on an adaptive grid that starts on a coarse
regular grid with $h_{min}=\frac{1}{9}$. 
During a refinement step the residual reduction might temporarily increase.
Once the grid settles into a steady state, an asymptotic convergence rate is
reached.
It is better than the corresponding convergence rate on a regular grid with
the same minimal mesh size. This is due to a reduced set of eigenvalues.

\begin{figure}
  \begin{center} 
  \includegraphics[width=0.45\textwidth]{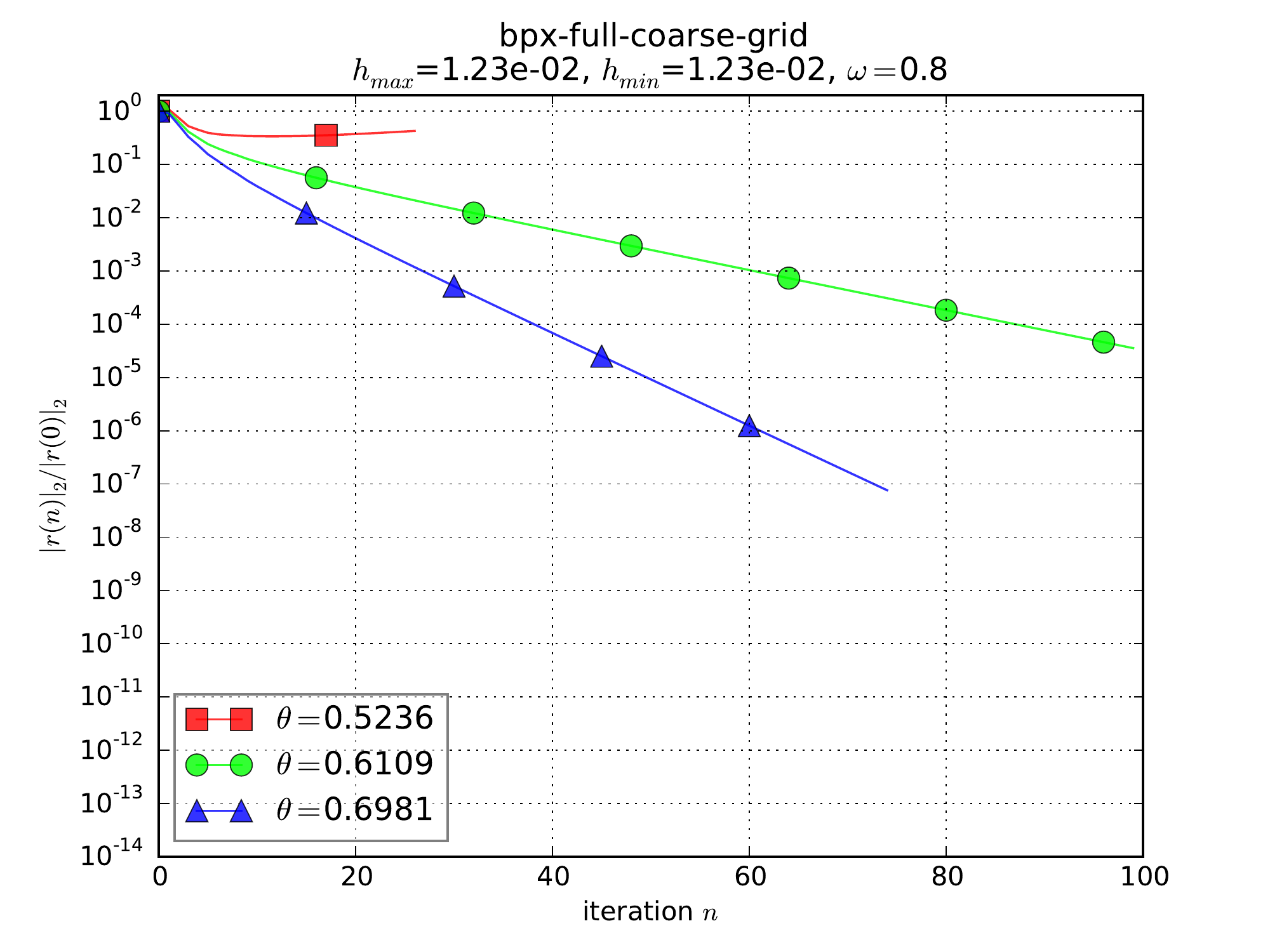}
  \includegraphics[width=0.45\textwidth]{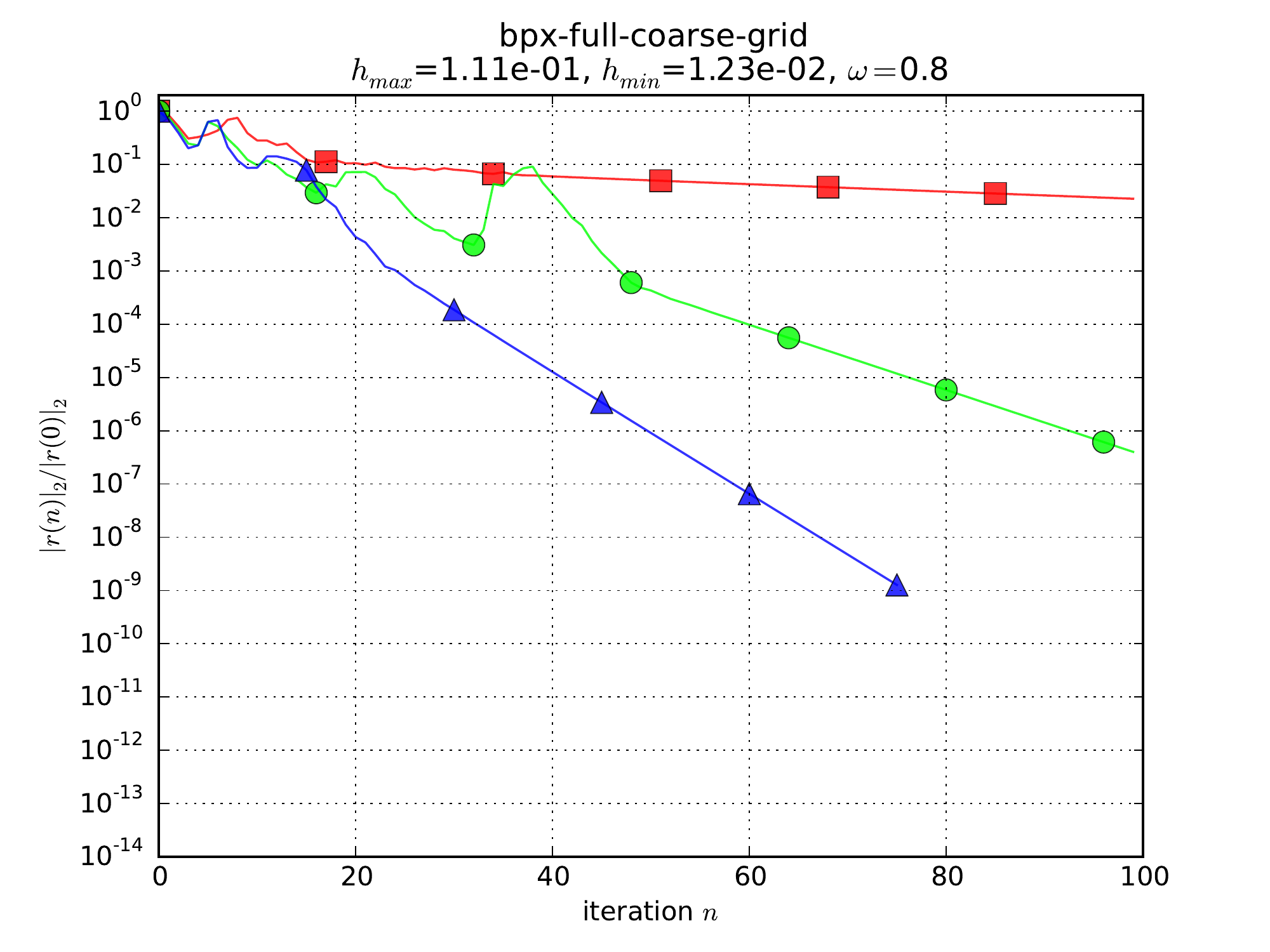}
  \end{center}
  \vspace{-0.2cm}
  \caption{Residual development for the BPX-variant in Algorithm~\ref{algo::tdbpx} with undamped coarse grid correction for $\phi = 45^2$. Different values of the complex rotation angle
  $\theta=30\degrees, 35\degrees, 40\degrees$ were tested on a regular grid (left) and an adaptive grid (right).}
  \label{figure:08b:convergence-thetas:phi2025}
\end{figure}

A typical stress test for indefinite Helmholtz solvers is the robustness as a function of increasing Helmholtz shift $\phi = k^2$. $k$ is the \emph{wave number}.
For one-dimensional problems a common restriction on the mesh size $h$ is given
by the \emph{ten-points-per-wave-length} rule, which translates into $kh <
0.625$. In the following two-dimensional experiments we keep $kh=\frac{5}{9}=0.5555\ldots$ and
test for different values of $k$ on regular grids. Note that from a physical
point of view a more stringent constraint on $h$ is required in higher dimensions,
such as $k^3h^2 = \mathcal{O}(1)$, in order to avoid pollution of the solution \cite{Bayliss:85:pollution,Ihlenburg:95:pollution2}.

\begin{figure}
\begin{center}
\includegraphics[width=0.45\textwidth]{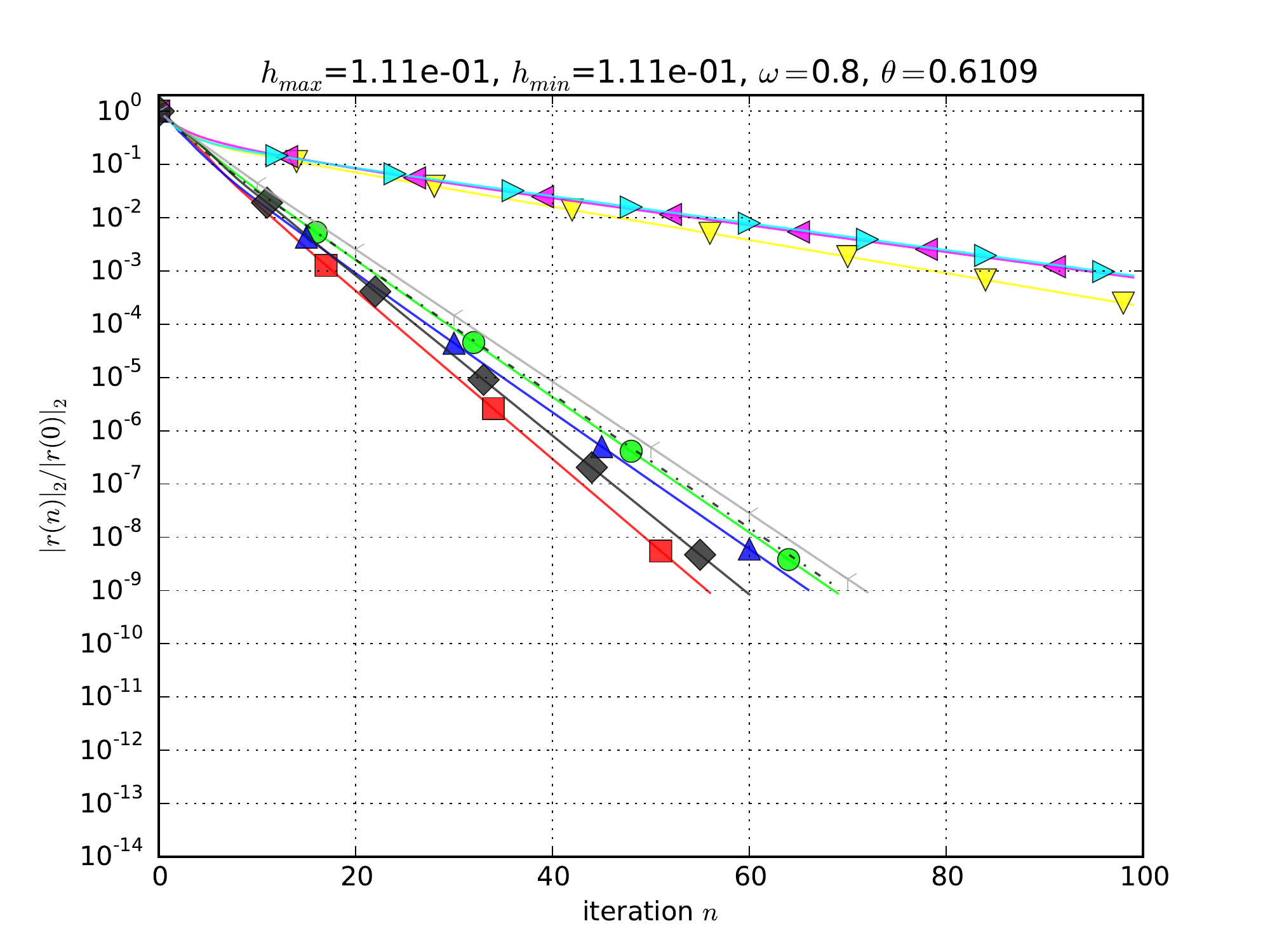}
\includegraphics[width=0.45\textwidth]{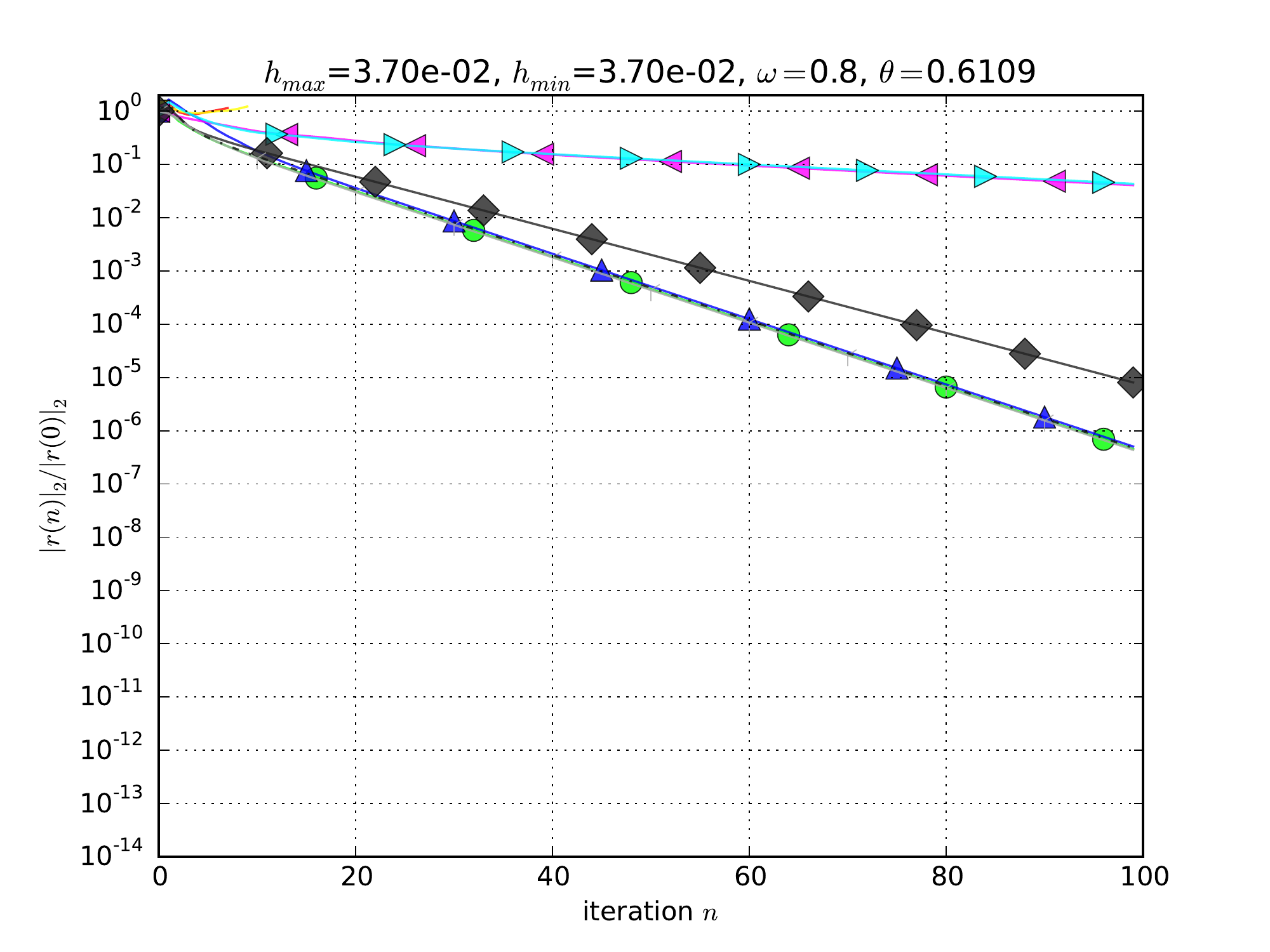}
\\
\includegraphics[width=0.45\textwidth]{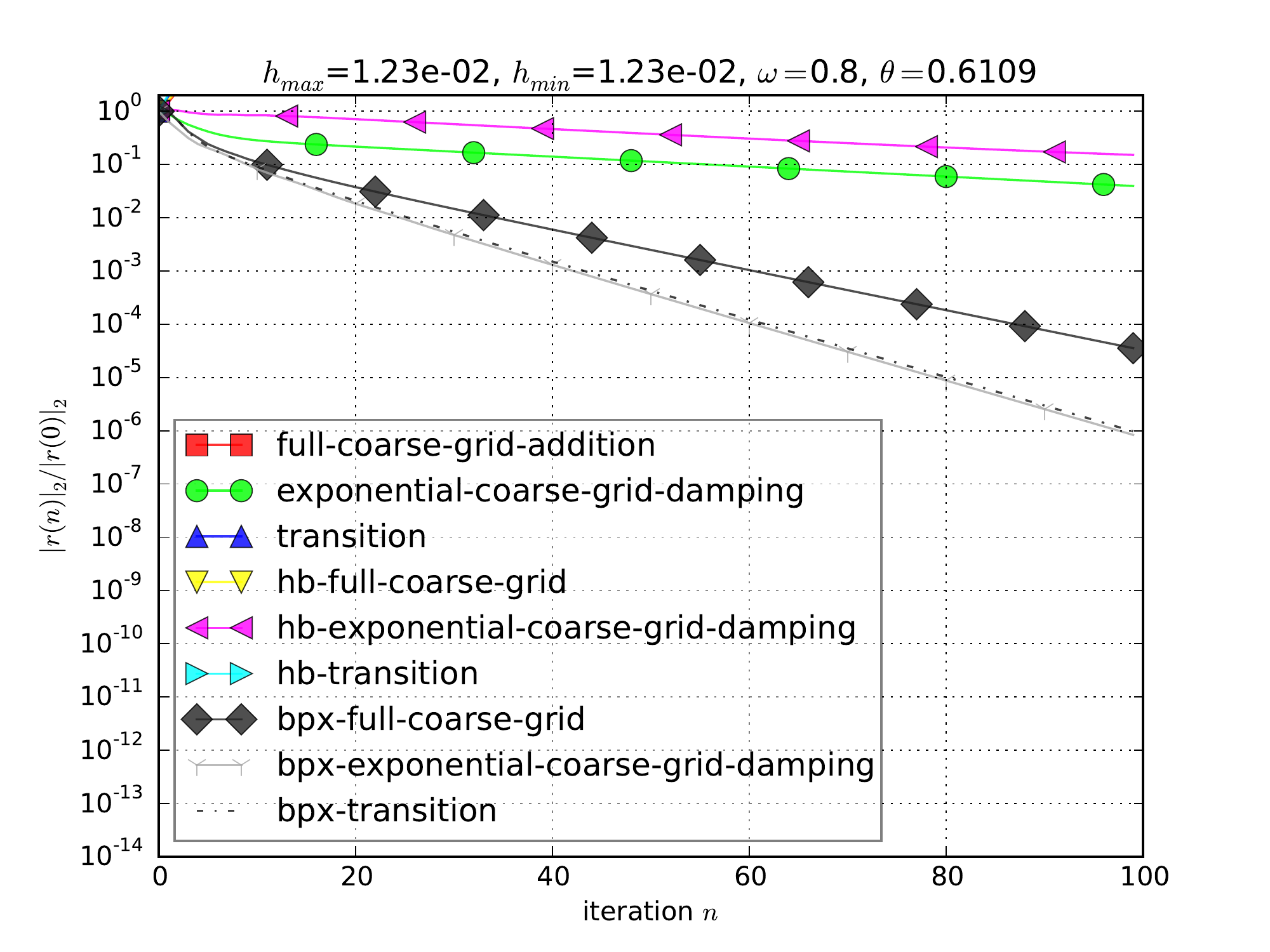}
\includegraphics[width=0.45\textwidth]{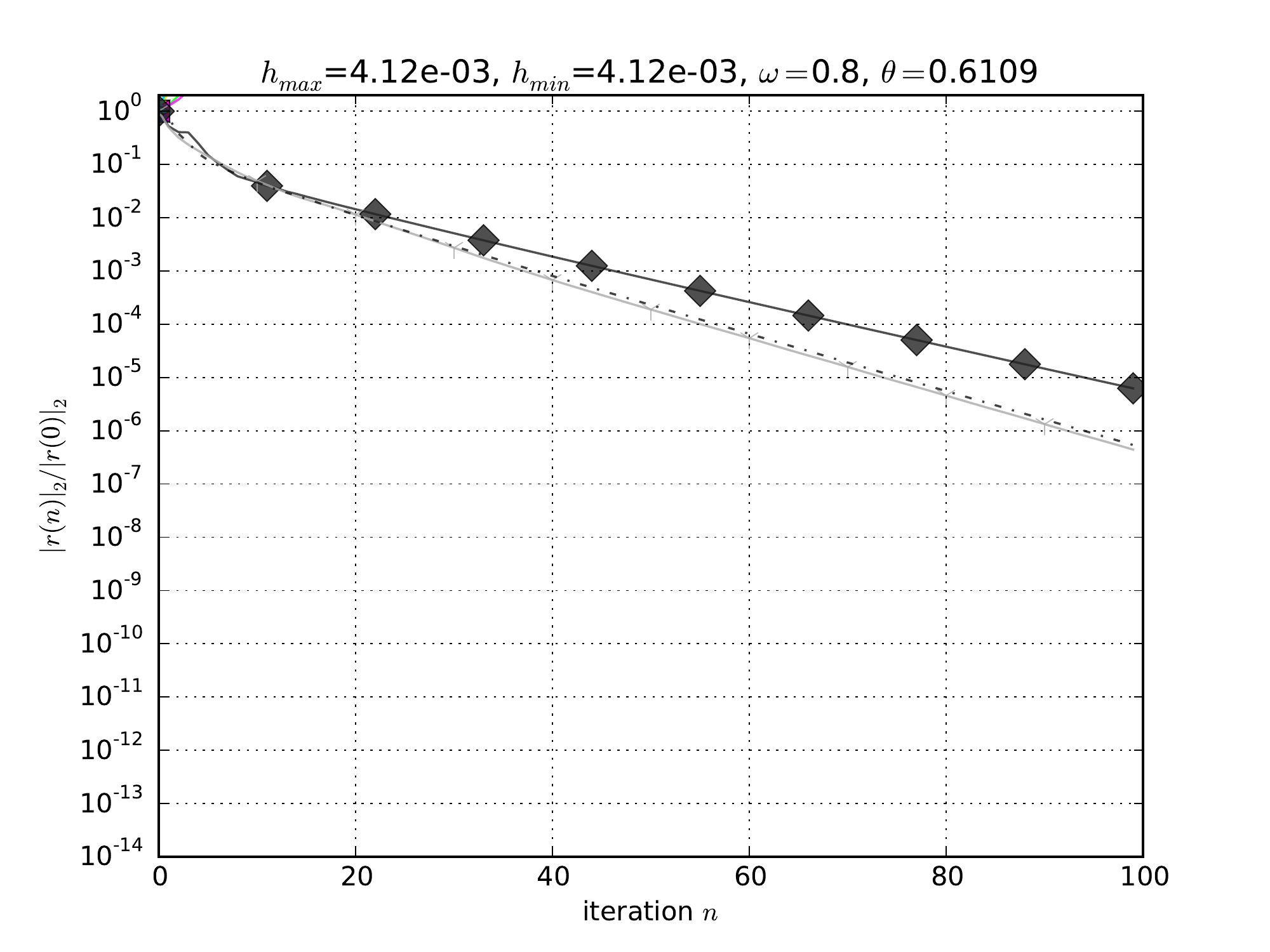}
\end{center}
\vspace{-0.2cm}
\caption{
  Residual development for different values of the Helmholtz shift: $\phi=5^2$
  (top left), $15^2$ (top right), $45^2$ (bottom left) and $135^2$ (bottom
  right). Complex rotation was set to $\theta = 35\degrees$.
  The \texttt{hb-} prefix marks hierarchical basis solvers, the \texttt{bpx-}
  prefix a BPX approach.
  \label{figure:08b:convergence-shifts:theta35}
}
\end{figure}

In Figure~\ref{figure:08b:convergence-shifts:theta35} the convergence behaviour of different solver variants is compared, similar to the experiments in Figure~\ref{figure:08a:convergence-regular} for the Poisson equation. The mesh width is now complex rotated over $\theta=35\degrees\approx0.6109$ in order to avoid divergence due to a non-trivial null space. The top left panel ($k=5$) shows a nice reduction rate. For increasing wave numbers ($k=15, 45, 135$), the term $\phi=k^2>0$ starts to dominate the PDE and the solvers suffer from the overshooting effects discussed in the previous section with $\phi<0$. As expected, only the BPX-variants seem to cope with the highest wave number $k=135$ in the bottom right panel.

\begin{figure}
\begin{center}
  \includegraphics[width=0.45\textwidth]{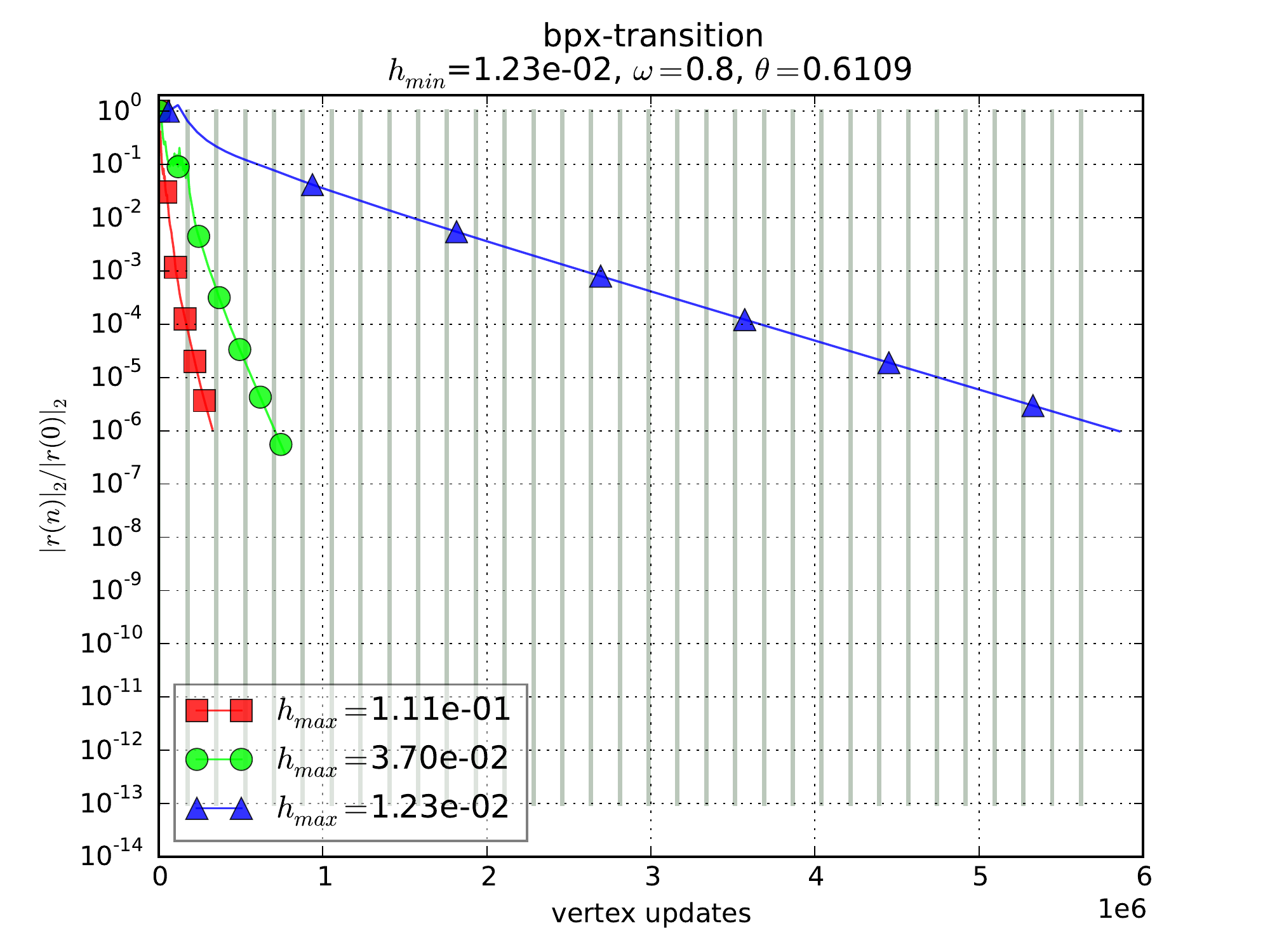}
  \includegraphics[width=0.45\textwidth]{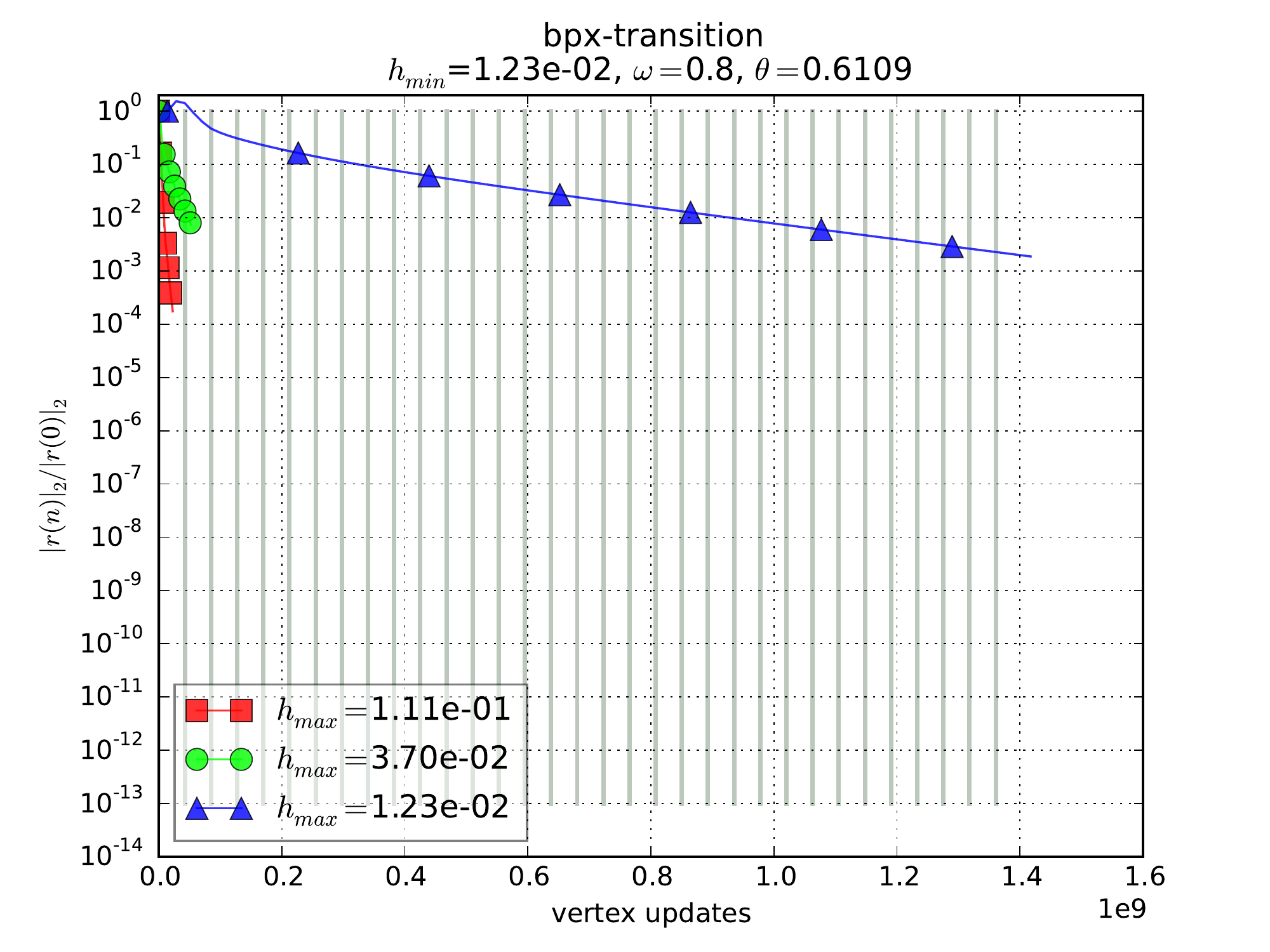}
\end{center}
\vspace{-0.2cm}
\caption{
  Adaptive grid with BPX additive multigrid that is successively refined from a
  prescribed maximum mesh size $h_{max}$ to the minimum mesh size $h_{min}$ for
  $p=2$ (left) and $p=3$ (right). Helmholtz variant of the experiment in
  Figure~\ref{figure:08a:convergence-adaptive}, for $\phi=45^2$ and complex
  rotation $\theta = 35\degrees$.
  \label{figure:08b:convergence-shifts-adaptive:theta35}
}
\end{figure}

The BPX-variants are further tested on adaptive grids in Figure~\ref{figure:08b:convergence-shifts-adaptive:theta35} for values $k=45$ and $135$. We let the corresponding mesh sizes now determine the finest possible resolution in the grid, and start each solve on a coarse regular grid with $h_{max} = \frac{1}{9}$. The result is an FMG-type solver that adaptively refines. The same as for the Poisson experiments in Figure~\ref{figure:08a:convergence-adaptive}, the horizontal axis indicates the number vertex updates versus the residual normalised residual norm on the vertical axis. Again, there is a significant benefit from coarser regions in the grid.
Adaptivity is a particularly useful functionality for the highly heterogeneous Helmholtz equations with space-dependent function $\phi=\phi(\vec{x})$, 
that arise in the quantum mechanical problems described in Section~\ref{section:model-problem}. 
A typical L-shaped refinement \cite{Zubair:10:Lshaped} is desirable for a good representation of extremely localised waves close to the domain boundary (cmp. Figure~\ref{fig::twofixedparticles} right). 
Nonetheless, there remain severe numerical stability issues to handle these evanescent waves. We refer to the concluding Section~\ref{section:conclusion} for an outlook on strategies, as this lies beyond the scope of the current paper.

\subsection{$p=2$ application scenario and grid adaptivity structure}
\label{section:results:gaussian}

With characteristics of the solver behaviour at hand, we finally study a
realistic channel matrix for $p=2$ idealised from the dynamics of Hydrogen or
Deuterium \cite{Zubair:12:channels}.
This leads in the channels' frequency domain to two-dimensional Helmholtz
problems 

\begin{eqnarray}
  \chi (x,y) & = & e^{\left( - \left( 125 x \right)^2 -\left( 125 y \right)^2
  \right)} \quad \mbox{and} \nonumber \\
  \phi (x,y) & = & 45^2 + 135^2 \cdot \left( e^{-\left( 15 x \right) ^2} + e^{-\left( 15 y \right) ^2}\right),
  \label{equation:08-gaussian}
\end{eqnarray}

\noindent
where the x-axis and the y-axis represent the distance from the centre of mass.
They consequently carry homogeneous Dirichlet values: The probability for an
electron to coincide with a nucleus equals zero.
The remaining two faces top and right are open boundaries.
They are consequently supplemented with homogeneous Dirichlet values as well,
but we rotate the cells close to these faces by $30^\circ $ in the complex
domain to eliminate wave reflections.
Close is $\frac{1}{3}$rd of the domain.
The remaining cells in the domain are complex rotated by an angle $\theta
\geq 0$, independently of this absorbing layer. As discussed in the previous
section, only for $\theta=0$ the original Helmholtz problem is solved.
Both an inhomogeneous right-hand side as well as an
inhomogeneous `material' function $\phi > 0$ are active inside the domain. The actual system parameters are channel dependent and are determined by the potential fields of all present particles.

For this setup $h=1/81$ is a physically reasonable choice for our purposes of studying the solution. 
The experiment is then sufficiently small to be solved by a direct solver. However, we solve it with pure Jacobi (cmp.~Figure
\ref{fig::twofixedparticles} right) which does not require us to change any
code.
For this, we use two-phase relaxation \cite{Ernst:11:HelmholtzIterativeMethods}, i.e.~two different
relaxation parameters $\omega _1$ and $\omega _2$ alternatingly, which 
follows \cite{Hadley:2005:ComplexOmega}\footnote{
  \cite{Hadley:2005:ComplexOmega} uses a relaxation $\alpha $ that is relative to a finite difference
  discretisation of the Laplacian. In our case, we consequently have $\omega =
  \alpha (1+\frac{h\phi }{4})$ with $\alpha = \sqrt{3}- \i$.
} 
with
\[
  \omega _1 = 0.01 \cdot \left( \sqrt{3}- \i \right) 
  \quad \mbox{and} \quad
  \omega _2 = - \overline{\omega _1}.
\]

\noindent
One of the two relaxation parameters has a negative real part
and both carry an imaginary component. 
Obviously, this approach becomes unfeasible
once $p$ increases or finer grid resolutions are required.
For the additive multigrid (transition scheme) and the BPX (undamped coarse
grid correction), we use $\omega = 0.4$ which avoids oscillations.
We also use $\omega = 0.4$ for all Jacobi smoothers that apply complex
rotations.

\begin{figure}
  \begin{center}
    \includegraphics[width=0.48\textwidth]{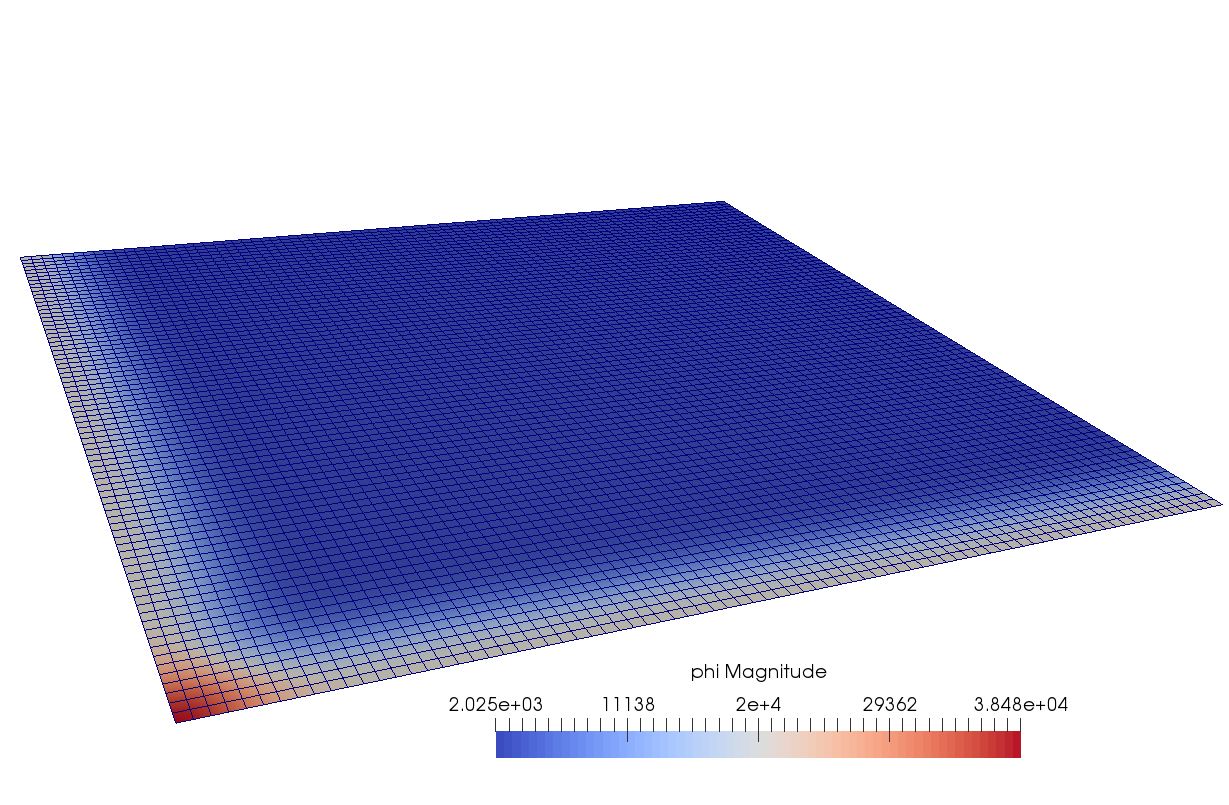}
    \hspace{0.1cm}
    \includegraphics[width=0.48\textwidth]{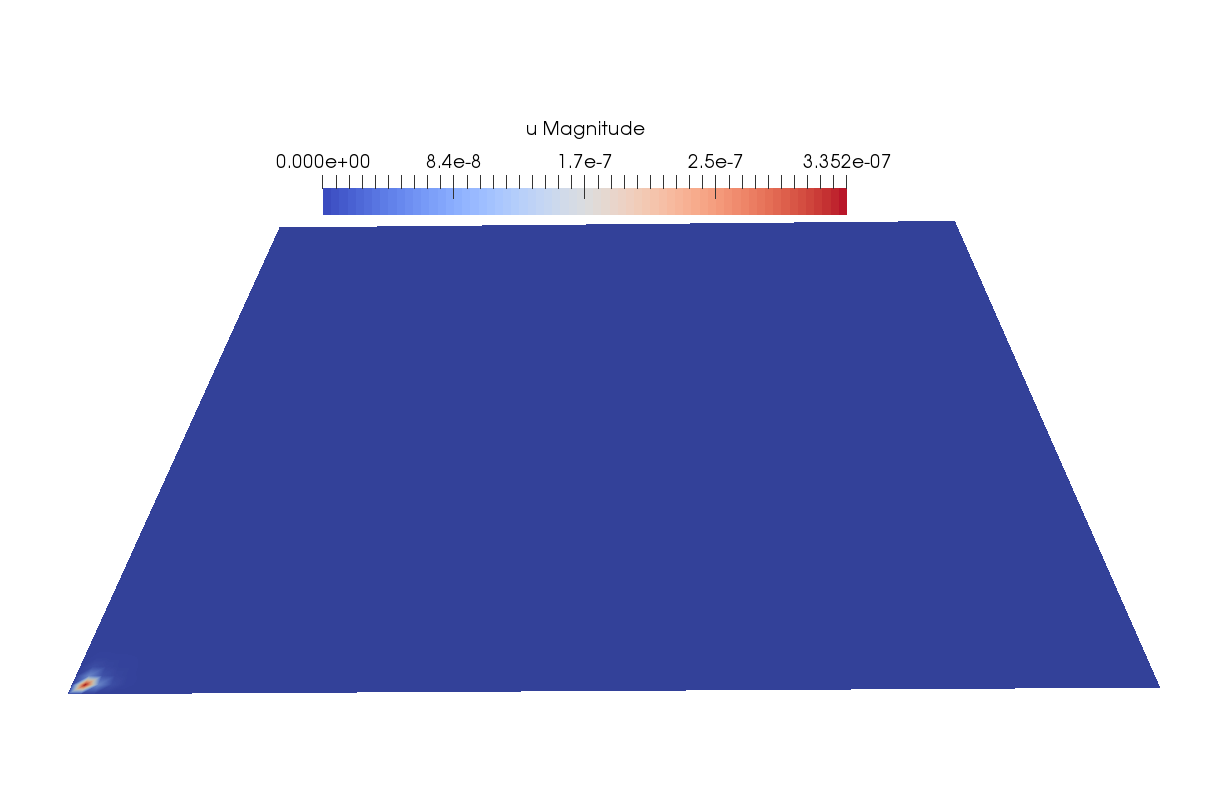}
  \end{center}
  \caption{
    The distribution of $\phi $ in the computational domain (left).
    It is invariant of the complex rotation.
    Complex rotation however does make a difference to the solution (right,
    $\theta = 45^\circ $).
  }
  \label{figure:08b:phi-and-solution}
\end{figure}

\begin{table}
 \begin{center}
   \tbl{
     Reduction of residual of \eqref{equation:08-gaussian} in max norm after 50
     iterations. The regular grid uses $1/81$, the adaptive one starts from
     $h_{max}=0.1$ and does permit the criterion to refine until
     $h_{min}=0.001$ is just underrun. $\bot$ denotes
     divergence.
     Figures in brackets show, if appropriate, the number of vertices used after
     50 iterations.
    \label{table:08gaussian:reduction}
   }
   {\footnotesize   \def\arraystretch{1.5}
   \begin{tabular}{l@{\hspace{0.1em}}|@{\hspace{0.14em}}rrr@{\hspace{0.14em}}|@{\hspace{0.14em}}rrr}
   & \multicolumn{3}{c}{regular} & \multicolumn{3}{c}{adaptive} \\
     $\theta$ & Jacobi & transition & BPX & Jacobi & transition & BPX \\
     \hline
     $0^\circ$ 
              & $4.78 \cdot 10^{-1}$ & $\bot$ & $\bot$
              & $\bot$ & $\bot$ & $\bot$ 
              \\
     \hline
     $18^\circ $ 
               & $\bot $ & $\bot $ & $4.44 \cdot 10^{-2}$
               & $\bot$ & $\bot$ & $\bot$ 
               \\   
     $25^\circ $
               & $2.07 \cdot 10^{-1}$ & $2.27 \cdot 10^{-2}$ & $2.62 \cdot 10^{-3}$
               & $9.82 \cdot 10^{-2}$ (1,536) & $5.30 \cdot 10^{-2}$ (1,448)&
               $5.72 \cdot 10^{-3}$ (1,500) \\
     $35^\circ $
               & $8.46 \cdot 10^{-5}$ & $2.00 \cdot 10^{-4}$ & $8.35 \cdot
               10^{-4}$
               & $9.49 \cdot 10^{-2}$ (1,544) & $3.97 \cdot 10^{-2}$ (1,324) &
               $1.61 \cdot 10^{-3}$ (1,480)
               \\
     $45^\circ $
               & $6.53 \cdot 10^{-7}$ & $6.67 \cdot 10^{-5}$ & $2.60 \cdot
               10^{-4}$
               & $9.14 \cdot 10^{-2}$ (1,524) & $3.09 \cdot 10^{-2}$ (1,308) &
               $1.16 \cdot 10^{-3}$ (1,460)
   \end{tabular}
   } 
 \end{center}
\end{table}

Though Jacobi converges for $\theta =0$ and reduces the residual after every two
grid sweeps, the convergence speed is unacceptably low even for this simple
setup.
However, the additive multigrid and BPX are not stable for $\theta =0$ and thus
can not be applied.
BPX becomes stable for $\theta \geq 18^\circ $, while the additive transition
scheme requires $\theta \geq 25^\circ $ (Table
\ref{table:08gaussian:reduction}).
With increasing $\theta$, the convergence speed of all solvers improves.
This improvement is rendered problematic as the quality of the preconditioner
suffers.
While the latter effect is not directly studied here, it is indirectly
illustrated by the Jacobi eventually outperforming the multigrid schemes.
For large $\theta$, all wave behaviour is damped out and we basically resolve
one peak around the coordinate system's origin (cmp. Figure~\ref{figure:08b:phi-and-solution} right).
The local solution characteristics render multilevel solvers
inappropriate.

If we start from a grid of $h_{max} = 0.1$ and allow the dynamic
refinement criterion to refine any cell coarser than $h_{min}=0.001$, all
solution processes automatically yield an adaptive grid (Figure
\ref{figure:08b:mesh}).
For all rotation choices, we end up with 1,308--1,544 vertices; a significant
saving compared to a regular grid with 6,400 vertices.
The purely feature-based criterion here has a two-fold role: 
In accordance with previous results, it rewrites the solve into an FMG-type
cycle.
At the same time, it allows the code to make the solution resolve the physical
problem characteristics economically.
While these solution characteristics depend sensitively on the choice of
rotation, we observe that the adaptivity structure is almost
rotation-invariant.
This is an effect that deserves further studies but obviously results from the
fact that the gradient around the Gaussian stimulus $\chi$ exceeds by magnitudes
any characteristic of the induced wave pattern.
Different to the regular grid results, we furthermore observe that BPX remains
superior to the other approaches for all $\theta$, while all variants outperform
their regular grid counterparts in terms of cost: each adaptive iteration is at
most 1/4th of the cost of a regular grid sweep.

\begin{figure}
  \begin{center}
    \includegraphics[width=0.3\textwidth]{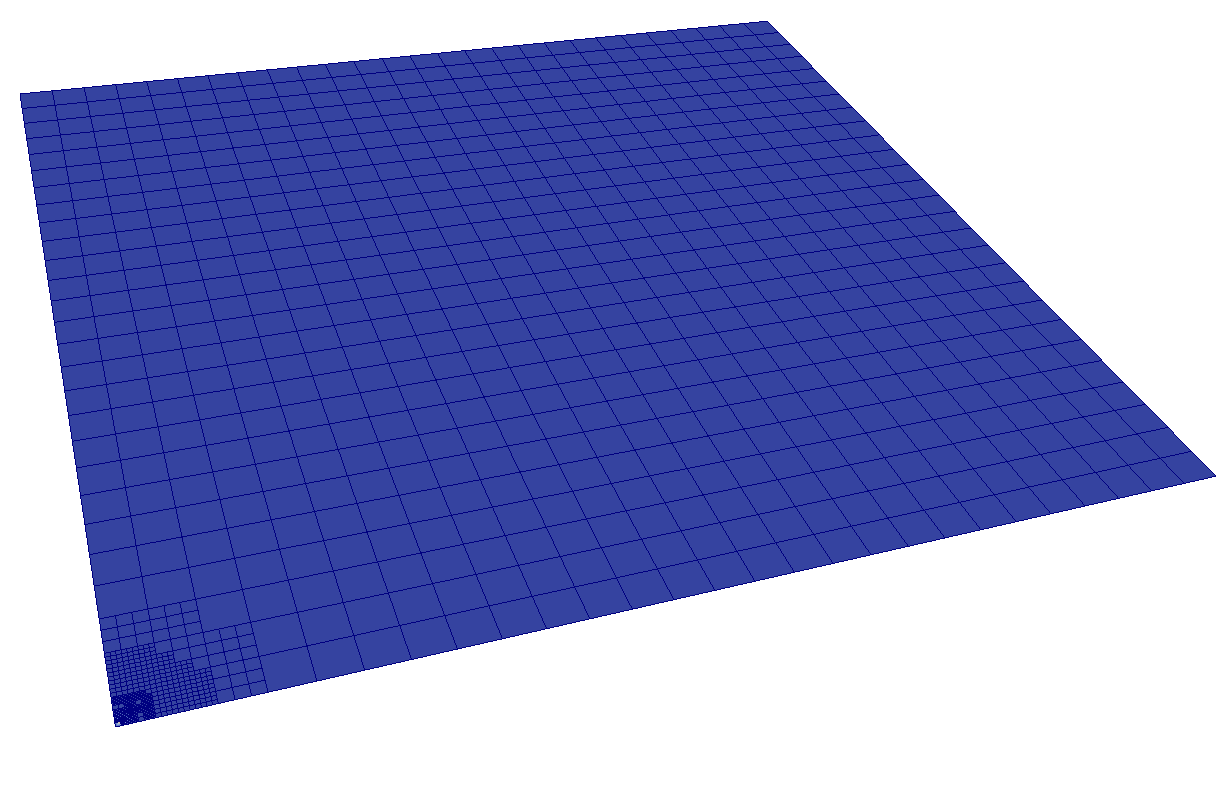}
    \hspace{0.1cm}
    \includegraphics[width=0.4\textwidth]{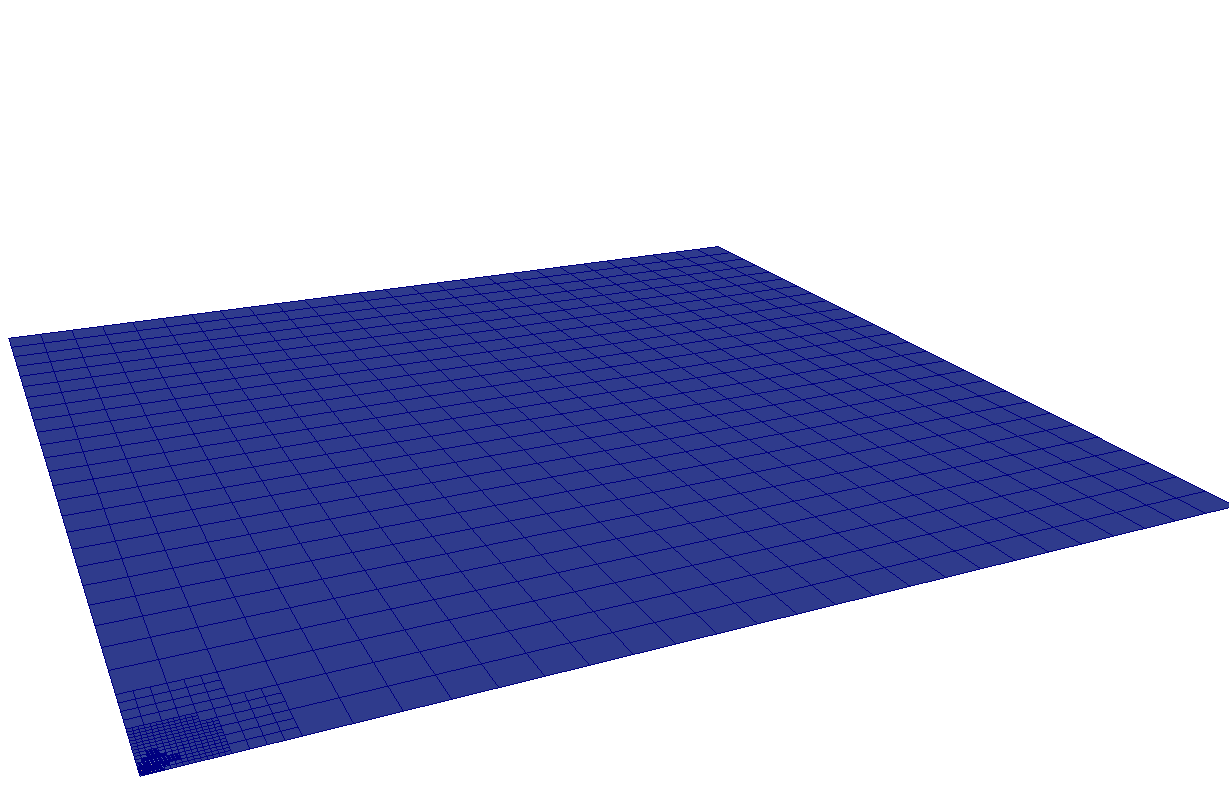}
  \end{center}
  \caption{
    Left: Dynamically adaptive grid with $h_{max} = 1/3$ and $h_{min}=0.001$ for
    a pure Jacobi solver at the moment of convergence without any complex rotation.
    Right: Dynamically adaptive grid of BPX for the same setup with a complex
    rotation of $\theta = 45^\circ $.
  }
  \label{figure:08b:mesh}
\end{figure}

\subsection{Hardware efficiency}

We finally study the algorithm's hardware characteristic and
compare the memory throughput to the Stream benchmark
\cite{McCalpin:95:Stream} ran on a single core of the Sandy Bridge.
An excellent cache usage mirrors results from 
\cite{Weinzierl:2009:Diss,Weinzierl:11:Peano}.
It results from the combination
of strict element-wise data access, stack-based data management and depth-first
spacetree traversal along a space-filling curve (Table
\ref{table:results:efficiency:sb-hw-counter}).
Element-wise formulation and depth-first fit, i.e.~localised data access even
for the grid transfer operations, are characteristics of the present multigrid
algorithms.
L1 and L3 cache measurements yield analogous results.
Basically all required data are always found in the L1 and L2 cache.
The combination of low cache misses with the data usage policy, i.e.~one
traversal per solver cycle and one data read per unknown, highlights that the
present approach is memory modest.
This statement holds independent of $p$.
Our approach is not memory-bound though the grid changes almost each iteration
and the code is a multiscale algorithm.
It however neither exploits the available memory bandwidth which is around
$8.35 \cdot 10^3$ MB/s for the Stream Triad \cite{McCalpin:95:Stream}
benchmark ran on a single core with the same settings, nor does it exploit the
vector registers.
Its arithmetic intensity is very low.
For $p>3$, vectorisation even makes the runtime increase.
We observe that few floating point operations per second face more than
$10^{11}$ total instructions.
The recursive code suffers from significant 
integer arithmetics; from administrative overhead.

\begin{table}
  \begin{center}
    \tbl{Performance counters on the Sandy Bridge for a dynamically
      adaptive FMG-type solve with $h_{min}=0.001$ ($p=2$), $h_{min}=0.005$
      ($p=3$), $h_{min}=0.01$ ($p=4$). The left value is obtained without
      vectorisation, the right results from an executable translated with simd
      facilities.
      L2 misses are relative measures (rates)
      compared to total number of accesses.
      \label{table:results:efficiency:sb-hw-counter}
    } 
    {
    \footnotesize
    \begin{tabular}{c|ccc}
      metric & p=2 & p=3 & p=4\\ 
\hline 
Runtime & 1.16e+02/1.13e+02 & 2.12e+02/2.11e+02 & 2.07e+03/2.31e+03\\ 
Instructions & 5.53e+11/5.53e+11 & 1.11e+12/1.09e+12 & 1.19e+13/1.11e+13\\ 
MFlops/s & 7.27e+02/8.89e+02 & 1.13e+03/1.40e+03 & 1.66e+03/1.98e+03\\ 
L2 requests & 1.77e-02/1.74e-02 & 1.48e-02/1.23e-02 & 9.55e-03/1.04e-02\\ 
L2 misses & 2.75e-03/2.77e-03 & 3.49e-03/3.00e-03 & 6.66e-04/8.29e-04\\ 
Used bandwidth (MB/s) & 2.53e+02/2.52e+02 & 1.35e+02/1.37e+02 & 7.80e+01/7.14e+01\\ 
Transferred memory (GB) & 2.90e+01 & 2.86e+01 & 1.65e+02\\ 

    \end{tabular}
    }
  \end{center}
\end{table}

We reiterate that \eqref{eq::pHelmholtz} is to be
solved for up to $c$ channels simultaneously.
As such, it is natural to make the solver tackle $\hat c \leq c$
problems simultaneously on the same grid.
Each and every unknown associated to a
vertex then is a $\mathbb{C}^{\hat c}$ tuple.
Rather than relying on $c$ independent problem solves, we fuse $\hat c$ problems
into one setup solved on one grid.
We refer to such an approach as {\em multichannel} variant.
It is studied here at hands of five exemplary configurations (different mesh
sizes, initial values, solvers) per $p$ choice.
Results from the Xeon Phi and Sandy Bridge qualitatively resemble each
other though we have to take into account different constraints on the total
memory---for $p=4$, $h_{min}=0.05$ already does not fit into the memory anymore
for $\hat c>8$.

We observe that solving multiple channels on one grid decreases the cost per
unknown monotonously (Figure \ref{figure:08c:multichannel}) up to $\hat c=8$:
the more channels fused into one grid the better the available memory bandwidth
usage.
Again, maintenance overhead is amortised.
This holds despite the fact that the refinement criterion for the fusion of
$\hat c$ channels into one grid has to be pessimistic.
If one channel requires
refinement, all $\hat c$ channels are mapped onto a finer grid. 
As our dynamically adaptive grid starts from few vertices and then refines,
it successively amortises the overhead among the vertices for $p=2$.
The cost per vertex decreases.
For $p=3$ this effect is negligible.
However, both $p$-choices exhibit cost peaks throughout the adaptive
refinement due to additional initialisation effort.
These relative cost are the smaller the more channels are fused.

\begin{figure}
  \begin{center} 
  \includegraphics[width=0.49\textwidth]{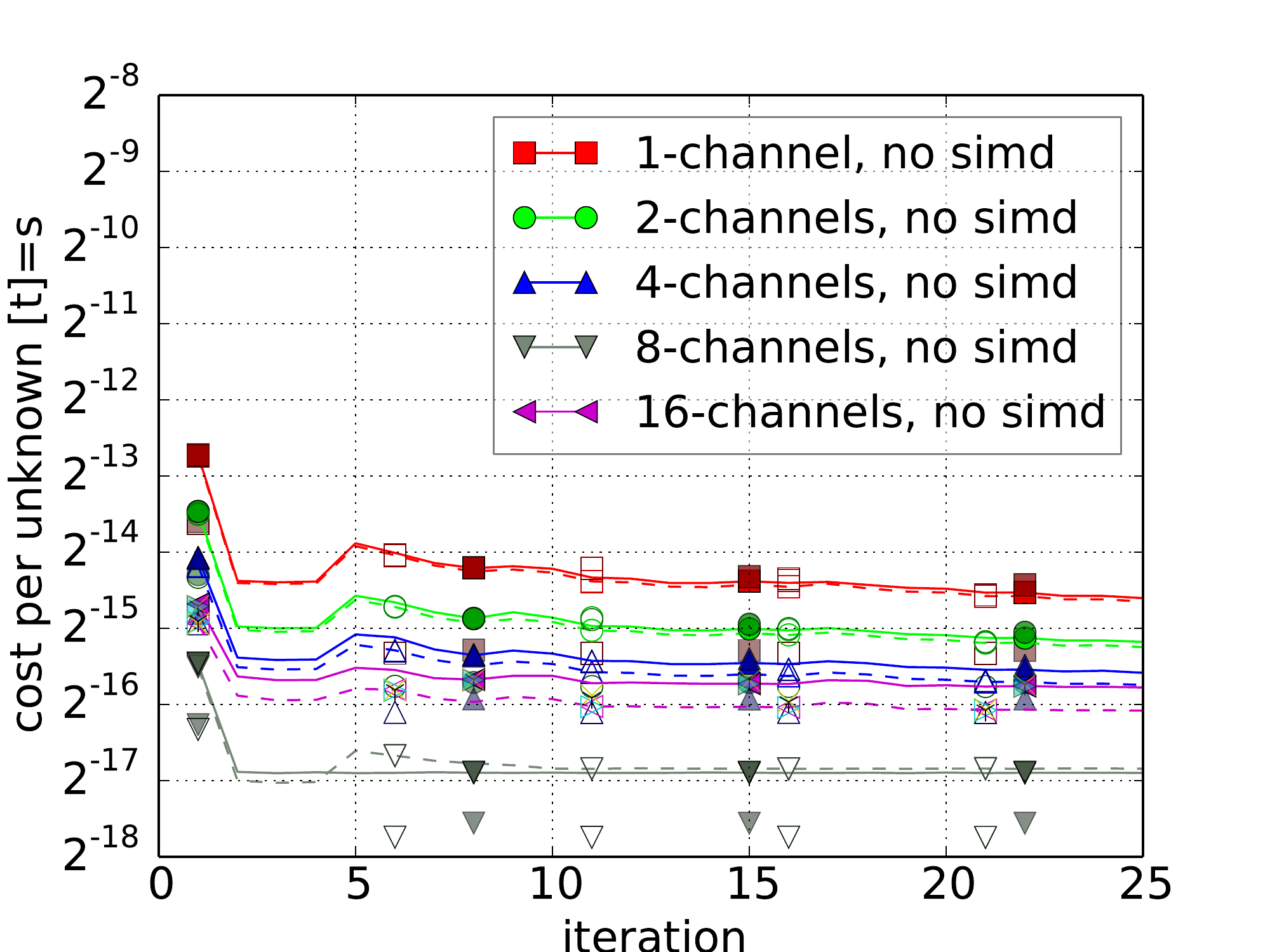}
  \includegraphics[width=0.49\textwidth]{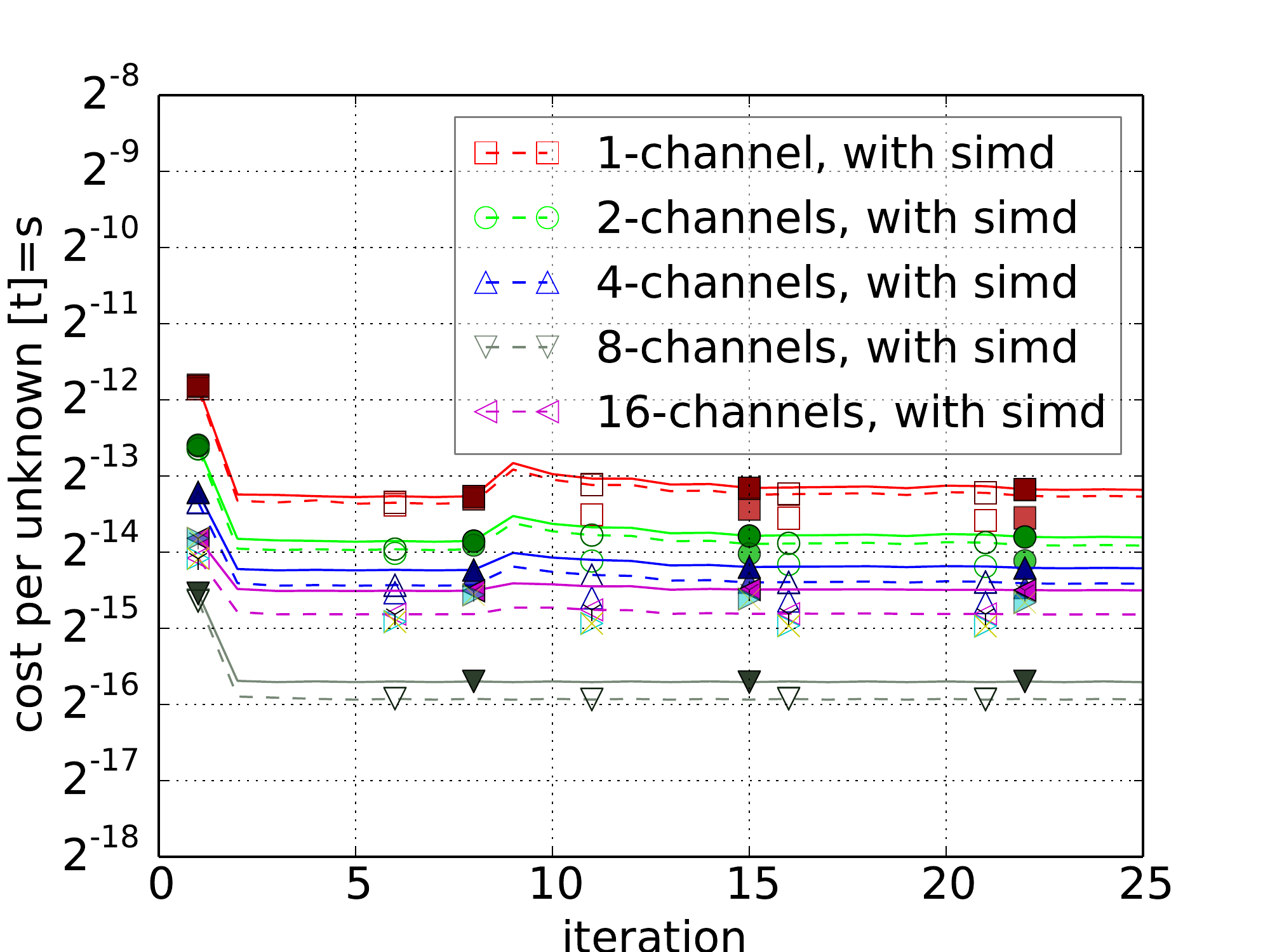}
  \end{center}
  \vspace{-0.2cm}
  \caption{
    Runtime per unknown per multigrid sweep for $p=2$ (left) and $p=3$ (right)
    on the Xeon Phi
    with several multichannel configurations and setups. Worst-case measurements
    with solid/dotted lines.
    \label{figure:08c:multichannel}
  }
\end{figure}

With bigger $\hat c$, the impact of vectorisation increases for $p=3$ (as well
as for $p=4$ which is not shown here) while vectorisation is
negligible or even counter-productive for $p=2$.
Even in the best case, it is still far below the theoretical upper bound. 
This is due to the fact that the multichannel approach does not
change the arithmetic intensity of the compute kernels. 
All efficiency gains stem from an improved vectorisation as the kernels run
through multiple channels per spacetree cell in a stream fashion.

We therefore propose to replace the Jacobi-like splitting of 
\eqref{eq::coupledHelmholtz} by a block-structured decomposition $A\equiv \hat
A + (A-\hat A)$ with
\begin{equation*}
\hat A = 
 \begin{pmatrix}
  H_{11} & A_{12} & A_{13} & \ldots & A_{1\hat c}  \\
  A_{21} & H_{22} & A_{23} & \ldots & A_{2\hat c}  \\
  A_{31} & A_{32} & H_{33} & \ldots & A_{3\hat c}  \\
  \vdots & \vdots & \vdots & \ddots & \vdots       \\  
  A_{\hat c1} & A_{\hat c 2} & \ldots & \ldots & H_{\hat c\hat c}  \\
  & & & & & H_{(\hat c+1)(\hat c+1)} & \ldots & A_{(\hat c+1)(2\hat c)} \\
  & & & & & \vdots & \ddots & \vdots \\
  & & & & & A_{(2 \hat c)(\hat c+1)} & \ldots & H_{(2 \hat c)(2\hat c)} \\
  & & & & & & & & \ddots  
 \end{pmatrix}.
\end{equation*}

\noindent
Though this is a stronger preconditioner and a harder system to solve, it
retains the memory access requirements of the multichannel
variant.
However, such a {\em coupled multichannel} approach with its denser per-grid
entity operators allows us to exploit vector facilities more efficiently since the arithmetic intensity increases.

\begin{table}
  \begin{center}
    \tbl{
      Sandy Bridge performance for different coupled multichannel
      variants. One characteristic adaptive setting per choice of $p$. 
      Vectorisation through Intel pragmas.
      Runtime is not normalised with problems solved simultaneously.
      \label{table:results:efficiency:sb-hw-counter-multichannel}
    } 
    {
    \footnotesize
    \begin{tabular}{c|ccccc}
      $p=2 $ & $\hat c=1$ & $\hat c=2$ & $\hat c=4$ & $\hat c=8$ & $\hat c=16$
      \\
      \hline
Runtime                 & 1.13e+02 & 1.67e+02 & 2.59e+02 & 5.20e+02 & 4.15e+02
\\
MFlops/s                & 8.89e+02 & 1.91e+03 & 3.84e+03 & 6.85e+03 & 1.07e+04
\\
L2 misses               & 2.77e-03 & 3.34e-03 & 4.33e-03 & 5.76e-03 & 7.19e-03
\\
Used bandwidth (MB/s)   & 2.52e+02 & 3.15e+02 & 3.63e+02 & 3.51e+02 & 2.93e+02
\\
      \hline
      $p=3$ \\
      \hline
Runtime                 & 2.11e+02 & 3.75e+00 & 4.56e+02 & 1.02e+03 & 2.78e+03
\\
MFlops/s                & 1.40e+03 & 2.43e+03 & 6.02e+03 & 1.03e+04 & 1.44e+04
\\
L2 misses               & 3.00e-03 & 3.46e-03 & 3.57e-03 & 4.42e-03 & 2.55e-02
\\
Used bandwidth (MB/s)   & 1.37e+02 & 2.45e+02 & 1.78e+02 & 1.70e+02  & 1.32e+02
\\
      \hline
      $p=4$ \\
      \hline
Runtime                 & 2.31e+03 & 3.39e+03 & 6.01e+03 & 3.07e+03 & 6.96e+03
\\
MFlops/s                & 1.98e+03 & 3.87e+03 & 7.05e+03 & 1.08e+04 & 1.55e+04
\\
L2 misses               & 8.29e-04 & 1.09e-03 & 3.63e-03 & 2.05e-02 & 3.01e-02
\\
Used bandwidth (MB/s)   & 7.14e+01 & 7.78e+01 & 7.90e+01 & 1.71e+02 & 1.48e+02
\\
    \end{tabular}
    }
  \end{center}
\end{table}

%
%
An illustration of hardware counter measurements and timings validates and
details these statements (Table
\ref{table:results:efficiency:sb-hw-counter-multichannel}).
For small $1 \leq \hat c \leq 4$, the total runtime increases with an increase
of $\hat c$, but the growth is sublinear.
It is cheaper to compute multiple fused problems than to deploy them to grids of
their own.
For $4 \leq \hat c \leq 16$, the behaviour is non-uniform: While the runtime for
some setups grows linear w.r.t.~problems solved, it sometimes even drops. 
We deduce that this is the sweet region of an optimal choice of $\hat c$. 
The channel fusion really pays off if we fix $\hat c=16$ ($p=2$) or $\hat
c=8$ ($p\in\{3,4\}$).
For $\hat c\geq 16$, the runtime then grows linearly in the number of problems
solved, i.e.~the fusion does not pay off anymore. 
Furthermore, these large equation systems quickly exceed the main memory available on the Xeon Phi.
Besides this constraint, the Phi's figures exhibit similar behaviour. 
For all setups, the bandwidth requirements do not increase significantly, and
the cache misses remain negligible.
While we cannot explain the drops in the execution times, we defer
from the multichannel experiments that the runtime anomalies have to result from
an advantageous usage of the floating point facilities---without a coupling, $\hat c=8$ is the sweet
spot for all choices of $p$.

Besides the aforementioned amortisation of administrative overhead, the coupling
with its increased matrices allows us to exhibit the vector facilities. 
We vectorise and reduce the relative idle times of the vector units
compared to the total runtime.
Both effects in combination improve the MFlop rate up to a factor of
almost ten.
This is around 20\% of the theoretical peak though the measurements comprise all
grid setup and management overhead.
The depth-first traversal in the spacetree with its strictly local
element operators allows us to achieve this without increased pressure on the
memory subsystem.
Due to the matrix-free rediscretisation approach, we can assume that the whole
blocks of $\hat A$ reside within the cache.
Due to the single touch policy per unknown and the localised traversal, we can
assume that each unknown is loaded into the cache exactly once.
This insight mirrors previous reports on this algorithmic paradigm
\cite{Mehl:06:MG,Weinzierl:2009:Diss,Weinzierl:11:Peano}.
We validate this at hands of the cache miss rate explaining the low bandwidth
requirements.

Though the multichannel approach reduces the total
concurrency, the channel block solves remain perfectly parallel.
Due to the memory characteristics, we can expect that these block solves can run
independently on the cores of our Sandy Bridge chip as one core does not
consume more than $\frac{1}{16}$ of the available bandwidth.
Due to the smaller main memory and the bigger core count, these multicore
predictions do not apply unaltered to the Xeon Phi.

  \section{Conclusion and outlook}
\label{section:conclusion}

The present paper introduces software and algorithms to realise a
family of dynamically adaptive additive multigrid solvers for Poisson and Helmholtz problems
working on complex-valued solutions.
The latter provides a straightforward way to realise complex grid rotation. 
Supporting the approach's elegance and clarity, we provide correctness
proofs, we demonstrate its robustness, and we discuss its runtime behaviour with
respect to convergence speed and single-core utilisation. 
As the underlying algorithmic framework makes the solve of multiple Helmholtz
problems or sets of problems perfectly parallel
\cite{Foster:95:ParallelPrograms}, we focus on an algorithm that realises a
single-touch policy---each piece of data is used only once or twice per
traversal and the time in-between two usages is small---and strict localised
operations. It does not stress the memory subsystem.
The experimental data reveals that this objective is met, while the solvers are
reasonably robust and have a small memory footprint.

Besides an application to the underlying problems from physics and chemistry, we
notably identify three methodological extensions of the present work.
The first extension area tackles two obvious shortcomings of the present
code base: 
smoothers and algebraic multigrid operators
\cite{Chen:12:DispersionMinimizing,Stolk:15:DispersionMinimizing,Tsuji:15:AMGPrecon}.
While work in this area is mandatory to facilitate the application of the
algorithm and software idioms to more challenging setups, no fundamental risk,
to the best of our knowledge, exists that these objectives can not be met.
Several fitting building blocks are in place and have to be integrated properly.
Our second extension area sketches open questions with respect to
high-performance computing architectures.
Finally, we pick up the problem of large $p$ again.

Compared to previous work \cite{Cools:14:LevelDependent}, we have used a uniform
complex grid rotation among all grid levels---well-aware that a level-dependent complex rotation might pay
off. The same can be explored for the complex shifted $\phi$ approach from \cite{Erlangga:04:CSL}.
We reiterate that, for refined vertices, $\phi$ might hold both the sampled
value plus an additional shift term and that, hence, no modification of any
computation is required once $\phi $ is set---even for level-dependent
complex shifts---since we determine $diag(H_\ell )^{-1}$ on-the-fly and do not rely on fixed diagonal values.
Level-dependent shifts or grid rotations can either be determined a priori or whenever a simple
operator analysis yields the insight that the operator enters a problematic
regime (diagonal element underruns certain threshold or switches sign, e.g.).
The latter property is particularly interesting for non-uniform $\phi $
distributions.
Further next steps in the multigrid context comprise the realisation of
operator-dependent grid transfer operators
\cite{Weinzierl:13:Hybrid,Yavneh:12:Nonsymmetric} and more suitable smoothers. 
Problem-dependent operators also might pay off along the boundary layers that
are poorly handled by the current geometric operators.
In this context, we emphasise that our algorithm relies on the Galerkin multigrid property in
\eqref{equation:multigrid:HTMG} while it so far realises rediscretisation.
This introduces an error on the coarse grid.
To the best of our knowledge, no analysis of such an error exists.
Though one has to assume that it is bounded and small, explicit operator
evaluation overcomes this problem completely.
For future work tackling the smoother challenge, we refer in particular to
patch-based approaches \cite{Ghysels:13:HideLatency,Ghysels:14:HideLatency}.
Proof-of-concept studies from other application areas exist that use the
same software infrastructure \cite{Weinzierl:14:BlockFusion} to embed
small regular Cartesian grids into each spacetree cell.
These small grids, patches, allow for improved robustness due to stronger
smoothers resulting from Chebyshev iterations, higher-order smoothing schemes on embedded regular grids or a
multilevel Krylov solver based on recursive coarse grid deflation
\cite{Sheikh:13:deflation,Erlangga:08:deflation}.
Deflation is a particular interesting feature for highly
heterogeneous Helmholtz problems where bound states emerge as isolated eigenvalues near the origin. 
These states are of special interest as they correspond to resonances in the
system
\cite{Aguilar:71:boundstates1,Balslev:71:boundstates2,Moiseyev:98:boundstates4,Simon:79:boundstates3}.
Their proper treatment seems to be implementationally straightforward within the
present code idioms.

Patch-based approaches allow for efficient smoothers, but they are also promising
with respect to multi- and manycores as well as MPI parallelisation.
These two levels of parallelisation are a second track to follow.
While our channels exhibit some perfect concurrency, the application of a
proper domain decomposition on the long term will become mandatory due 
to the explosion of cores or high memory requirements.
Several papers state that additive multigrid algorithms---though inferior to
multiplicative variants in terms of convergence---are well-suited to parallel
architectures due to their higher level of concurrency in-between the levels
(notably \cite{Chow:04:Survey} and references therein or
\cite{Vassilevski:14:CommunicationReductionInMG}).
For the present codes, these statements are problematic.
Our additive solvers exhibit a tight inter-grid data exchange and benefit from
vertical integration.
It is doubtful whether it is advantageous to deploy different grid resolution
solves to different cores.
Yet, the topic deserves further studies, notably if multiple sweeps are
used to solve the individual levels' problems.
Furthermore, the exact interplay of concurrent channel solves, time-in-between inter-grid
data transfer and shared and distributed memory
parallelisation remains to be investigated.
Segmental refinement \cite{Adams:15:Segmental} in contrast seems to
be an obviously promising technique.
Starting from a reasonably fine grid, the underlying tree here is split up into
independent subtrees deployed to different cores. 
This approach should integrate with the present ideas where we ensure high core
efficiency through our vertical integration of various levels, while
segmentation yields parallelism through horizontal decoupling of the grid
levels.

To the best of our knowledge, $p \in \{3,4\}$ in our application context already
is a step forward compared to many state-of-the-art simulation runs.
On the long term, however, bigger $p$ have to be mastered.
While stronger smoothers and better grid transfer operators might
be able to deliver codes that do not suffer from the reduced coarse correction
impact, the explosion of unknowns for $p\geq 5$ remains a challenge.
Alternative techniques such as sparse grids \cite{Bungartz:04:SparseGrids}
or sampling-based methods then
might become a method of choice.

  \begin{acks}
Bram Reps is supported by the FP7/2007-2013 programme under grant agreement No 610741 (EXA2CT).
Tobias Weinzierl appreciates the support from Intel through Durham's Intel
Parallel Computing Centre (IPCC) which made it possible to access the latest
Intel software.
The AMR software base \cite{Software:Peano} is supported and benefits 
from funding received from the European Union’s Horizon 2020
research and innovation programme under grant agreement No 671698 (ExaHyPE).
This work is dedicated to Christoph Zenger.
He has first introduced the additive scheme in the context of Peano spacetree
codes and he has acted as \linebreak
(co-)advisor to 
\cite{Griebel:90:HTMM,Griebel:94:Multilevel,Weinzierl:2009:Diss} which 
are starting points of the present work.
All software is freely available from \cite{Software:Peano}.
\end{acks}

  \appendix
  \section*{APPENDIX}
  \setcounter{section}{0}
  
  \section{Remarks on traditional additive multigrid}

The top-down Algorithm~\ref{algo::tdadd} can be derived from the classic 
additive variant (Algorithm~\ref{algo::add}) in small steps. 
For this, we retain the bottom-up formulation (Algorithm~\ref{algo::buadd})
but, different to the plain correction scheme from Algorithm~\ref{algo::add}, 
already realise FAS.

\begin{algorithm}[htb]
 \SetAlgoNoLine
    \begin{algorithmic}[1]
       \Function{buFAS}{$\ell $}
            \State $du_{\ell} \leftarrow \omega _{\ell} S(u_\ell,b_\ell)$
           \Comment{Bookmark update due to a Jacobi step.}
         \State
         \eIf{$\ell = \ell _{min}$}{
            \State \phantom{xx} $u_{\ell} \gets u_\ell + du_\ell$
         }
         {
           \State \phantom{xx} $u_{\ell -1} \leftarrow I u_\ell$
          \Comment Inject into next coarser level
        \State \phantom{xx} \Comment to be able to compute the
          $\hat u_{\ell}$.
        \State \phantom{xx} $\hat u \leftarrow u_\ell - Pu_{\ell-1}$
          \Comment Determine hierarchical surplus.
        \State \phantom{xx} $\hat r_\ell \leftarrow b_\ell - H_\ell \hat u$
        \State \phantom{xx} $b_{\ell -1} \leftarrow R \hat r_\ell$
           \Comment{Determine coarse grid rhs.}
        \State \phantom{xx} \Call{buFAS}{$\ell -1$} \label{algo::buadd::recur}
           \Comment{Go to coarser level.}
        \State \phantom{xx} $u_\ell \leftarrow u_\ell + du_{\ell} +
        P\left(u_{\ell-1} - Iu_\ell \right)$
        \label{algo::buadd:prolongation}
        }
      \EndFunction
  \end{algorithmic}
  \caption{
    Additive FAS scheme running through the grid bottom-up, i.e.~from fine
    grid to coarse level.
    Invoked by \textsc{buFAS}($\ell _{max}$).
    \label{algo::buadd}
  }
\end{algorithm}


\noindent
This variant introduces a helper variable per level $du_\ell$ holding the impact
of the smoother $S$, i.e.~its correction of the solution on level $\ell$.
This helper variable translates into $sc$ and $sf$ for the final top-down
algorithm.
In the final line \ref{algo::buadd:prolongation}, we update the
unknown on each level with a correction due to the smoother held in $du_\ell$, and we also add
the coarse grid contribution.
To obtain a BPX-type solver, we have to rewrite this line into 
\[
  u_\ell
  \leftarrow u_\ell + du_{\ell} + P\left(u_{\ell-1} - I(u_\ell +
du_{\ell})\right).
\]

\noindent
Both the coarse grid update and the smoother update are independent of each other and both rely
on the injected solution from the previous traversal. 
Due to the former property, an additive scheme is realised.
Due to the latter property, we may not only project some coarse data
$u_{\ell-1}$, but we have to reduce this coarse contribution by the value injected from the current
level $Iu_\ell$.
For the pure additive multigrid, this is one major difference to the similar scheme proposed in
\cite{Mehl:06:MG} and references therein.


\begin{lemma}
  The prolongation of $u_{\ell-1} - I(u_\ell + du_{\ell})$ ensures that only a
  coarse correction is prolongated though the coarse grid holds a injected fine
  grid representation (FAS).
\end{lemma}

\begin{proof}
In what follows we will introduce the variable $corr_\ell = du_{\ell} - PIdu_{\ell} + Pcorr_{\ell-1}$ to label the update to the previous solution after one cycle. It consists of the net contribution $du_{\ell}- PIdu_{\ell}$ from the last smoothing step on the current level, plus the total prolongated correction from the coarse grid $Pcorr_{\ell-1}$. By definition, the new solution on level $\ell$ equals,
\begin{eqnarray*}
u_\ell  &=& u^{prev}_\ell + du_{\ell} + P\left(u_{\ell-1}   - I(u^{prev}_\ell   + du_{\ell})\right)\\
        &=& u^{prev}_\ell + du_{\ell} + Pu_{\ell-1}         - PIu^{prev}_\ell   - PIdu_{\ell} \\
    &=& u^{prev}_\ell + du_{\ell} + Pu_{\ell-1} - Pu^{prev}_{\ell-1}    - PIdu_{\ell} \\
    &=& u^{prev}_\ell + du_{\ell} - PIdu_{\ell} + P(u_{\ell-1} - u^{prev}_{\ell-1}) \\
    &=& u^{prev}_\ell + du_{\ell} - PIdu_{\ell} + Pcorr_{\ell-1} \\
    &=& u^{prev}_\ell + corr_\ell.
\end{eqnarray*}
We now show that the injection property holds, indeed for the finer level we have,
\begin{eqnarray*}
Iu_{\ell+1}     &=& I(u^{prev}_{\ell+1} + corr_{\ell+1}) \\
        &=& Iu^{prev}_{\ell+1} + I du_{\ell+1} - IPIdu_{\ell+1} + IPcorr_{\ell} \\
        &=& u^{prev}_{\ell} + I du_{\ell+1} - Idu_{\ell+1} + IPcorr_{\ell} \\
        &=& u^{prev}_{\ell} + corr_{\ell},
\end{eqnarray*}
and so we can update the solution on level $\ell$ before the update on the finer level $\ell+1$, and still conserve the injection property that is needed for the FAS scheme. 
\end{proof}

\noindent
The lemma's prolongation can be postponed to the subsequent tree traversal.
It then acts as prelude there.
This allows for the transform to a top-down algorithm and
Algorithm~\ref{algo::buadd} the logical starting point for the formulation of a
top-down algorithm.
We start the cycle with the update of the solution and
the lines preceding the recursive call on line~\ref{algo::buadd::recur} are computed in a fine to coarse order. 
To facilitate this `permutation', we introduced an additional variable
$t_{\ell-1}=Idu_{\ell}$ that carries over the injected smoothing update from the
previous cycle.
Such a technique is a common pattern in pipelining.

  \section{Remarks on element-wise matrix-free mat-vecs}
\label{section:appendix:element-wise}

\begin{figure}[h]
  \begin{center} 
  \includegraphics[width=0.8\textwidth]{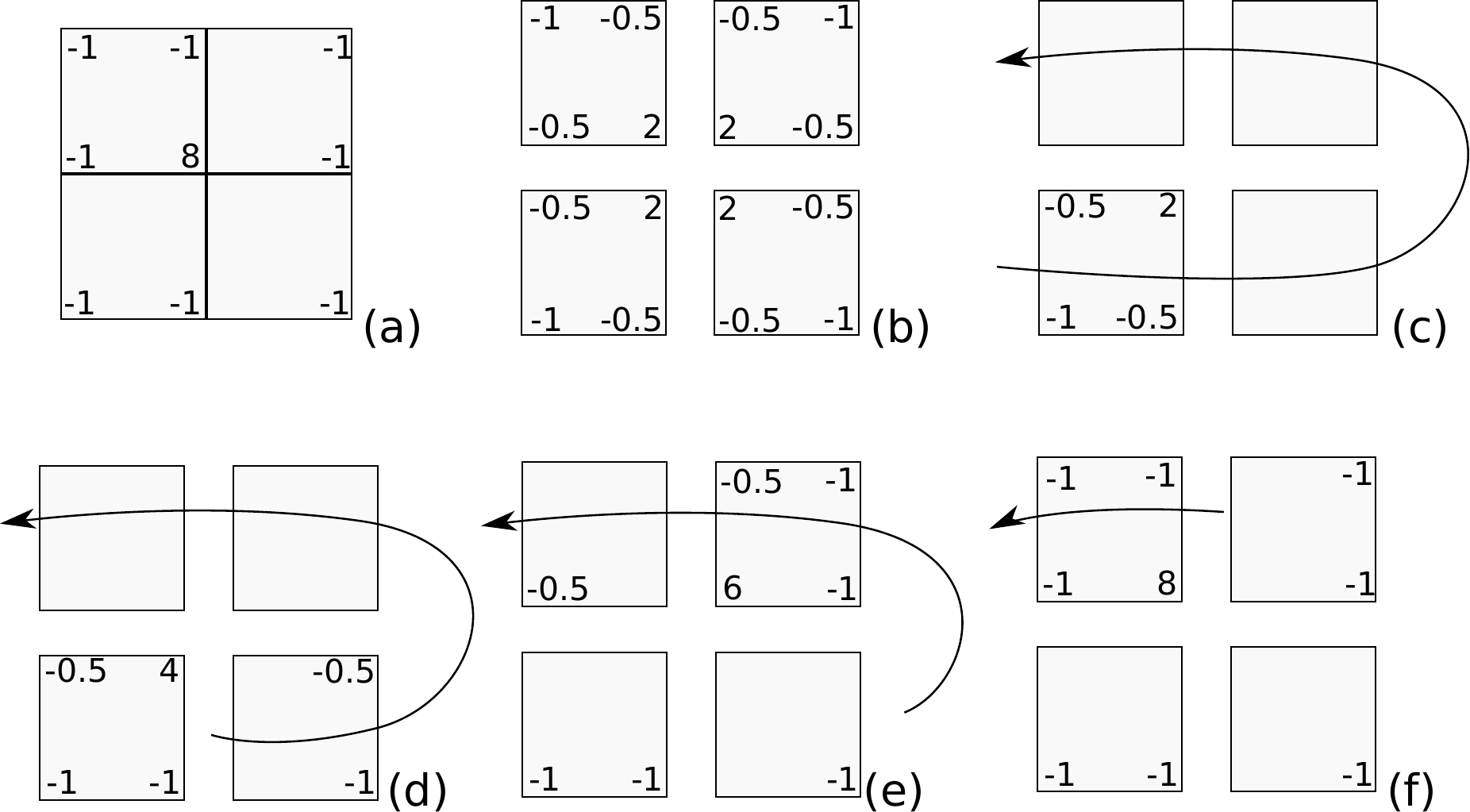}
  \end{center}
  \vspace{-0.2cm}
  \caption{
    Cartoon of element-wise matrix-free matrix-vector evaluations.
    \texttt{(a)} shows an exemplary stencil, i.e.~one line of the matrix-vector
    product.
    \texttt{(b)} decomposes the stencil among the elements.
    These components are used within the cells, \texttt{(c)} through \texttt{(f)}, that illustrate how the original stencil thus is successively reassembled.
    However, all data accesses are strictly element local.
    \label{figure:element-wise-matrix-free}
  }
\end{figure}

The element-wise matrix-free evaluation of matrix-vector products (matvec) is a
state-of-the-art technique in scientific computing.
To simplify reproducibility and for matter of completeness, it is reiterated here.

We split up the stencils among the affected elements, i.e.~rewrite them into an
element-wise representation (Figure \ref{figure:element-wise-matrix-free},
\texttt{(a)} into \texttt{(b)}).
For each vertex, we store a tuple with the current solution $u$ as well as two
helper variables $r$ and $diag$.

The latter are set to zero before the traversal runs through any adjacent cell
of the respective vertex, i.e.~prior to Figure \ref{figure:element-wise-matrix-free},
\texttt{(c)}.
When we enter a cell, we read the $p^2$ adjacent $u$ values and apply the local
assembly matrix to these values. 
The result is added to the temporary variable $r$.
Analogously, we accumulate the diagonal value $diag$.
As we run through the grid, the whole stencil, i.e.~assembly matrix line, is
successively accumulated within $r$ and $diag$ (Figure \ref{figure:element-wise-matrix-free},
\texttt{(c)} through \texttt{(f)}).
Once all adjacent elements of a particular vertex have been run through, Figure
\ref{figure:element-wise-matrix-free} \texttt{(f)}, $r$ and $diag$ hold the
matvec evaluation result associated to the vertex in $r$ and the diagonal
element value within $diag$.
We add the right-hand side to $r$ and update the $u$ value using both $r$ and
$diag$.

 \bibliographystyle{ACM-Reference-Format-Journals}
 \bibliography{paper}
 
 \received{t.b.d.}{t.b.d.}{t.b.d.}

\end{document}